\documentclass{article}

\usepackage{fullpage}
\usepackage{parskip}
\usepackage{amsmath, amsthm, amsfonts, mathdots, thm-restate, mathrsfs, mathtools}
\usepackage{graphicx, tablefootnote, hhline, wrapfig, float, etoolbox}
\usepackage{booktabs, subcaption, multicol, multirow, vwcol, xspace, soul}
\usepackage[title]{appendix}
\usepackage[T1]{fontenc}
\usepackage[dvipsnames]{xcolor}
\usepackage[numbers]{natbib}
\usepackage{hyperref}
\hypersetup{
    colorlinks=true,
    linkcolor=blue,
    urlcolor=blue,
    citecolor=blue,
}

\AfterEndEnvironment{wrapfigure}{\setlength{\intextsep}{0pt}}
\usepackage{pgfplots,tikz}
\pgfplotsset{compat=1.15}

\usepackage[ruled,linesnumbered,noend]{algorithm2e}

\SetAlFnt{\small}
\SetAlCapFnt{\small}
\SetAlCapNameFnt{\small}
\SetAlCapHSkip{0pt}
\IncMargin{-\parindent}
\SetKwInput{KwData}{Input}
\SetKwInput{KwResult}{Output}

\newtheorem{theorem}{Theorem}

\newtheorem{lemma}{Lemma}
\newtheorem{corollary}{Corollary}
\newtheorem{observation}{Observation}
\newtheorem{claim}{Claim}

\theoremstyle{plain}
\newtheorem{definition}{Definition}

\allowdisplaybreaks

\newcommand{\OPT}{\mathrm{OPT}}
\newcommand{\xOPT}{x^{\OPT}}

\newcommand{\poly}{\mathrm{poly}}
\newcommand{\polylog}{\mathrm{polylog}}

\newcommand\mc{\mathcal}

\newcommand{\conv}{\mathrm{conv}}

\newcommand{\R}{\mathbb{R}}

\newcommand{\Z}{\mathbb{Z}}

\newcommand{\RO}{\R_{\ge 0}}

\newcommand{\ZO}{\Z_{\ge 0}}

\newcommand{\ov}{\overline}

\newcommand{\domain}{\mathcal{D}}   
\newcommand{\class}{\mathbf{C}}     
\newcommand{\da}{\downarrow}        
\newcommand{\Perm}{\mathrm{Perm}}

\newcommand{\simplex}[1]{\mathbf{\Delta}_{#1}}
\newcommand{\topk}{\mathbf{Top}}
\newcommand{\ordered}{\mathbf{Ord}}

\newcommand{\smn}{\mathbf{Sym}}
\newcommand{\prob}[1]{\textup{\textsc{#1}}\xspace}

\newcommand{\clustering}{$k$-\prob{Clustering}}
\newcommand{\ufl}{\prob{Uncapacitated-Facility-Location}}
\newcommand{\completiontimes}{\prob{Completion-Times}}

\newcommand{\machineloadsidenticaljobs}{\prob{Machine-Loads-Identical-Jobs}}
\newcommand{\mlij}{\prob{MLIJ}}
\newcommand{\osc}{\prob{Ordered-Set-Cover}}
\newcommand{\ovc}{\prob{Ordered-Vertex-Cover}}
\newcommand{\coveringpolyhedra}{\prob{Covering-Polyhedron}}

\newcommand{\orderedtsp}{\prob{Ordered-TSP}}
\newcommand{\composable}{\prob{Composable}}
\newcommand{\ordsat}{\prob{Ordered-Satisfaction}}

\newcommand{\orderandcount}{\texttt{OrderAndCount}\xspace}
\newcommand{\iterativeordering}{\texttt{IterativeOrdering}\xspace}
\newcommand{\iterativeclustering}{\texttt{IterativeClustering}\xspace}

\usepackage{authblk}

\begin{document}

\title{Balancing Notions of Equity: Trade-offs Between Fair Portfolio Sizes and Achievable Guarantees}

\date{}
\author[1]{Swati Gupta}
\author[2]{Jai Moondra}
\author[2]{Mohit Singh}

\affil[1]{\small Massachusetts Institute of Technology \protect\\ {\small \tt swatig@mit.edu}}
\affil[2]{\small Georgia Institute of Technology \protect \\
{\small \tt jmoondra3@gatech.edu, mohit.singh@isye.gatech.edu}}

\maketitle

\begin{abstract}
    Motivated by fairness concerns, we study the `portfolio problem': given an optimization problem with set $\domain$ of feasible solutions, a class $\class$ of fairness objective functions on $\domain$, and an approximation factor $\alpha \ge 1$, a set $X \subseteq \domain$ of feasible solutions is an \emph{$\alpha$-approximate portfolio} if for each objective $f \in \class$, there is an $\alpha$-approximation for $f$ in $X$. Choosing the classes of top-$k$ norms, ordered norms, and symmetric monotonic norms as our equity objectives, we study the trade-off between the size $|X|$ of the portfolio and its approximation factor $\alpha$ for various combinatorial problems. For the problem of scheduling identical jobs on unidentical machines, we characterize this trade-off for ordered norms and give an exponential improvement in size for symmetric monotonic norms over the general upper bound. We generalize this result as the \orderandcount framework that obtains an exponential improvement in portfolio sizes for covering polyhedra with a constant number of constraints. Our framework is based on a novel primal-dual counting technique that may be of independent interest.
    We also introduce a general \iterativeordering framework for simultaneous approximations or portfolios of size $1$ for symmetric monotonic norms, which generalizes and extends existing results for problems such as scheduling, $k$-clustering, set cover, and routing.
\end{abstract}


\paragraph*{Acknowledgements.} This work was supported by NSF Grants CCF-2106444, CCF-1910423, 2112533, NSF CAREER Grant 2239824, and the Georgia Tech ARC-ACO Fellowship.

\section{Introduction}\label{sec: introduction}

With rapid adoption and proliferation of data-driven decisions, widespread inequalities exist in our society in various forms, often perpetuated by optimized decisions to problems in practice. For example, the existence of food deserts is well-documented across the world \cite{usda_usda_2023, battersby_africas_2014, cummins_food_2002, gartin_food_2012}. The US Department of Agriculture \cite{usda_usda_2023} defines a food desert as a low-income census tract where families below the poverty line do not have a large\footnote{``Large'' is defined as a store with at least $\$ 2$ million annual profit and containing all traditional food departments.} grocery chain within 1 mile of their location in urban areas or 10 miles in rural areas. Gupta et al. \cite{gupta_which_2023} similarly show that medical deserts -- regions with a significant fraction of the population below the poverty line, but far off from the nearest medical facility --  disproportionately affect racial minorities in the US. The decisions to open such facilities are driven by demand, and therefore, optimized decisions tend to overlook sparsely populated regions with vulnerable populations.

As another example, over the last decade, many retailers have adopted scheduling
optimization systems \cite{bernhardt2021data}. These systems draw on a variety of data to predict customer demand and make decisions about the most efficient workforce schedule. Some systems, e.g. Percolata, estimate sales productivity scores for each worker and create schedules based
on these scores. Concerns about fairness of workload again arise, as such optimizations result in highly variable,
unpredictable, and discordant schedules for workers. Further, there is evidence of workload inequity in many work environments, including academia \cite{o2021equity}, last-mile delivery drivers \cite{lyu2022towards}, and hospital workers \cite{rooddehghan2015nurses}.

In such applications, the decision is often to maximize the efficiency in the system, however, this results in unequal costs borne by various groups of people. A large number of fairness notions have been proposed in the literature that attempt at ``balancing'' such costs across groups or individuals, such as minimizing some norm of the distances traveled by groups of people \cite{chakrabarty_approximation_2019, chlamtac_approximating_2022, gupta_which_2023, patton_submodular_2023}, finding simultaneous solutions \cite{kumar_fairness_2000, goel_simultaneous_2006, golovin_all-norms_2008}, balancing statistical outcomes in machine learning \cite{chouldechova_fair_2017, feldman_certifying_2015, hardt_equality_2016}, and balancing allocations in social welfare problems \cite{conitzer_group_2019}.
However, even these notions of fairness can be fundamentally incompatible in the sense that a single solution may not be fair with respect to two or more notions of fairness \cite{kleinberg_inherent_2018, gupta_which_2023}. One workaround is to understand the possibilities offered by a (small) set of solutions, called {\it portfolios}, so that there is some representative solution achieving approximate fairness for any single notion of fairness \cite{gupta_which_2023}. Motivated by the practice of selecting organ transplantation policies, \cite{gupta_which_2023} define the \emph{portfolio} problem as follows: {\it given an optimization problem with a set or domain of feasible solutions $\domain$, a class $\class$ of objective functions that represent various equity notions, an approximation factor $\alpha$, and size $s$, find a {\it portfolio} $X \subseteq \domain$ of $s$ solutions, so that for any objective $f \in \class$ there exists a solution $x\in X$ that $\alpha$-approximates $\min_{x\in \mathcal{D}} f(x)$.} X is called an $\alpha$-approximate portfolio. The case $s = 1$ captures simultaneous approximations \cite{kumar_fairness_2000, goel_simultaneous_2006, chandra_worst-case_1975, azar_all-norm_2004, golovin_all-norms_2008}.

For various combinatorial problems and different classes of objectives, it is not clear what the minimum size of an $\alpha$-approximate portfolio needed to achieve a given approximation factor is. Larger portfolios are needed for better approximations, and the goal is to keep size $s$ of the portfolio small. Further, as the set $\class$ of equity objectives grows larger, small portfolios may not even exist. We study portfolios for various combinatorial problems where feasible solutions induce a vector of loads or costs on individuals, such as scheduling, covering, facility location, and routing problems. For the class of equity objectives, we study
\begin{enumerate}
    \item \emph{Top-$k$ norms} \cite{goel_simultaneous_2006, chakrabarty_interpolating_2018}, where the top-$k$ norm of a vector $x \in \R^d$ is the sum of the $k$ highest coordinates of $x$ by absolute value. Top-$k$ norms generalize the $L_1$ and $L_\infty$ norms.
    \item \emph{Ordered norms} \cite{chakrabarty_approximation_2019, patton_submodular_2023}, where given a non-zero \emph{weight vector} $w \in \R_{\ge 0}^d$ with decreasing weights $w_1 \ge \dots \ge w_d \ge 0$, the ordered norm of $x \in \R_{\ge 0}^d$ is the weighted sum of coordinates of $x$ with the $k$th highest coordinate of $x$ weighted by the $k$th highest weight $w_k$. Ordered norms generalize top-$k$ norms and have a natural fairness interpretation of minimizing the cost of the most mistreated individuals when $x$ is a vector of individual costs.
    \item \emph{Symmetric monotonic norms} \cite{kumar_fairness_2000, goel_simultaneous_2006, golovin_all-norms_2008, chakrabarty_approximation_2019}, which are norms that are (i) invariant to the permutation of coordinates and (ii) nondecreasing in each coordinate. $L_p$ norms, top-$k$ norms, and ordered norms are all symmetric monotonic norms.\footnote{Ordered norms are fundamental to symmetric monotonic norms in two aspects: each symmetric monotonic norm (1) is $O(\log d)$-approximated by some ordered norm \cite{patton_submodular_2023}, and (2) is the supremum of some set of ordered norms \cite{chakrabarty_approximation_2019}.}
\end{enumerate}

In this work, we partially answer the question:
\begin{center}
{\it ``What is the trade-off between achievable portfolio size and corresponding approximation factors for various combinatorial optimization problems? Is there a general recipe for constructing small portfolios for ordered and symmetric monotonic norms?''}
\end{center}

In particular, we focus on three general combinatorial problems: scheduling, covering, and facility location, motivated by workplace scheduling and access to critical facilities.
While much effort has gone into determining the best possible simultaneous approximations (portfolio of size $1$), little is known about the construction of portfolios of size greater than 1. For top-$k$ norms, Goel and Meyerson \cite{goel_simultaneous_2006} essentially obtain a $(1 + \epsilon)$-approximate portfolio of size $O\left(\frac{\log d}{\epsilon}\right)$; the same bound also holds for $L_p$ norms \cite{golovin_all-norms_2008, gupta_which_2023}.

However, for ordered norms, only a general construction of $\poly(d^{1/\epsilon})$-sized $(1 + \epsilon)$-approximate portfolios was known before this work, due to Chakrabarty and Swamy \cite{chakrabarty_approximation_2019}, while no bound was known for symmetric monotonic norms. We observe that their result generalizes to symmetric monotonic norms (Lemma \ref{obs: polynomial-portfolios-for-smn}). It was also known that a solution that is simultaneously $\alpha$-approximate for all top-$k$ norms is, in fact, simultaneously $\alpha$-approximate for all symmetric monotonic norms \cite{goel_simultaneous_2006}.
    {\it This property is no longer true for portfolios of size greater than 1} (e.g., see our Example 1, Theorem \ref{thm: portfolios-for-top-k-norms-are-not-portfolios-for-ordered-norms}, or Theorem \ref{thm: portfolio-size-lower-bound-smn}). In particular, we show that the approximation guarantee of a portfolio for top-$k$ norms and ordered norms can differ by a factor polynomial in $d$. Consequently, we cannot restrict to constructing portfolios only for top-$k$ norms and need new techniques for the much larger sets of ordered norms and symmetric monotonic norms. We show that there exist sets $\domain \subseteq \R^d$ for which the portfolio size must be $d^{\Omega(1/\log\log d)}$ (i.e., nearly polynomial in $d$) for ordered and symmetric monotonic norms even for approximation as large as $O(\log d)$ (see Theorem \ref{thm: portfolio-size-lower-bound-smn}).

\begin{figure}[t]
    \begin{minipage}{0.54\textwidth}
        \centering
        \includegraphics[width=0.8\linewidth]{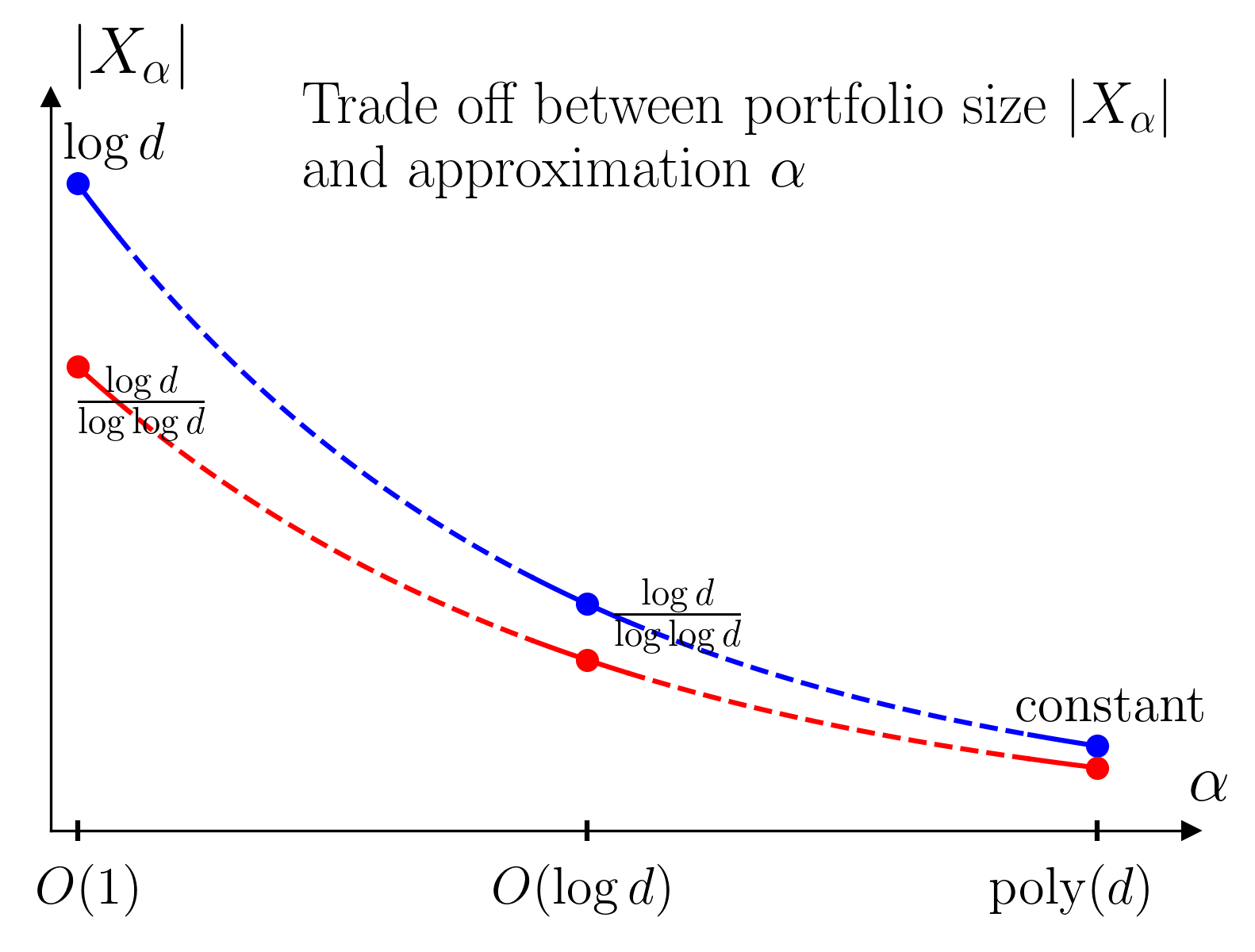}
    \end{minipage}
    \hfill
    \begin{minipage}{0.45\textwidth}
        \caption{A qualitative plot to illustrate the trade-off between approximation $\alpha$ and the smallest portfolio size $|X_\alpha|$ for the \textsc{Machine-Loads-Identical-Jobs} problem for ordered norms. The worst-case lower bound $|X_\alpha| = \Omega\left(\frac{\log d}{\log \alpha + \log\log d}\right)$ is illustrated in red, and the {upper bound} $|X_\alpha| = O\left(\frac{\log d}{\log(\alpha/4)}\right)$ is illustrated in blue. The two bounds converge for $\alpha = \Omega(\log d)$.}
        \label{fig: trade-off-chart}
    \end{minipage}
\end{figure}

\subsection{Our contributions}

To obtain smaller portfolios for covering and scheduling problems, we develop a general framework called \orderandcount. We obtain size-$\polylog(d)$ portfolios using this framework, an \emph{exponential} improvement over the general size bound (see Table \ref{tab: portfolios-size-more-than-1}). In addition, we consider simultaneous approximations where we unify and generalize previously known algorithms as the \iterativeordering framework. Table \ref{tab: portfolios-size-1} summarizes our new results obtained using this framework. We detail these contributions next.

\subsubsection{Characterizing trade-off for \machineloadsidenticaljobs} As our first result, we consider the {\sc Machine-Loads-Identical-Jobs} (\mlij) problem where $n$ identical jobs must be scheduled on $d$ unidentical machines to minimize some norm of the vector of machine loads. This is a simple model for workload distribution among $d$ workers with different processing speeds, and various norms correspond to various fairness criteria for fair distribution of jobs. We prove the following result for this setting:

\begin{restatable}{theorem}{identicaljobsscheduling}\label{thm: identical-jobs-portfolio}
There is a polynomial-time algorithm that given any instance of the \textsc{\textup{Machine-Loads-Identical-Jobs}} \textup{(\mlij)} problem with $d$ machines and any $\alpha > 4$, finds a portfolio $X$ of size
\[
    |X| = O\left(\frac{\log d}{\log (\alpha/4)}\right)
\]
that is (i) $\alpha$-approximate for ordered norms and (ii) $O(\alpha \log d)$-approximate for symmetric monotonic norms. Further, there exists a family of instances of  \mlij for which the size of any $\alpha$-approximate portfolio for ordered norms is lower bounded by $\Omega\left(\frac{\log d}{\log \alpha + \log\log d}\right)$. This characterizes the size-approximation trade-off for $\alpha>4$.
\end{restatable}

\begin{table}[t]
    \centering
    \caption{Approximations for size $>1$ portfolios for ordered norms and symmetric monotonic norms, for arbitrary $\epsilon \in (0, 1]$. Previously, only a $\poly(d^{1/\epsilon})$-sized portfolio was known \cite{chakrabarty_approximation_2019} for $(1 + \epsilon)$-approximation for
    ordered norms, for dimension $d$ problems.}
    \label{tab: portfolios-size-more-than-1}
    \footnotesize
    \hspace*{-1em}
    \begin{tabular}{ccccc}
        \hline
        \multirow{4}{*}{\begin{tabular}[c]{@{}c@{}} Problem or set of \\
        feasible vectors $\domain$ \end{tabular}} &
        \multirow{4}{*}{\begin{tabular}[c]{@{}c@{}} Worst-case \\ approximation factor \\ for simultaneous\\ approximation\end{tabular}} &
        \multicolumn{3}{c}{Guarantees for portfolio of size $> 1$} \\ \cmidrule{3-5} & &
        \multicolumn{1}{c}{Size} &
        \multicolumn{1}{c}{\begin{tabular}[c]{@{}c@{}}Approximation\\ for ordered norms\end{tabular}} &
        \begin{tabular}[c]{@{}c@{}}Approximation\\ for symmetric \\ monotonic norms\end{tabular} \\ \hline
        \begin{tabular}[c]{@{}c@{}} \textsc{Machine-Loads-Identical-Jobs}\\ $d$ machines \\ (Theorem \ref{thm: identical-jobs-portfolio}) \end{tabular} &
        $\Omega(\sqrt{d})$ &
        \multicolumn{1}{c}{{$O\left(\frac{\log d}{\epsilon}\right)$}} &
        \multicolumn{1}{c}{$4 + \epsilon$} &
        $O(\log d)$ \\ \hline
        \begin{tabular}[c]{@{}c@{}}\coveringpolyhedra \\
        with $r$ constraints:\\ $\{x \in \R_{\ge 0}^d: Ax \ge b \}$,\\ $A \in \R_{\ge 0}^{r \times d}, b \in \R_{\ge 0}^r$ \\ (Theorem \ref{thm: covering-polyhedron})\end{tabular} &
        $\Omega(\sqrt{d})$ &
        \multicolumn{1}{c}{$\left(\frac{\log (d/\epsilon)}{\epsilon}\right)^{O(r^2)}$} &
        \multicolumn{1}{c}{$1 + \epsilon$} &
        $O(\log d)$ \\ \hline
    \end{tabular}
\end{table}

Note that the above result completely characterizes the trade-off between portfolio sizes and achievable approximation factors (up to $\log \log$ factor) for the {\sc Machine-Loads-Identical-Jobs} problem (See Figure \ref{fig: trade-off-chart}). To obtain this result, we use our \texttt{OrderAndCount} approach, which exploits the fact that each ordered norm, while convex in general, is a \emph{linear function} when restricted to a region where all vectors satisfy the same order of coordinate values. That is, if vector $x \in \R^d$ satisfies $x_{\pi(1)} \geq x_{\pi(2)} \ge \hdots \ge x_{\pi(d)} \ge 0$ for some order $\pi$ on $[d]$, the ordered norm $\|x\|_{(w)}$ is the linear function $\sum_k w_k x_{\pi(k)}$. This gives the following algorithm to obtain portfolios for ordered norms: for each order $\pi$, we can restrict to the set $\domain_{\pi}$ of vectors in $\domain$ that satisfy order $\pi$, and collect the set of extreme points of $\domain_{\pi}$. This in general results in exponentially many solutions (there are exponentially many orders $\pi$ and potentially exponentially many extreme points of each $\domain_{\pi}$). We show that for \mlij, (i) it suffices to restrict to a specific order $\pi^*$ (that depends on the problem instance), and that (ii) there are at most $d$ extreme points of $\domain_{\pi^*}$. We give a rounding algorithm to show that these $d$ extreme points can further be $\alpha$-approximated by a subset of $O(\log_{\alpha/4} d)$ \emph{integral} points, which results in the desired portfolio.

\subsubsection{Exponential improvement in portfolios for covering}

Next, we consider the \coveringpolyhedra problem, which simply includes $r$ covering constraints of the form: $a^\top x \geq b$ (for $a \in \mathbb{R}_{\ge 0}^{d}$, and $b \in \mathbb{R}_{\ge 0}$) together with nonnegativity $x \ge 0$. This generalizes the \mlij problem above, and models many natural scenarios for workload distribution.

Many problems can be modeled as the covering polyhedron, for example, a fair centralized server that must balance the workload on $d$ machines, each with $r$ parallel processing units \cite{nikolic_survey_2022}. This load-balancing problem also appears in the context of volunteer-dependent non-profit organizations, such as HIV social care centers, blood donation drives, food recovery organizations \cite{manshadi_online_2020}, etc. Numerous studies have been conducted on the reasons for the attrition of volunteers, and overburdening by the amount of demands placed on them is one of the key ones \cite{locke2003hold, knapp1995volunteers}. This work can help balance workloads in volunteer organizations, to help mitigate its impact on attrition.

Back to the machine load scheduling terminology, if $b_j$ units for the $j$th job type need to be scheduled, and the machine $i \in [d]$ has processing speed $A_{j, i}$ for the $j$th type of job, then the total loads $x_i, i \in [d]$ on the machines must satisfy $\sum_{i \in [d]} A_{j, i} x_i \ge b_j$. For a given norm $\|\cdot\|_f$ or fairness criterion, this translates into minimizing $\|x\|$ over the covering polyhedron $\{x \in \R^d: Ax \ge b, x \ge 0\}$.

The challenge in extending \texttt{OrderAndCount} to such problems is (i) bounding the number of possible orders that the optimal solution $x^*$ might satisfy, and then (ii) selecting a subset of corresponding extreme points for each order that must be included in the portfolio. For the first challenge, we develop a novel {\it primal-dual counting technique} which allows us to count the number of possible orders in an appropriate dual space that is structurally much simpler (Section \ref{sec: covering-polyhedra}). For the second challenge, we show that a sparsification procedure allows us to reduce the number of extreme points for each order. Together, using \texttt{OrderAndCount}, we give poly-logarithmic sized portfolios for \coveringpolyhedra for constant $r$:

\begin{restatable}{theorem}{coveringpolyhedron}\label{thm: covering-polyhedron}
For \coveringpolyhedra in $d$ dimensions and $r$ constraints, for any $\epsilon \in (0, 1]$, there is a portfolio $X$ of size
\[
    |X| = O\big(\log (d/\epsilon)/\epsilon\big)^{3r^2 - 2r},
\]
which is (i) $(1 + \epsilon)$-approximate for ordered norms, and (ii) $O(\log d)$-approximate for symmetric monotonic norms. The running time of the algorithm is polynomial in $d$ and  $(\log(d)/\epsilon)^{r^2}$.
\end{restatable}

The above result shows that the trade-off between $\epsilon$ and $X_{1 + \epsilon}$ is that $|X_{1 + \epsilon}|^{1/\Omega(r^2)} \cdot \epsilon$ remains nearly a constant.
For all $r = o\left(\frac{\sqrt{\log d}}{\log\log d}\right)$, this result is the first {\it exponential improvement} over the current best bound of $\poly(d^{1/\epsilon})$ \cite{chakrabarty_approximation_2019}, to the best of our knowledge.

\begin{table}[t]
    \centering
    \small
    \caption{A summary of simultaneous approximations for symmetric monotonic norms, obtained using the \iterativeordering framework. Here, a bicriteria $(\alpha, \beta)$-approximation to a \clustering problem opens at most $\beta k$ facilities, while being within factor $\alpha$ of the optimum (see Appendix \ref{sec: clustering-and-facility-location}). $\gamma$ is a parameter for composable problems (see Section \ref{sec: iterative-ordering}).}
    \label{tab: portfolios-size-1}
    \begin{tabular}{ccccc}
        \hline
        \multicolumn{2}{c}{\begin{tabular}[c]{@{}c@{}}Domain/set $\domain$\\ of feasible vectors\end{tabular}} &
        \multicolumn{1}{c}{\begin{tabular}[c]{@{}c@{}}Existence\\ simultaneous\\ approximation\end{tabular}} &
        \begin{tabular}[c]{@{}c@{}}Polynomial-time\\ simultaneous\\ approximation\end{tabular} &
        Reference \\ \hline
        \multicolumn{1}{c}{\begin{tabular}[c]{@{}c@{}}$\gamma$-\composable\\ problem\end{tabular}} &
        \textbf{\begin{tabular}[c]{@{}c@{}}This\\ work\end{tabular}} &
        \multicolumn{1}{c}{$(\sqrt{\gamma} + 1)^2$} &
        - &
        Theorem \ref{thm: iterative-ordering} \\ \hline
        \multicolumn{1}{c}{\begin{tabular}[c]{@{}c@{}}\completiontimes\\ $(\gamma = 1)$\end{tabular}} &
        \textbf{\begin{tabular}[c]{@{}c@{}}This\\ work\end{tabular}} &
        \multicolumn{1}{c}{$4$} &
        $8$ &
        Theorem \ref{thm: iterative-ordering} \\ \hline
        \multicolumn{1}{c}{\multirow{5}{*}{\begin{tabular}{c} \orderedtsp\\ $(\gamma = 2)$ \end{tabular}}} &
        \multirow{3}{*}{\begin{tabular}[c]{@{}c@{}}Previous\\ work\end{tabular}} &
        \multicolumn{2}{c}{$16$} &
        \cite{golovin_all-norms_2008} \\ \cmidrule{3-5}
        \multicolumn{1}{c}{} &
        &
        \multicolumn{2}{c}{$8$} &
        \cite{farhadi_traveling_2021} \\ \cmidrule{2-5}
        \multicolumn{1}{c}{} &
        \textbf{\begin{tabular}[c]{@{}c@{}}This\\ work\end{tabular}} &
        \multicolumn{1}{c}{{$3 + 2 \sqrt{2} \simeq 5.83$}} &
            {$6 + 4 \sqrt{2} \simeq 11.66$} &
        Theorem \ref{thm: iterative-ordering} \\ \hline
        \multicolumn{1}{c}{\multirow{2}{*}{\begin{tabular}[c]{@{}c@{}}\osc\\ (on ground set\\ of $n$ elements) \\ ($\gamma = 1$)\end{tabular}}} &
        \begin{tabular}[c]{@{}c@{}}Previous\\ work\end{tabular} &
        \multicolumn{2}{c}{$O(\log n)$} &
        \cite{golovin_all-norms_2008} \\ \cmidrule{2-5}
        \multicolumn{1}{c}{} &
        \textbf{\begin{tabular}[c]{@{}c@{}}This\\ work\end{tabular}} &
        \multicolumn{1}{c}{$4$} &
        - &
        Theorem \ref{thm: iterative-ordering} \\ \hline
        \multicolumn{1}{c}{\multirow{5}{*}{\begin{tabular}[c]{@{}c@{}}\clustering\\ (on $n$ points,\\ bicriteria\\ approximations)\end{tabular}}} &
        \multirow{3}{*}{\begin{tabular}[c]{@{}c@{}}Previous\\ work\end{tabular}} &
        \multicolumn{1}{c}{$\left(3 + \epsilon, O((\log n) + 1/\epsilon)\right)$} &
        $\left(9 + \epsilon, O((\log n) + 1/\epsilon) \right)$ &
        \cite{kumar_fairness_2000} \\ \cmidrule{3-5}
        \multicolumn{1}{c}{} &
        &
        \multicolumn{1}{c}{$\left(1 + \epsilon, O\left(\frac{\log n}{\epsilon}\right)\right)$} &
        $\left(6 + \epsilon, O\left(\frac{\log n}{\epsilon}\right)\right)$ &
        \cite{goel_simultaneous_2006} \\ \cmidrule{2-5}
        \multicolumn{1}{c}{} &
        \textbf{\begin{tabular}[c]{@{}c@{}}This\\ work\end{tabular}} &
        \multicolumn{1}{c}{$\left(1 + \epsilon, O\left(\frac{\log n}{\epsilon}\right)\right)$} &
        $\left(3 + \epsilon, O\left(\frac{\log n}{\epsilon}\right)\right)$ &
        Theorem \ref{thm: clustering} \\ \hline
    \end{tabular}
\end{table}

\subsubsection{Improved approximations using \iterativeordering}

We next turn our attention to portfolios of size-1, i.e., simultaneous approximations \cite{kumar_fairness_2000, goel_simultaneous_2006} for symmetric monotonic norms, and show stronger approximation guarantees for specific problems. We develop an \iterativeordering framework for problems that minimize the symmetric monotonic norm over the vector of times each element is ``satisfied''. For example, in scheduling problems, a job is satisfied when it is completed, in set cover problems, an element is satisfied when it is covered, and in routing problems, a vertex is satisfied when it is visited, etc. We recursively solve the problem, by dividing it into smaller subproblems and stitching the subproblem solution together to get an approximation. The guarantees on satisfaction times are preserved \emph{pointwise}\footnote{That is, if $\tilde{x}$ is the (sorted) approximate vector and $x^*$ is the (sorted) optimal vector, then we show coordinate-wise bounds such as $\tilde{x}_i \leq \alpha x^*_i \;\; \forall \; i \in [d]$.}, leading to simultaneous approximation guarantees for all symmetric monotonic norms. This generalizes the approach of many previous papers, e.g., \cite{blum_minimum_1994, golovin_all-norms_2008, farhadi_traveling_2021} for the traditionally studied notion of polynomial-time computable simultaneous approximations, while also providing novel guarantees on the \emph{existence} of certain simultaneous approximations.
The key improvements we obtain due to the \texttt{Iterative-Ordering} framework for symmetric monotonic norms (summarized in Table \ref{tab: portfolios-size-1}) are:

\begin{itemize}
    \item \textsc{Completion-Times}: For minimizing symmetric monotonic norms of the \emph{completion times} of jobs (e.g., jobs on the cloud computing servers \cite{wang_multi-resource_2015}) in a scheduling problem, we show the existence of simultaneous $4$-approximation and polynomial-time simultaneous $8$-approximation. These are the first constant-factor results for this problem, to the best of our knowledge. Note the contrast with the previously discussed problem of minimizing \emph{machine loads}, where a size-$1$ portfolio may not even be $o(\sqrt{m})$-approximate for all symmetric monotonic norms, even for the case of identical jobs (see Theorem \ref{thm: identical-jobs-portfolio}). For \completiontimes, we also give an instance (see Appendix \ref{sec: lower-bounds}) where no simultaneous $1.13$-approximation exists.

    \item \osc: For minimizing symmetric monotonic norms of covering time of $n$ elements of a ground set, we show the existence of a simultaneous $4$-approximation, although previously only a polynomial time $O(\log n)$-approximation was known \cite{golovin_all-norms_2008}, which up to constants is the best possible assuming P $\neq$ NP \cite{feige_threshold_1998}. This result highlights the difference between existence and polynomial time computable simultaneous approximation.

    \item \orderedtsp: For minimizing symmetric monotonic norms over the time each vertex of a given graph is visited in a Hamiltonian tour, we show the existence of a $5.83$-approximation. The previously-known lower bound on the existence of a simultaneous approximation is $1.78$ \cite{farhadi_traveling_2021}, therefore, ours bridges the gap in the existence of simultaneous approximations for {\sc Ordered-TSP}; although the best polynomial-time approximation remains the $8$-approximation of \cite{farhadi_traveling_2021}.

    \item {\sc $k$-\textsc{Clustering}}: For finding $k$ facilities that minimize symmetric monotonic norms of client distances to open facilities, we give for any $\epsilon \in (0, 1]$ a polynomial-time bicriteria approximation which (a) has objective value within factor $3 + \epsilon$ of the optimal for any symmetric monotonic norm, and (b) opens at most $O\left(\frac{\log n}{\epsilon}\right) \cdot k$ facilities. This improves upon the previous bicriteria approximation of \citet{goel_simultaneous_2006} that has objective value bound $6 + \epsilon$ with the same bound on the number of facilities.
\end{itemize}

The rest of the paper is organized as follows: we give related work in Section \ref{sec: related-work} and preliminaries in Section \ref{sec: preliminaries}. \machineloadsidenticaljobs is discussed in Section \ref{sec: identical-jobs-scheduling}, \coveringpolyhedra is discussed in Section \ref{sec: covering-polyhedra}, and \iterativeordering with corresponding results is discussed in Section \ref{sec: iterative-ordering}. We discuss open problems and conclude in Section \ref{sec: discussion}.

\section{Related work}\label{sec: related-work}

Portfolios were explicitly first studied by Gupta et al. \cite{gupta_which_2023} who studied them for facility location problems.
Similar notions were implicit in other previous works: \cite{goel_simultaneous_2006} essentially constructed $O(\log d)$-size $O(1)$-approximate portfolios for top-$k$ norms in dimension $d$, \cite{golovin_all-norms_2008} used the structure of $L_p$ norms to get a similar bound, and \cite{chakrabarty_approximation_2019} essentially constructed $\poly(d)$-size $O(1)$-approximate portfolios for ordered norms.
All three techniques rely on counting the number of unique norms (up to $O(1)$-approximation). In contrast, our methods rely on counting vectors in the set $\domain$ of feasible vectors. This shift is useful, for example, in obtaining polynomial-size portfolios for symmetric monotonic norms (see Appendix \ref{sec: missing-proofs-1}).

Portfolios of size-$1$ or simultaneous approximations have been very well-studied, with the earliest results going as far back as \cite{chandra_worst-case_1975}, who study the scheduling problem of minimizing loads on identical machines and show that Graham's \cite{graham_bounds_1966}'s greedy algorithm is a $1.5$-approximation for all $L_p$ norms. Azar and Taub \cite{azar_all-norm_2004} improve this to $1.388$-approximation, and for all symmetric monotonic norms. Note that this is in contrast to our results
to minimize machine loads on \emph{unidentical machines} and \emph{identical jobs}, where (see Theorem \ref{thm: identical-jobs-portfolio} and Table \ref{tab: portfolios-size-more-than-1}) a simultaneous $O(1)$-approximation may not exist. Kumar and Kleinberg \cite{kumar_fairness_2000} studied simultaneous approximations for all symmetric monotonic norms for clustering, scheduling, and flow problems. In particular, for \clustering, they obtained a $(9 + \epsilon, O(\log n) + \epsilon^{-1})$-approximation in polynomial time. \cite{goel_simultaneous_2006} improved this to $(6 + \epsilon, O((\log n)/\epsilon)$.

\cite{kumar_fairness_2000, goel_simultaneous_2006, golovin_all-norms_2008} all studied general techniques to obtain simultaneous approximations that often involve (implicitly) obtaining portfolios and combining them into one solution. Goel and Meyerson \cite{goel_simultaneous_2006} proved that a simultaneous $\alpha$-approximation for top-$k$ norms is a simultaneous $\alpha$-approximation for symmetric monotonic norms. Golovin et al. \cite{golovin_all-norms_2008} observed that the basic structure of \cite{blum_minimum_1994}'s algorithm for the Traveling Salesman Problem (\orderedtsp) can be applied to many other problems, obtaining logarithmic or constant-factor approximate simultaneous approximations. Farhadi et. al \cite{farhadi_traveling_2021} further improved the approximation factor for \orderedtsp to $8$. Our \iterativeordering framework generalizes this fundamental idea, combining these algorithms into one algorithm.

Optimizing for a \emph{fixed} non-standard objective has been widely considered in the literature, and the list is too long to fit here.
\cite{chakrabarty_approximation_2019} studied ordered norm and symmetric monotonic norm objectives for scheduling and clustering problems and proved that any symmetric monotonic norm is the supremum of some ordered norms, thus establishing ordered norms as fundamental to the study of symmetric convex functions.
\cite{patton_submodular_2023} proved that any symmetric monotonic norm can be $O(\log d)$-approximated by an ordered norm, further strengthening this connection.

\section{Preliminaries}\label{sec: preliminaries}

We give formal definitions and useful preliminary results in this section. Omitted proofs are included in Appendix \ref{sec: missing-proofs-1}. Throughout, we assume that $\domain \subseteq \mathbb{R}_{\ge 0}^d$ is a set of nonnegative feasible vectors to a combinatorial problem with each coordinate representing the cost to individuals/groups (e.g., distances to open facilities in facility location problems or machine loads in scheduling problems). First, we define portfolios formally:

\begin{definition}[Portfolios]\label{def: portfolios}
Given a set $\domain$ of nonnegative \emph{feasible} vectors, a class of objectives $\class: \domain \to \RO$, and an approximation parameter $\alpha \ge 1$, a portfolio $X \subseteq \domain$ is a set of vectors such that for all objectives $f \in \class$,  \[ \min_{x \in X} f(x) \le \alpha \min_{x \in \domain} f(x).\]
When the portfolio has size $1$, it is called a \emph{simultaneous $\alpha$-approximation} \cite{goel_simultaneous_2006, kumar_fairness_2000}.
\end{definition}

Our first lemma shows that portfolios can be composed in different ways:

\begin{lemma}[Portfolio composition]\label{lem: portfolio-composition}
Given class $\class$ of functions over $\domain \subseteq \R_{\ge 0}^d$
\begin{enumerate}
    \item If $X_1$ is an $\alpha_1$-approximate portfolio for $\class$ over $\domain$ and $X_2$ is an $\alpha_2$-approximate portfolio for $\class$ over $X_1$, then $X_2$ is an $\alpha_1\alpha_2$-approximate portfolio for $\class$ over $\domain$.
    \item If $\domain = \bigcup_{i \in [n]} \domain_i$ and $X_i$ is an $\alpha$-approximate portfolio for $\class$ over $\domain_i$ for each $i \in [n]$, then $\bigcup_{i \in [n]} X_i$ is an $\alpha$-approximate portfolio for $\class$ over $\domain$.
\end{enumerate}
\end{lemma}

For vector $x \in \R^d$, we denote $x^\da$ as the vector with coordinates of $x$ sorted in decreasing order. We also denote $\mathbf{1}_k \in \R^d$ as the vector with $k$ ones followed by zeros.\\

\begin{definition}[Norm classes]\label{def: norms} Given a vector $x \in \R^d$, \begin{enumerate}
                                                                                    \item for $k \in [d]$, the top-$k$ norm of $x$, denoted $\|x\|_{\mathbf{1}_k}$, is the sum $\sum_{i \in [k]} |x|_i^\da$ of $k$ highest coordinates of $|x|$. The class of top-$k$ norms is denoted $\topk$;
                                                                                    \item given a nonzero \emph{weight vector} $w \in \R^d$ such that $w_1 \ge \dots \ge w_d \ge 0$, the ordered norm $\|x\|_{(w)}$ is defined as $w^\top |x|^\da$. The class of ordered norms is denoted $\ordered$;
                                                                                    \item a symmetric monotonic norm is a norm that is monotone in each coordinate and invariant to the permutation of coordinates. The class of symmetric monotonic norms is denoted $\smn$, and an arbitrary norm in $\smn$ is denoted $\|\cdot\|_f$.
\end{enumerate}
\end{definition}

Note that $\topk \subseteq \ordered \subseteq \smn$. For nonnegative $x, y \in \RO^d$, we say that $y$ majorizes $x$ or $x \preceq y$ if $\|x\|_{\mathbf{1}_k} \le \|y\|_{\mathbf{1}_k}$ for all $k \in [d]$. The following lemma shows that majorization implies monotonicity for any symmetric monotonic norm.

\begin{lemma}[\cite{hardy_inequalities_1952}]\label{lem: majorization-norm-order-2}
If $x \preceq y$, then $\|x\|_f \le \|y\|_f$ for any $\|\cdot\|_f \in \smn$.
\end{lemma}

The above lemma helps obtain a simultaneous approximation for $\smn$ using those for $\topk$ norms. Given a set $\domain$, if $x^*$ is simultaneously $\alpha$-approximate for $\topk$ then $\|x^*\|_{\mathbf{1}_k} \le \alpha \|y\|_{\mathbf{1}_k} $ for all $k \in [d]$ and $y \in \domain$, i.e., that $x^* \preceq \alpha y$ for all $y \in \domain$. As an immediate consequence:

\begin{lemma}[\cite{goel_simultaneous_2006}, Theorem 2.3]\label{thm: top-k-implies-all-norm}
For any $\domain$, if $x^*$ is a simultaneous $\alpha$-approximation for $\topk$, then $x^*$ is a simultaneous $\alpha$-approximation for $\smn$.
\end{lemma}

For any $\domain$, it is possible to construct a  $(1 + \epsilon)$-approximate portfolio of size $O((\log d)/\epsilon)$ for $\topk$, by simply considering the minimizers of top-$k$ norms for $k = \lfloor 1 + \epsilon \rfloor, \lfloor (1 + \epsilon)^2 \rfloor, \dots$. One may wonder if this portfolio is also $(1 + \epsilon)$-approximate portfolio for symmetric monotonic norms; however, this is not the case.

\noindent \textbf{Example 1.} Consider the set $\domain = \left\{x, y, z\right\} \in \R^d$ of three feasible vectors $x = (\sqrt{d}, 0, \dots 0), y = (1, \dots, 1)$, and $z = d^{1/3}\left(1, \frac{1}{\sqrt{2}}, \dots, \frac{1}{\sqrt{d}}\right)$. Then, given a top-$k$ norm,
$$
\|x\|_{\mathbf{1}_k} = \sqrt{d}, \quad \|y\|_{\mathbf{1}_k} = k, \quad  \|z\|_{\mathbf{1}_k} = d^{1/3}\sum_{i \in [k]} \frac{1}{\sqrt{i}} = \Theta(d^{1/3} \sqrt{k}).
$$
For each top-$k$ norm, either $x$ or $y$ is optimal, i.e., $\{x, y\}$ is an optimal portfolio for $\topk$. However, consider the ordered norm for weight vector $w = \left(1, \frac{1}{\sqrt{2}}, \dots, \frac{1}{\sqrt{d}}\right)$:
$$
\|x\|_{(w)} = \sqrt{d}, \quad \|y\|_{(w)} = \Theta(\sqrt{d}), \quad \|z\|_{(w)} = d^{1/3} \sum_{i \in [d]} \frac{1}{i} = \Theta(d^{1/3} \log d).
$$
Then both $x$ and $y$ are $\Omega\left(\frac{d^{1/6}}{\log d}\right)$-approximations for $\|\cdot\|_{(w)}$, i.e., $\{x, y\}$ is at best a $\poly(d)$-approximate portfolio for ordered norms. \qed

We can say even more: despite the above Lemma \ref{thm: top-k-implies-all-norm} for $\topk$, the best-known upper bound on the size of a $(1 + \epsilon)$-approximate portfolio for $\smn$ is polynomial in $d^{1/\epsilon}$ (our proof is a slight generalization of \cite{chakrabarty_approximation_2019}; see Appendix \ref{sec: missing-proofs-1}):

\begin{lemma}\label{obs: polynomial-portfolios-for-smn}
For any $\domain$ and $\epsilon \in (0, 1]$, there is a $(1 + \epsilon)$-approximate portfolio of size $\poly(d^{1/\epsilon})$ for $\smn$ over $\domain$.
\end{lemma}

Further, we show in Theorem \ref{thm: portfolio-size-lower-bound-smn} that this bound is nearly tight for ordered norms and symmetric monotonic norms: there exist sets $\mathcal{D} \subseteq \R^d$ where any $\alpha$-approximate portfolios must have size $d^{1/\Omega(\log\log d)}$ even for approximation $\alpha$ as large as $O(\log d)$.

Next, it is known that any symmetric monotonic norm $\|\cdot\|_f$ on $\R^d$ can be $O(\log d)$-approximated by an ordered norm on $\R^d$ \cite{patton_submodular_2023}. Consequently, the same bound also holds for portfolio approximations:

\begin{lemma}\label{lem: portfolios-ordered-norms-to-all-norms}
For any $\domain$, an $\alpha$-approximate portfolio $X \subseteq \domain$ for $\ordered$ is an $O(\alpha \log d)$-approximate portfolio for $\smn$.
\end{lemma}

Finally, we characterize the class of duals to ordered norms and state the corresponding Cauchy-Schwarz inequality, which will be used in our \orderandcount framework. An order $\pi$ on a finite set $X$ is a bijection between $X$ and $\{1, \dots, |X|\}$; for simplicity we denote the set of all orders on $X$ as $\Perm(X)$ or as $\Perm(d)$ when $X = [d]$. We say that a vector $x \in \R_{\ge 0}^d$ \emph{satisfies} an order $\pi \in \Perm(d)$ if $x_{\pi(1)} \ge \dots \ge x_{\pi(d)}$.

\begin{restatable}[Dual ordered norms]{lemma}{dualorderednorm}\label{lem: dual-ordered-norm}
Given a weight vector $w \in \R^d$, the dual norm $\|\cdot\|_{(w)}^*$ to ordered norm $\|\cdot\|_{(w)}$ is given by
\begin{align*}
    \|y\|_{(w)}^* = \max_{k \in [d]} \frac{\|y\|_{\mathbf{1}_k}}{\|w\|_{\mathbf{1}_k}}.
\end{align*}
\end{restatable}

\begin{restatable}[Ordered Cauchy-Schwarz]{lemma}{orderedcauchyschwarz}\label{lem: ordered-cauchy-schwarz}
For all $x, y \in \RO^d$, $\|x\|_{(w)} \|y\|_{(w)}^* \ge x^\top y.$
Further, equality holds if and only if
\begin{enumerate}
    \item there is some order $\pi \in \mathrm{Perm}(d)$ such that $x, y$ both satisfy $\pi$.
    \item for each $k \in [d]$ either $x^\da_k = x^\da_{k + 1}$ or $\frac{\|y\|_{\mathbf{1}_k}}{\|w\|_{\mathbf{1}_k}} = \|y\|_{(w)}^*$.
\end{enumerate}
\end{restatable}

\section{\orderandcount for \machineloadsidenticaljobs (\mlij)}\label{sec: identical-jobs-scheduling}

In this section, we introduce the \orderandcount framework and prove Theorem \ref{thm: identical-jobs-portfolio} for the \textsc{Machine-Loads-Identical-Jobs} (\mlij) problem. Recall that we seek to assign $n$ copies of a job among $d$ processors/machines with different processing times $p_i, i \in [d]$. This is the simplest model for workload distribution where some tasks must be distributed among individuals in a workplace: processors correspond to individuals, processing times represent their efficiencies, and balancing loads on machines corresponds to managing the workloads of the individuals.
Given a norm $\|\cdot\|_f$ on $\R^d$, the goal is to schedule the jobs to minimize the norm of the vector of machine loads. We seek a portfolio of solutions (i.e. schedules) for ordered norms $\ordered$ and symmetric monotonic norms $\smn$.

To see why size$>$1-portfolios are necessary at all, we observe a simple example where no solution is simultaneous $o(\sqrt{d})$-approximation: suppose there are $n = d$ jobs and $p_1 = 1$ while $p_2 = \dots = p_d = \sqrt{d}$. The optimal solution for $L_\infty$ (i.e. maximum load) minimization assigns one job per machine to get maximum load $\sqrt{d}$. The optimal solution for $L_1$ (i.e. total load) minimization assigns all jobs to the most efficient machine, i.e., machine $1$, for a total load of $d$. Therefore, any assignment with $< d/2$ jobs on machine $1$ is an $\Omega(\sqrt{d})$-approximation for $L_1$ norm, and any assignment with $\ge d/2$ jobs on machine $1$ is an $\Omega(\sqrt{d})$-approximation for $L_\infty$ norm. This motivates us to increase the portfolio size.

In Section \ref{sec: identical-jobs-scheduling-upper-bound}, we prove the upper bound on portfolio size in Theorem \ref{thm: identical-jobs-portfolio}, guaranteeing for each $\alpha > 4$ a size-$O\left(\frac{\log d}{\log(\alpha/4)}\right)$ portfolio that is $\alpha$-approximate for $\ordered$ and $O(\alpha \log d)$-approximate for $\smn$. We prove the lower bound showing that any $\alpha$-approximate portfolio for ordered norms must have size $\Omega\left(\frac{\log d}{\log \alpha + \log\log d}\right)$ in Section \ref{sec: identical-jobs-scheduling-lower-bound}. We will also prove (Theorem \ref{thm: portfolios-for-top-k-norms-are-not-portfolios-for-ordered-norms}) that there are instances of \mlij with optimal portfolio of size $2$ for $\topk$ but with no $O(1)$-approximate portfolio of size $o\left(\frac{\log d}{\log\log d}\right)$ for $\ordered$.

We start with some notation. Since all jobs are identical, we can identify a schedule by the number of jobs on each machine. If $n_i \in \ZO$ jobs are scheduled on machine $i$, then $\sum_{i \in [d]} n_i = n$, and the load vector is $x = x(n) = (n_1 p_1, \dots, n_d p_d)$. Therefore, the set of feasible vectors is $\domain = \{x \in \R_{\ge 0}^d: x_i = n_i p_i \ \forall \ i \in [d], \sum_{i} n_i = n\}$.
We can relabel the machine indices and assume without loss of generality that $0 < p_1 \le \dots \le p_d$.

\subsection{Portfolio upper bound}\label{sec: identical-jobs-scheduling-upper-bound}

At a high level, we show that special instances of \mlij that we call \emph{doubling instances} -- those where each $p_i$ is a power of $2$ -- satisfy two key properties: (i) any instance of \mlij is $2$-approximated by some doubling instance (Lemma \ref{lem: job-times-doubling-instances}), and (ii) the optimal solution $x^{\OPT}$ to a doubling instance satisfies $x^{\OPT}_1 \ge x^{\OPT}_2 \ge \dots \ge x^{\OPT}_d$ (Lemma \ref{lem: identical-jobs-opt-monotonic}), i.e., must satisfy a specific order of coordinates.
These inequalities allow us to relax the integrality constraints and consider the polyhedron $\mc{P} = \{x: \sum_i \frac{x_i}{p_i} = n; x_1 \ge \dots \ge x_d \ge 0\}$, where the coordinate-wise inequality constraints can be put in for doubling instances. This sets up
\orderandcount: there is only one possible order for vectors $x \in \mc{P}$, which is $x_1 \ge \dots \ge x_d \ge 0$.
Each ordered norm $\|x\|_{(w)} = w^\top x$ is a linear function over $\mc{P}$, and so the set of vertices of $\mc{P}$ form an optimal portfolio for ordered norms over $\mc{P}$ for the doubling instance and a $2$-approximate portfolio for the original instance. We show that we can restrict to $O(\log_{\alpha/4} d)$ of these vertices, losing factor $\alpha/4$. Finally, we lose another factor $2$ in rounding fractional solutions to integral ones, to get an overall approximation factor $\alpha$ for ordered norms.

\begin{wrapfigure}{r}{0.45\textwidth}
    \centering
    \includegraphics[width=0.40\textwidth]{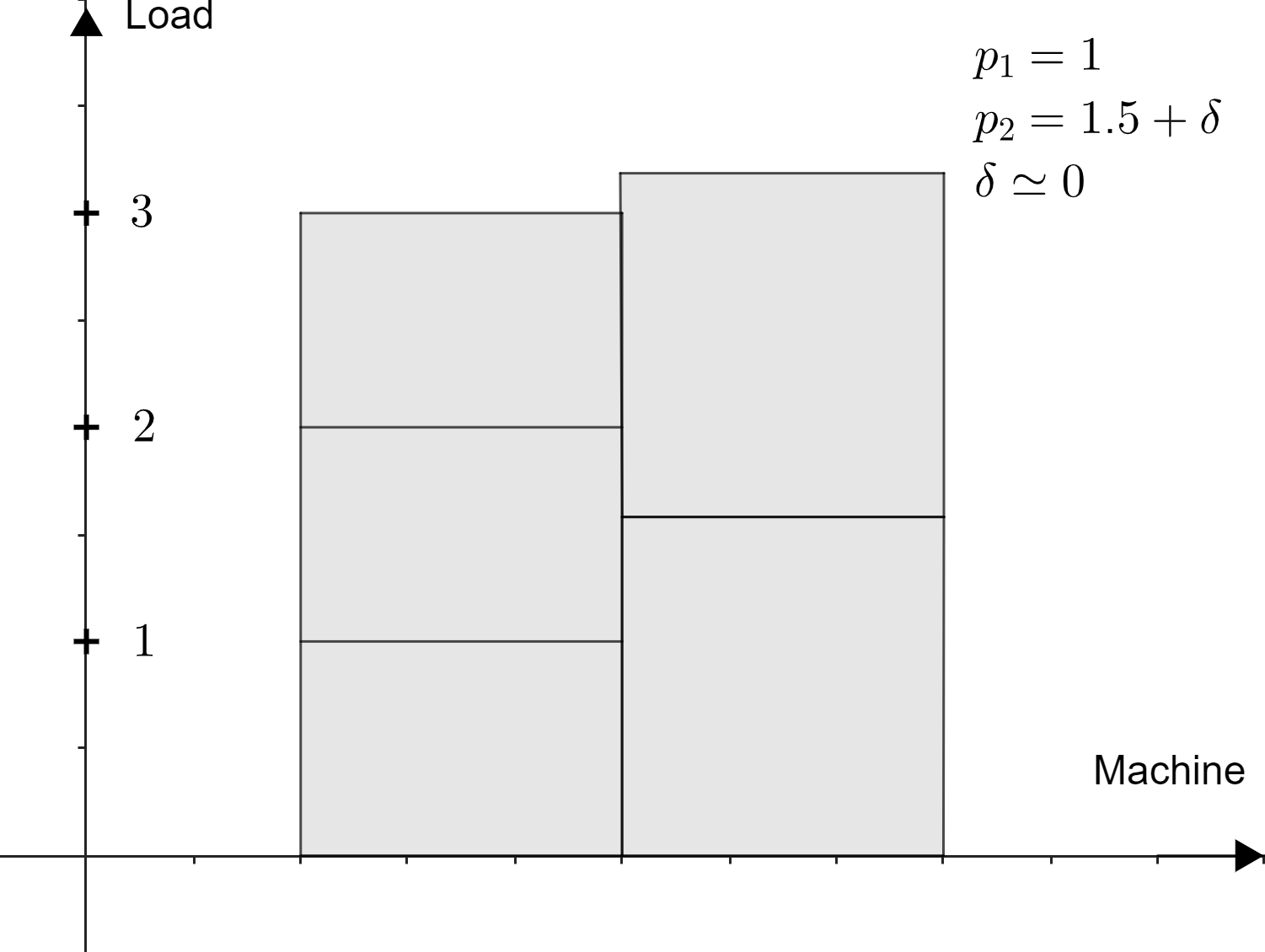}
    \caption{An example for makespan minimization with $2$ machines and $5$ jobs where $\xOPT_1 < \xOPT_2$ for optimal load vector $\xOPT$.}
    \label{fig: nonmonotonic-example}
    \vspace{-4em}
\end{wrapfigure}

\begin{lemma}\label{lem: job-times-doubling-instances}
Given an instance of \mlij with $d$ machines and $n$ copies of a job, we can get an instance of the problem with $d$ machines and $n$ jobs such that: for any load vector $x'$ for this modified instance, the corresponding load vector $x$ for the original instance satisfies
\[
    \frac{1}{\sqrt{2}} x \le x' \le \sqrt{2} x.
\]
\end{lemma}

\begin{proof}
    To construct the new instance, round each $p_i$ to its closest power of $2$, say $p_i'$. Then $\frac{1}{\sqrt{2}} p_i' \le p_i \le \sqrt{2} p_i'$. When $n_i$ jobs are scheduled on processor $i$, corresponding load vectors $x = (n_1 p_1, \dots, n_d p_d)$ and $x' = (n_1 p_1', \dots, n_d p_d')$ are within factor $\sqrt{2}$ of each other.
\end{proof}

\begin{corollary}\label{cor: identical-jobs-doubling-to-general}
For $\ordered$, an $\alpha$-approximate portfolio for an instance of \mlij\ can be obtained from a $\frac{\alpha}{2}$-approximate portfolio for the corresponding doubling instance.
\end{corollary}

Here is the first main idea of \orderandcount: we show next that for doubling instances, optimal load vector $\xOPT$ for any norm always satisfies the order $\xOPT_1 \ge \dots \ge \xOPT_d$. This is false if the instance is not doubling; see Figure \ref{fig: nonmonotonic-example}.

\begin{lemma}\label{lem: identical-jobs-opt-monotonic}
Suppose $x^{\OPT}$ is the optimal load vector for some symmetric monotonic norm $\|\cdot\|_f$ for a doubling instance. We can assume without loss of generality that $x^{\OPT}_1 \ge x^{\OPT}_2 \ge \dots \ge x^{\OPT}_d$.
\end{lemma}

\begin{proof}
    Suppose $\xOPT_i < \xOPT_{i + 1}$ for some $i$. Transfer one job from machine $i + 1$ to machine $i$, to get the new load vector $x$ defined as:
    \[
        x_{l} = \begin{cases}
                    \xOPT_l & \text{if} \: l \neq i, i + 1, \\
                    \xOPT_i + p_i & \text{if} \: l = i, \\
                    \xOPT_{i + 1} - p_{i + 1} & \text{if} \: l = i + 1.
        \end{cases}
    \]
    Since $p_i$ divides $p_{i + 1}$ and $\xOPT_{i + 1} > \xOPT_i$, we get that $\xOPT_{i + 1} - \xOPT_i \ge p_i$. Therefore,
    \[
        \max(x_i, x_{i + 1}) = \max\left(\xOPT_i + p_i, \xOPT_{i + 1} - p_{i + 1}\right) \le \xOPT_{i + 1} = \max\left(\xOPT_i, \xOPT_{i + 1}\right).
    \]
    Further, $x_i + x_{i + 1} < \xOPT_i + \xOPT_{i + 1}$. That is, $(x_i, x_{i + 1}) \prec (\xOPT_i , \xOPT_{i + 1})$. Since all other coordinates of $x$ and $\xOPT$ are equal, a simple inductive argument shows that $x \preceq \xOPT$. Lemma \ref{lem: majorization-norm-order-2} implies that $\|x\|_f \le \|\xOPT\|_f$, finishing the proof.
\end{proof}

For the rest of this section, we restrict ourselves to doubling instances; we will give an $\alpha/2$-approximate portfolio of size $\le 1 + \log_{\alpha/4} d$ for ordered norms over doubling instances.
For any weight vector $w$, Lemma \ref{lem: identical-jobs-opt-monotonic} allows us to relax the integer program (\ref{LP: identical-jobs-ip}) to a linear program: while not every load vector forms a feasible solution to \ref{LP: identical-jobs-ip}, Lemma \ref{lem: identical-jobs-opt-monotonic} shows that there is an optimal solution that is feasible for this IP.

\begin{minipage}{\textwidth}
    \begin{minipage}{0.48\textwidth}
        \begin{align}
            \label{LP: identical-jobs-ip}\tag{IP1} \min & \: w^\top x & \text{s.t.} \\
            \sum_i \frac{x_i}{p_i} &= n, \\
            x_i &\ge x_{i + 1} & \forall \: i \in [d - 1], \\
            \frac{x_i}{p_i} &\in \Z_{\ge 0} & \forall \: i \in [d],
        \end{align}
    \end{minipage} %
    \hfill
    \begin{minipage}{0.48\textwidth}
        \begin{align}
            \label{LP: identical-jobs-relaxation}\tag{LP1} \min & \: w^\top x & \text{s.t.} \\
            \sum_i \frac{x_i}{p_i} &= n, \\
            x_i &\ge x_{i + 1} & \forall \: i \in [d - 1], \\
            x &\ge 0.
        \end{align}
    \end{minipage}
\end{minipage}
\smallskip

Our next lemma characterizes the $d$ vertices of the constraint polytope $\mathcal{P} := \{x: \sum_i \frac{x_i}{p_i} = n; x_1 \ge \dots \ge x_d \ge 0\}$ of \ref{LP: identical-jobs-relaxation}. We omit the straightforward proof.

\begin{lemma}\label{lem: identical-jobs-vertices}
For any weight vector $w$, the optimal solution $x^*$ to \ref{LP: identical-jobs-relaxation} satisfies for some $l \in [d]$ that
\begin{align*}
    x_1^* = \dots = x_l^* = \frac{n}{\sum_{i \in [l]} \frac{1}{p_i}}, \quad x^*_{l + 1} &= \dots = x^*_d = 0.
\end{align*}
\end{lemma}

For $l \in [d]$, denote the $l$th vertex as $x(l) := \frac{n}{\sum_{i \in [l]} \frac{1}{p_i}} \mathbf{1}_l$, with $l$ non-zero entries. Call $x(l)$ \emph{good} if
\begin{equation}\label{eqn: scheduling-good-vertex}
\frac{n}{\sum_{i \in [l]} \frac{1}{p_i}} \ge p_l,
\end{equation}
i.e., the value of each non-zero coordinate is at least the processing time corresponding to the last non-zero coordinate. Clearly, $x(1)$ is good since $n \ge 1$, and if $x(l)$ is good then $x(l - 1)$ is also good. The next lemma says that if $x(l)$ is good, then it can be rounded to an integral load vector:

\begin{lemma}\label{lem: identical-jobs-good-rounding}
If $x(l)$ is good, then it can be rounded to $\hat{x}(l)$ that is feasible for \ref{LP: identical-jobs-ip} and $\frac{1}{2} x(l) \le \hat{x}(l) \le 2 x(l)$.
\end{lemma}

\begin{proof}
    Denote $n_i = \frac{x(l)_i}{p_i}$ for all $i \in [d]$, then $n_{l + 1} = \dots = n_d = 0$ and $\sum_{i \in [d]} n_i = n$. Then one can assign either $\hat{n}_i = \lfloor n_i \rfloor$ or $\hat{n}_i = \lceil n_i \rceil$ jobs to machine $i \in [d]$, while ensuring that $\sum_{i \in [d]} \hat{n}_i = n$. The load on machine $i \in [d]$ in this new schedule is $\hat{x}(l)$, with $\hat{x}(l)_i = p_i \hat{n}_i$.

    By definition of good vertices, $x(l)_i \ge p_l \ge p_i$ for each $i \in [l]$. Therefore, we get $n_i \ge 1$, thus implying $\frac{1}{2} n_i \le \lfloor n_i \rfloor \le n_i$ and $n_i \le \lceil n_i \rceil \le 2 n_i$ for all $i \in [l]$. This implies $\frac{1}{2} n_i \le \hat{n}_i \le 2n_i$ for all $i \in [d]$. Since $n_i = \frac{x(l)_i}{p_i}$ and $\hat{n}_i = \frac{\hat{x}(l)_i}{p_i}$, we get the result.
\end{proof}

Let $L$ be the largest index such that $x(L)$ is good.
Our next lemma shows that rounding good vertices gives a $2$-approximate portfolio for ordered norms:

\begin{lemma}\label{lem: identical-jobs-large-portfolio}
$\{\hat{x}(1), \dots, \hat{x}(L)\}$ is a $2$-approximate portfolio for $\ordered$ over the doubling instance.
\end{lemma}

\begin{proof}
    Fix a weight vector $w$. Let $x^{\OPT}$ be the (integral) optimal load vector for $\|\cdot\|_{(w)}$, and let $l$ be the largest index such that $\xOPT_{l} > 0$. We will first show that there exists an index $l' \le l$ such that (i) $x(l')$ is good, and (ii) $\|x(l')\|_{(w)} \le \|\xOPT\|_{(w)}$. Together with Lemma \ref{lem: identical-jobs-good-rounding}, this implies that $\|\hat{x}(l')\|_{(w)} \le 2\|\xOPT\|_{(w)}$, implying the lemma.

    We note first that $x(l)$ is good: since $\xOPT$ is integral and $\xOPT_l \neq 0$, we have $\xOPT_l \ge p_l$. From Lemma \ref{lem: identical-jobs-opt-monotonic}, we have $\xOPT_1 \ge \dots \ge \xOPT_l \ge p_l$. Since $\sum_{i \in [l]} \frac{\xOPT_i}{p_i} = n$, we get $n \ge \sum_{i \in [l]} \frac{p_l}{p_i} = p_l \sum_{i \in [l]} \frac{1}{p_i}$. That is, $x(l)$ is good.

    In particular, this implies that $x(l')$ is good for each $l' \le l$, so it is now sufficient to show that there is some $l' \le l$ such that $\|x(l')\|_{(w)} \le \|\xOPT\|_{(w)}$. Consider the following linear program:
    \begin{align}
        \label{LP: identical-jobs-relaxation-2}\tag{LP2} \min & \: w^\top x & \text{s.t.} \\
        \label{LP: constraint-total-jobs} \sum_i \frac{x_i}{p_i} &= n, \\
        x_i &\ge x_{i + 1} & \forall \: i \in [d - 1], \\
        x_{l + 1} &= \dots = x_d = 0.
    \end{align}
    $\xOPT$ is feasible for this LP by assumption. Further, by an argument similar to Lemma \ref{lem: identical-jobs-vertices}, we get that the vertices of the constraint polytope for this LP are $x(1), \dots, x(l)$. Therefore, there is some $l' \le l$ such that $\|x(l')\|_{(w)} = w^\top x(l') \le w^\top \xOPT = \|\xOPT\|_{(w)}$, finishing the proof.
\end{proof}

We are ready to prove the upper bound in Theorem \ref{thm: identical-jobs-portfolio}. We will convert the $2$-approximate portfolios of size $d$ for doubling instances to an $\alpha/2$-approximate portfolio of size $\sim \log_{\alpha/4} d$, which implies $\alpha$-approximate portfolios of size $\sim \log_{\alpha/4} d$ for \mlij\ by Corollary \ref{cor: identical-jobs-doubling-to-general}.

\begin{proof}[Proof of upper bound in Theorem \ref{thm: identical-jobs-portfolio}]
    We claim that for all indices $l, i \in [d]$ such that $i \le \frac{\alpha}{4} l$, we have $x(l) \preceq \frac{\alpha}{4} x(i)$. Therefore, $\|x(l)\|_{(w)} \le \frac{\alpha}{4} \|x(i)\|_{(w)}$ for all ordered norms $\|\cdot\|_{(w)}$ from Lemma \ref{lem: majorization-norm-order-2}, implying that $\left\{x((\alpha/4)^j) : j \in [0, 1 + \log_{(\alpha/4)} L]\right\}$ is an $(\alpha/2)$-approximate portfolio over doubling instances.

    Since $p_1 \le \dots \le p_d$ and $i \le \frac{\alpha}{4} l$, we have $\sum_{j \in [l]} \frac{1}{p_j} \ge
    \frac{4}{\alpha} \sum_{j \in [i]} \frac{1}{p_j}$. Therefore, for all $k \le l$, we have
    \begin{align*}
        \sum_{j \in [k]} x(l)_j = \frac{kn}{\sum_{i \in [l]} \frac{1}{p_j}} \le \frac{\alpha}{4} \cdot \frac{kn}{\sum_{j \in [i]} \frac{1}{p_j}} = \frac{\alpha}{4} \cdot \sum_{j \in [k]} x(i)_k.
    \end{align*}

    Further, for $k > l$,
    \[
        \sum_{j \in [k]} x(l)_j = \sum_{j \in [l]} x(l)_j = \frac{nl}{\sum_{j \in [l]} \frac{1}{p_j}} \le \frac{\alpha}{4} \frac{nl}{\sum_{j \in [i]} \frac{1}{p_j}} \le \frac{\alpha}{4} \sum_{j \in [i]} x(i)_j \le
        \frac{\alpha}{4} \sum_{j \in [k]} x(i)_j.
    \]
    Therefore, $x(l) \preceq (\alpha/4) x(i)$. This completes the proof.
\end{proof}

\subsection{Portfolio lower bound}\label{sec: identical-jobs-scheduling-lower-bound}

We prove the lower bound by giving an appropriate doubling instance with $d$ machines where any $\alpha$-approximate portfolio for ordered norms must have size $O\left(\frac{\log d}{\log \alpha + \log\log d}\right)$. Given $d$, let $S = S(d)$ be a superconstant that we specify later; assume that $S$ is an integer that is a power of $2$. Let $L$ be the largest integer such that $1 + S^2 + \dots + S^{2L} \le d$, then $L = \Theta(\log_S d)$. The $d$ machines are divided into $L + 1$ classes from $0$ to $L$: there are $S^{2l}$ machines in the $l$th class and the processing time on these machines is $p_l = S^l$. The number of jobs $n$ is $S^{3L}$; it is chosen so as to ensure that all vertices in the constraint polytope for \ref{LP: identical-jobs-relaxation} are good, and can be rounded to an integral solution that is only worse by a factor at most $2$ (Lemma \ref{lem: identical-jobs-good-rounding}).

There are $L + 1$ weight vectors for our instance. The first weight vector is $w(0) = (1, 1, \dots, 1)$. The second weight vector is $w(1) = \left(1, \frac{1}{S^2}, \frac{1}{S^2}, \dots, \frac{1}{S^2}\right)$. More generally, for $l \in [0, L]$,
\[
    w(l) = \bigg(1, \underbrace{\frac{1}{S^2}, \dots, \frac{1}{S^2}}_{S^2}, \underbrace{\frac{1}{S^4}, \dots, \frac{1}{S^4}}_{S^4}, \dots, \underbrace{\frac{1}{S^{2l - 2}}, \dots, \frac{1}{S^{2l - 2}}}_{S^{2l - 2}}, \underbrace{\frac{1}{S^{2l}}, \dots, \frac{1}{S^{2l}}}_{\text{remaining}} \bigg).
\]

With some foresight, we choose $S$ such that $\frac{S}{L} = 5\alpha$. We claim the following: for each $l \in [0, L - 1]$,
\begin{enumerate}
    \item There is a schedule $\hat{x}(l)$ for this instance with $\|\hat{x}(l)\|_{(w(l))} \le n lS^{-l}$. \label{todo: scheduling-lower-bound-condition-1}
    \item Any schedule $y$ that schedules more than $n/4$ jobs on machines in classes $l + 1$ to $L$ has $\|y\|_{(w(l))} \ge \frac{nS}{4} \cdot S^{-l}$. Combined with the above and since $\alpha \le \frac{S}{4L}$, it cannot be an $\alpha$-approximation for the $w(l)$-norm problem. \label{todo: scheduling-lower-bound-condition-2}
    \item Any schedule $y$ that schedules more than $n/4$ jobs on machines in classes $0$ to $l - 1$ has $\|y\|_{(w(l))} \ge \frac{nS}{2} \cdot S^{-l}$. Therefore, it cannot be an $\alpha$-approximation for the $w(l)$-norm problem either. \label{todo: scheduling-lower-bound-condition-3}
    \item $L = \Theta(\log_S d) = \Omega\left(\frac{\log d}{\log \alpha + \log\log d}\right)$. \label{todo: scheduling-lower-bound-condition-4}
\end{enumerate}

Claims \ref{todo: scheduling-lower-bound-condition-1}, \ref{todo: scheduling-lower-bound-condition-2}, and \ref{todo: scheduling-lower-bound-condition-3} imply that any $\alpha$-approximate solution for norm $w(l)$ must schedule at least $n/2$ jobs on machines in class $l$. Another application of claims \ref{todo: scheduling-lower-bound-condition-2} and \ref{todo: scheduling-lower-bound-condition-3} then implies that a portfolio that is $\alpha$-approximate for weight vectors $\{w(0), \dots, w(L - 1)\}$ must have distinct solutions for each weight vector, and therefore has size at least $L$. Claim \ref{todo: scheduling-lower-bound-condition-4} then implies our theorem.

Claim \ref{todo: scheduling-lower-bound-condition-4} is just computation: $L = \Theta(\log_S d) = \Theta(\log_{\alpha L} d) = \Theta\left(\frac{\log d}{\log \alpha + \log L}\right)$. If $L = \Omega(\log d)$, then we are done since the target size is anyway $\Theta\left(\frac{\log d}{\log \alpha + \log \log d}\right) = O(\log d)$ for constant $\alpha$. Otherwise, $\log L = O(\log \log d)$ and so $L = \Theta\left(\frac{\log d}{\log \alpha + \log L}\right) = \Theta\left(\frac{\log d}{\log \alpha + \log \log d}\right)$.

We move to claim \ref{todo: scheduling-lower-bound-condition-1}. As alluded to before, $n = S^{3L}$ has been chosen so that each vertex $x(l)$ of the constraint polytope is good (see inequality (\ref{eqn: scheduling-good-vertex})):
\[
    \frac{n}{1 \cdot \frac{1}{1} + S^2 \cdot \frac{1}{S} + \dots + S^{2L} \cdot \frac{1}{S^L}} \ge \frac{n}{2 S^{L}} \ge S^L = p_L.
\]

With this in hand, it is sufficient to give a fractional solution $x(l)$ with $\|x(l)\|_{(w(l))} = \Theta(n lS^{-l})$, since Lemma \ref{lem: identical-jobs-good-rounding} then implies the existence of an integral solution $\hat{x}(l)$ with norm at most twice. Consider $x(l) = \left(a, \dots, a, 0, \dots, 0\right)$ where the first $1 + S^2 + \dots + S^{2l}$ coordinates are non-zero and equal to $a$; all other coordinates are $0$. Since a total of $n$ jobs must be scheduled (constraint (\ref{LP: constraint-total-jobs})),
\[
    n = a \left(1 \cdot \frac{1}{1} + S^2 \cdot \frac{1}{S} + \dots + S^{2l} \cdot \frac{1}{S^l}\right) \ge a S^l,
\]
so that $a \le \frac{n}{S^l}$. Therefore,
\begin{align*}
    \|x(l)\|_{(w(l))} = a \times \text{ sum of first } (1 + S^2 + \dots + S^{2l}) \text{ coordinates of } w(l) = a \cdot l \le n l S^{-l}.
\end{align*}

We move to claim \ref{todo: scheduling-lower-bound-condition-2}. Let $y$ schedule more than $n/4$ jobs on machines in classes $l + 1$ to $L$. Irrespective of how these $n/4$ jobs are distributed, they contribute a total load of at least $(n/4) \times S^{l + 1}$. Since all coordinates of $w(l)$ are at least $\frac{1}{S^{2l}}$, the contribution of these jobs to $\|y\|_{(w(l))}$ is at least
\[
    \frac{1}{S^{2l}} \times \frac{n}{4} S^{l + 1} = \frac{nS}{4} \cdot S^{- l}.
\]
Since $l \le L = o(S)$, we get $\|y\|_{(w(l))} = \omega(n l S^{-l})$.

Finally, we prove claim \ref{todo: scheduling-lower-bound-condition-3}. Consider the restricted instance with only machines from classes $0, \dots, l - 1$ and $n/4$ jobs. Let $x$ be the optimal fractional solution for this instance for $L_\infty$ norm; it is easy to see that $x$ must have equal loads on machines, so that from constraint (\ref{LP: constraint-total-jobs}):
\[
    n = \|x\|_\infty \left(1 \cdot \frac{1}{1} + S^2 \cdot \frac{1}{S} + \dots + S^{2l - 2}\cdot\frac{1}{S^{l - 1}}\right) \le 2 \|x\|_\infty S^{l - 1},
\]
implying $\|x\|_\infty \ge \frac{n S^{-l + 1}}{2}$. Therefore, any integral optimal solution $\hat{x}$ to this restricted instance must also satisfy
\[
    \|\hat{x}\|_\infty \ge \|x\|_\infty \ge \frac{n S^{-l + 1}}{2}.
\]
Since $y$ is a solution to the larger original instance, we have $\|y\|_\infty \ge \|\hat{x}\|_\infty$. Finally, since $w(l) = 1$ by assumption, we get $\|y\|_{(w(l))} \ge \|y\|_\infty$. Together, we get $\|y\|_{(w(l))} \ge \frac{nS}{2} \cdot S^{-l}$. This completes the proof of the claim and of Theorem \ref{thm: identical-jobs-portfolio}.

\paragraph{Portfolios for different classes of norms.} Recall Lemma \ref{thm: top-k-implies-all-norm}: if $x^*$ is a simultaneous $\alpha$-approximation for each top-$k$ norm, then it is a simultaneous $\alpha$-approximation for all symmetric monotonic norms. One might naturally wonder if this is true for portfolios: is an $\alpha$-approximate portfolio for top-$k$ norms also an $\alpha$-approximate portfolio for all symmetric monotonic norms?
We show that not only is this false but that the gap between portfolio sizes for top-$k$ norms and ordered norms can be unbounded, by constructing such instances for \mlij.

\begin{theorem}\label{thm: portfolios-for-top-k-norms-are-not-portfolios-for-ordered-norms}
There exist instances of \mlij on $d$ machines for which
\begin{enumerate}
    \item there is an $O(1)$-approximate portfolio $X$ of size $2$ for $\topk$, and
    \item any $O(1)$-approximate portfolio $X'$ for $\ordered$ has size $\Omega\left(\frac{\log d}{\log\log d}\right)$.
\end{enumerate}
\end{theorem}

\begin{proof}
    From Theorem \ref{thm: identical-jobs-portfolio}, there exist instances of \mlij with $d$ machines where any $O(1)$-approximate portfolio for ordered norms must have size $O\left(\frac{\log d}{\log\log d}\right)$. We will show here that \emph{all} doubling instances of \mlij admit $O(1)$-approximate portfolio of size $2$ for all top-$k$ norms. Together, this implies the theorem.

    Recall Lemmas \ref{lem: identical-jobs-good-rounding}, \ref{lem: identical-jobs-large-portfolio}: $X' = \{\hat{x}(1), \dots, \hat{x}(L)\}$ is an $O(1)$-approximate portfolio for all ordered norms where $\frac{1}{2} x(l) \le \hat{x}(l) \le 2 x(l)$ for all $l \in [L]$. Therefore, $\|\hat{x}(l)\|_{\mathbf{1}_k}$ is within factor $2$ of $\|x(l)\|_{\mathbf{1}_k}$ for all $k \in [m]$. Further for all $k \in [m]$,
    \[
        \|x(l)\|_{\mathbf{1}_k} = \begin{cases}
                                      \frac{l n}{\sum_{i \in [l]} \frac{1}{p_i}} & \text{if} \: l \le k, \\
                                      \frac{k n}{\sum_{i \in [l]} \frac{1}{p_i}} & \text{if} \: l > k.
        \end{cases}
    \]
    Fix $k$. Since $p_i \le p_{i + 1}$ for all $i$, $\frac{l n}{\sum_{i \in [l]} \frac{1}{p_i}}$ is non-increasing in $l$. Further, $ \frac{k n}{\sum_{i \in [l]} \frac{1}{p_i}}$ is decreasing in $l$. Therefore, the smallest among $\|x(l)\|_{\mathbf{1}_k}, l \in [L]$ is either $\|x(1)\|_{\mathbf{1}_k}$ or $\|x(L)\|_{\mathbf{1}_k}$. This implies
    \begin{align*}
        \min\{\|\hat{x}(1)\|_{\mathbf{1}_k}, \|\hat{x}(L)\|_{\mathbf{1}_k}\} &\le 2 \min\{\|{x}(1)\|_{\mathbf{1}_k}, \|{x}(L)\|_{\mathbf{1}_k}\} \\
        &\le 2 \min\{\|{x}(1)\|_{\mathbf{1}_k}, \|{x}(2)\|_{\mathbf{1}_k}, \dots, \|{x}(L)\|_{\mathbf{1}_k}\} \\
        &\le 4 \min\{\|\hat{x}(1)\|_{\mathbf{1}_k}, \|\hat{x}(2)\|_{\mathbf{1}_k}, \dots, \|\hat{x}(L)\|_{\mathbf{1}_k}\}.
    \end{align*}

    Since $\{\hat{x}(1), \dots, \hat{x}(L)\}$ is an $O(1)$-approximate portfolio for all ordered norms, this implies that $\{\hat{x}(1), \hat{x}(L)\}$ is an $O(1)$-approximate portfolio for all top-$k$ norms.
\end{proof}

\noindent \textbf{Example 2.} One can also show that \emph{portfolios for ordered norms are not portfolios for $L_p$ norms:} consider an instance of identical jobs scheduling with $p_i = \sqrt{i}$ for each $i \in [d]$. Denote $\rho(l) = \sum_{i \in [l]}\frac{1}{p_i} = \sum_{i \in [l]} \frac{1}{\sqrt{i}}$; also denote  the $d$th Harmonic number $H_d = \sum_{i \in [d]} \frac{1}{i} = \Theta(\log d)$. Then for each $l \in [d]$, $x(l) = \frac{n}{\rho(l)} \big(\underbrace{1, \dots, 1}_{l}, 0, \dots, 0 \big)$. Recall (Lemmas \ref{lem: identical-jobs-good-rounding}, \ref{lem: identical-jobs-large-portfolio}) that there exists an $L \in [d]$ such that (1) $\frac{1}{2} x(l) \le \hat{x}(l) \le 2 x(l)$ for all $l \in [L]$ and (2) $X' =\{\hat{x}(1), \dots, \hat{x}(L)\}$ is an $O(1)$-approximate portfolio for ordered norms for some $L \in [d]$. We claim that each $x \in X'$ is an $\Omega(\sqrt{H_d})$-approximation for the $L_2$ norm objective.

For each $l \in [L]$, $\rho(l) \le 1 + 2 \int_1^l \frac{2}{\sqrt{x}} \le 4 \sqrt{l}$, so that
\[
    2 \|\hat{x}(l)\|_2 \ge \|x(l)\|_2 = \frac{n}{\rho(l)} \cdot \sqrt{l} \ge \frac{n}{4}.
\]
Consider the following assignment: assign $n_i = \frac{n}{i H_d}$ jobs to machine $i \in [d]$, and choose $n$ large enough so that each $n_i$ is integral. Then this is a valid assignment since $\sum_{i \in [d]}n_i = n$ by definition of $H_d$. The machine loads for this assignment are $x_i = n_i p_i = \frac{n}{H_d \sqrt{i}}$. The $L_2$ norm of $x$ is
\[
    \|x\|_2 = \frac{n}{H_d} \sqrt{\sum_{i \in [d]} \frac{1}{i}} = \frac{n}{\sqrt{H_d}}.
\]
Therefore, each $\hat{x}(l)$ is an $\Omega(\sqrt{H_d})$-approximation for the $L_2$ norm.

\section{\orderandcount for \coveringpolyhedra}\label{sec: covering-polyhedra}

In this section, we use \texttt{OrderAndCount} prove Theorem \ref{thm: covering-polyhedron} to obtain portfolios for \coveringpolyhedra. A $d$-dimensional covering polyhedron is defined as $\mc{P} = \{x \in \R^d: Ax \ge b, x \ge 0\}$ where $A \in \R_{\ge 0}^{r \times d}$ is the constraint matrix with $r$ constraints and $b \in \R_{\ge 0}^r$. As alluded to before, such polyhedra model workload management in settings with $r$ splittable jobs to be distributed among $d$ machines, each of which can run all $r$ jobs concurrently. We give an algorithm that given $\mathcal{P}$ and any constant $\epsilon > 0$, obtains a portfolio of size $O\left(\left(\frac{\log d/\epsilon}{\epsilon}\right)^{3r^2 - 2r}\right)$ that is (i) $(1 + \epsilon)$-approximate for ordered norms and (ii) $O(\log d)$-approximate for symmetric monotonic norms.

We focus on the result for $\ordered$ since the result for $\smn$ follows from Lemma \ref{lem: portfolios-ordered-norms-to-all-norms}. Assume that $b = \mathbf{1}_r = (1, \dots, 1)^\top$, without loss of generality by rescaling rows of $A$ if necessary (and removing rows with $b=0$ since they will be feasible anyway).

For any order or permutation $\pi$ on $[d]$, define restriction $\mc{P}_\pi := \mc{P} \cap \{x \in \R^d: x_{\pi(1)} \ge \dots \ge x_{\pi(d)} \ge 0\}$.
Our high-level plan is the same: any ordered norm $\|\cdot\|_{(w)}$ is a linear function on each $\mc{P}_\pi$. Therefore, the minimum norm point $x(w) := {\arg\min}_{x \in \mc{P}} \|x\|_{(w)}$ must be one of the vertices of some $\mc{P}_\pi$. Call $X$ the union of sets of vertices across all orders $\pi$; then $X$ is an \emph{optimal} portfolio for $\ordered$. However, two main issues potentially blow up the size $|X|$:
\begin{enumerate}
    \item Each $\mc{P}_\pi$ can have too many vertices. For each vertex of $\mc{P}_\pi$, $d$ out of $r + d$ constraints $Ax \ge \mathbf{1}_r, x_{\pi(1)} \ge \dots \ge x_{\pi(d)} \ge 0$ must be tight. Therefore, $\mc{P}_\pi$ may have $\binom{d + r}{d} \sim d^r$ vertices.
    \item There are $d!$ orders $\pi \in \text{Perm}(d)$. Since we are taking a union over all such orders, we get the following rough bound on the portfolio size $|X|$:
    \begin{equation}\label{eqn: portfolio-size-naive-bound}
    \left(\substack{\text{number of vertices}\\ \text{in each } \mc{P}_\pi}\right) \times \left(\substack{\text{number of orders}\\ \pi}\right) \sim d^r \times d!.
    \end{equation}
\end{enumerate}

Broadly, we first use a \emph{sparsification} idea to reduce the effective dimension to $\left(\frac{\log (d/\epsilon)}{\epsilon}\right)^r$ from $d$, losing approximation factor $1 + \epsilon$. This is done by counting the number of unique columns of $A$ up to factor $1 + \epsilon$.
Sparsification also gives an upper bound on the number of vertices in the restricted region $\mc{P}_\pi$ corresponding to each order.
There are still too many orders to sum over, and this is where the \emph{primal-dual counting technique} comes in. It allows us to restrict to a small number of permutations $\pi$ by counting in a suitable `dual' space to our primal problem:

\begin{minipage}{\textwidth}
    \begin{minipage}{0.49\textwidth}
        \begin{equation}\label{lp: polyhedron-norm-minimization}\tag{Primal}
        \min_{x \ge 0} \|x\|_{(w)} \quad \text{s.t.} \; Ax \ge b.
        \end{equation}
    \end{minipage}
    \hfill
    \begin{minipage}{0.49\textwidth}
        \begin{equation}\label{lp: polyhedron-norm-minimization-dual}\tag{Dual}
        \min \|\lambda^\top A\|_{(w)}^* \quad \text{s.t.} \; \lambda \in \mathbf{\Delta}_r.
        \end{equation}
    \end{minipage}
\end{minipage}
The advantage with the `dual' is that the underlying polytope -- the simplex $\mathbf{\Delta}_r$ in $r$ dimensions -- is easier to handle, and additionally is in $r$ dimensions instead of $d$.
The key ingredient connecting the primal and the dual will be the Cauchy-Schwarz inequality for ordered norms (Lemma \ref{lem: ordered-cauchy-schwarz}).

\subsection{Sparsification}\label{sec: covering-polyhedra-sparsification}

Denote $N = O\big(\frac{\log(d/\epsilon)}{\epsilon}\big)$. We give a sparsification procedure that reduces the number of distinct columns in $A$ to $N^r$. For each row of matrix $A$, this sparsification (1) removes `small' entries in the row and (2) restricts the number of unique entries in the row to $N$. Since there are $r$ rows, the number of distinct columns after sparsification is $N^r$.

\begin{algorithm}[t]
    \caption{\texttt{SparsifyPolyhedron}($\mc{P}$)}
    \DontPrintSemicolon
    \KwData{covering polyhedron $\mc{P} = \{x \in \R^d: Ax \ge \mathbf{1}_r, x \ge 0\}$, error parameter $\epsilon \in (0, 1]$}
    \KwResult{another covering polyhedron  $\widetilde{\mc{P}} = \{x \in \R^d: \widetilde{A}x \ge \mathbf{1}_r, x \ge 0\}$}
    define $\mu = \frac{3d^2}{\epsilon}$ and initialize $\widetilde{A} = \mathbf{0}_{r \times d}$\;
    \For{$i = 1$ to $r$}
    {
        define $a_i^* = \max_{j \in [d]} A_{i, j}$ and $B(i) = \left\{j \in [d]: A_{i, j} < \frac{a_i^*}{\mu} \right\}$\;
    \For{$j \in [d]$}
    {
        \If{$j \in B(i)$}
        {set $\widetilde{A}_{i, j} = 0$\;}
        \Else
        {
            let $l \in [0, \lfloor \log_{(1 + \epsilon/2)} \mu \rfloor]$ be the unique integer such that
            \[
                \frac{a_i^*}{\mu} \left(1 + \frac{\epsilon}{2}\right)^l \le A_{i, j} <  \frac{a_i^*}{\mu} \left(1 + \frac{\epsilon}{2}\right)^{l + 1}
            \]
            set $\widetilde{A}_{i, j} = \frac{a_i^*}{\mu} \left(1 + \frac{\epsilon}{2}\right)^l$\;
        }
    }
    }
    \textbf{return} $\widetilde{A}$, $\widetilde{P} = \{x \in \R^d: \widetilde{A}x \ge \mathbf{1}_r, x \ge 0\}$
\end{algorithm}

\begin{lemma}\label{lem: column-unique-values}
The columns of matrix $\widetilde{A} \in \RO^{r \times d}$ output by Algorithm {\tt SparsifyPolyhedron} take one of $N^r$ values, i.e., $[d]$ can be partitioned into $S_1, \dots, S_{N^r}$ such that for any $j, j' \in S_l$, the $j$th and $j'$th columns of $\widetilde{A}$ are the same.
\end{lemma}

\begin{proof}
    Fix row $i \in [r]$. By construction, each entry in the $i$th row of $\widetilde{A}$ is in the set $\{0\} \cup \left\{\frac{a_i^*}{\mu} \left(1 + \frac{\epsilon}{2}\right)^l: l \in [0, \lfloor \log_{(1 + \epsilon/2)} \mu \rfloor ]\right\}$. These are $O(\log_{(1 + \epsilon/2)} \mu) = O(\log_{(1 + \epsilon/2)} (d^2/\epsilon)) = O\left(\frac{\log(d/\epsilon)}{\epsilon}\right) = N$ distinct numbers. Since each column is composed of $r$ entries, one from each row, we get a total of $N^r$ possible values for a column.
\end{proof}

Sparsification only loses a factor $(1 + \epsilon)$ in the approximation (proof deferred to Appendix \ref{sec: missing-proofs-2}):

\begin{lemma}\label{lem: sparsification-approximation}
$\widetilde{\mc{P}} = \{x: \widetilde{A} x \ge \mathbf{1}_r, x \ge 0\}$ output by Algorithm {\tt SparsifyPolyhedron} is a $(1 + \epsilon)$-approximate portfolio for $\smn$ over $\mc{P}$.
\end{lemma}

These lemmas allow us to work with $\widetilde{P} = \{x: \widetilde{A}x \ge \mathbf{1}_r, x \ge 0\}$ with the nice property that columns of $\widetilde{A}$ take at most $N^r$ distinct values. We will give an optimal portfolio for $\ordered$ over $\widetilde{\mc{P}}$ of size $O(N^{3r^2 - 2r})$. Using Lemma \ref{lem: portfolio-composition}, this is sufficient to prove Theorem \ref{thm: covering-polyhedron}. Hereafter, we will only work with the sparsified matrix $\widetilde{A}$ and polyhedron $\widetilde{\mc{P}}$. For ease of notation, we drop the symbol $\widetilde{A}$ and assume that the original matrix $A$ and corresponding polyhedron $\mc{P}$ are already given to us in the sparsified form.

Let $S_1, \dots, S_{N^r}$ denote the partition of $[d]$ based on the value of columns of $A$, i.e., for each $l \in [N^r]$ and $j, j' \in S_l$, $j$th and $j'$th columns of $A$ are the same. Further, define $\domain = \{x \in \RO^d: x_j = x_{j'} \: \forall \: j, j' \in S_l, \: \forall \: l \in [N^r]\}$, i.e., the set of all non-negative vectors that attain the same value for all $j \in S_l$, for all $l \in [N^r]$. Define $\mc{P}^= = \mc{P} \cap \domain$. Recall that for weight vector $w$, we define $x(w) := {\arg\min}_{x \in \mc{P}} \|x\|_{(w)}$. Our first lemma shows that $x(w) \in \mc{P}^=$:

\begin{lemma}\label{lem: P=-is-optimal}
Given a weight vector $w$, we can assume without loss of generality that for all $l \in [N^r]$ and $j, j' \in S_l$, $x(w)_j = x(w)_{j'}$. That is, $\mc{P}^=$ is an optimal portfolio for $\smn$ over $\mc{P}$.
\end{lemma}

\begin{proof}
    Suppose $x(w)_j \neq x(w)_{j'}$, say $x(w)_j > x(w)_{j'}$. Then consider $\ov{x} \in \R^d$ such that $\ov{x}_k = x(w)_k$ for all $k \neq j, j'$, and $\ov{x}_j = \ov{x}_{j'} = \frac{x(w)_j + x(w)_{j'}}{2}$. Then $\ov{x} \preceq x(w)$ and so Lemma \ref{lem: majorization-norm-order-2} gives  $\|\ov{x}\|_{(w)} \le \|x(w)\|_{(w)}$.

    Further, clearly $\ov{x} \ge 0$ since $x(w) \ge 0$. Since the $j$th and $j'$th columns of $A$ are equal, $A\ov{x} = Ax(w) \ge \mathbf{1}_r$, or that $\ov{x} \in \mc{P}$.
\end{proof}

We define 'reduced orders' next, which are simply orders in the smaller space $\R^{N^r}$:

\begin{definition}[Reduced orders]
    An order $\rho$ on $[N^r]$ is called a \emph{reduced order}. For $x \in \domain$ and $l \in [N^r]$, denote $z(x)_l = x_j$ for $j \in S_l$. $x \in \domain$ is said to \emph{satisfy} reduced order $\rho$ if $z_{\rho(1)} \ge \dots \ge z_{\rho(N^r)} \ge 0$. Given a reduced order $\rho$, define polyhedron
    \[
        \mc{P}_\rho^= = \{x \in \mc{P} \cap \domain: x \text{ satisfies reduced order } \rho\}.
    \]
\end{definition}

At this point, a natural first attempt at bounding the portfolio size is to count the number of ordered norms in the space of `reduced' vectors $\{z(x): x \in \mc{P}^=\} \subseteq \R^{N^r}$. After all, \cite{chakrabarty_approximation_2019}'s result shows that there are at most $\poly(N^{r/\epsilon})$ ordered norms in $\R^{N^r}$ up to a $(1 + \epsilon)$-approximation. However, this approach fails because ordered norms on $\R^d$ cannot be translated appropriately into an ordered norm on the smaller space $\R^{N^r}$.

For example, consider the covering polyhedron $\mathcal{P} = \{x \in \R_{\ge 0}^3: x_1 \ge 2, x_2 + x_3 \ge 4, 2x_1 + x_2 + x_3 \ge 10\}$. The point $(3, 2, 2) \in \mathcal{P}$ is the (unique) minimizer of the $L_1$ norm, which corresponds to weight vector $w = (1, 1, 1)$. The constraint polytope for $\mathcal{P}$ has two unique columns, and the corresponding `reduced covering polyhedron' is $\mathcal{P}' = \{z \in \R^2: z_1 \ge 2, z_2 \ge 2, z_1 + z_2 \ge 5\}$. A point $(a, b, b) \in \mathcal{P}$ corresponds to the point $(a, b) \in \mathcal{P}'$. However, by a majorization argument, the point $(5/2, 5/2) \in \mathcal{P}'$ minimizes \emph{all ordered norms} on $\mathcal{P}'$, but the corresponding point $(5/2, 5/2, 5/2) \in \mathcal{P}$ with $L_1$ norm $7.5$ is sub-optimal for the $L_1$ norm.

Therefore, it is not sufficient to count ordered norms in $\R^{N^r}$, and we need an alternate approach that we describe next. Suppose that we are given some reduced order $\rho$. Then for $x \in \mc{P}^=_\rho$, $\|x\|_{(w)}$ is a linear function of $x$. Therefore, given a weight vector $w$, if $x(w)$ satisfies reduced order $\rho$, then $x(w)$ is one of the vertices of polyhedron $\mc{P}^=_\rho$. With this observation, the rest of the proof is organized as follows:
\begin{itemize}
    \item For each reduced order $\rho$, $\mc{P}^=_\rho$ has at most $N^{r^2} + 1$ vertices (Lemma \ref{lem: fixed-order-portfolio}).
    \item Consider the set $\Pi$ of reduced orders such that for any weight vector $w$, $x(w)$ satisfies some reduced order $\rho \in \Pi$, i.e, $\Pi = \{\text{reduced order } \rho: \exists \: w \text{ where } x(w) \text{ satisfies } \rho \}$. Then we will show that $|\Pi| \le N^{2r(r - 1)}$ (Lemma \ref{lem: bound-on-number-of-orders}).
\end{itemize}

Together, these observations mean that $X := \bigcup_{\rho \in \Pi} \left(\text{vertices of } \mc{P}_\rho^=\right)$ is an optimal portfolio for $\ordered$ over $\mc{P}_=$. By Lemma \ref{lem: P=-is-optimal}, $\mc{P}_=$ is an optimal portfolio for $\ordered$ over $\mc{P}$. Therefore, Lemma \ref{lem: portfolio-composition} implies that $X$ is an optimal portfolio for $\ordered$ over $\mc{P}$. Further,
\begin{align*}
    |X| &= \Big|\bigcup_{\rho \in \Pi} \left(\text{vertices of } \mc{P}_\rho^=\right)\Big| \le \sum_{\rho \in \Pi} \left|\left(\text{vertices of } \mc{P}_\rho^=\right)\right| \\
    &\le \sum_{\rho \in \Pi} (N^{r^2} + 1) = |\Pi| (N^{r^2} + 1) \le N^{2r(r - 1)} (N^{r^2} + 1) = O(N^{3r^2 - 2r}).
\end{align*}

This implies Theorem \ref{thm: covering-polyhedron}. We prove Lemma \ref{lem: fixed-order-portfolio} next and defer Lemma \ref{lem: bound-on-number-of-orders} to the next section.

\begin{lemma}\label{lem: fixed-order-portfolio}
For each reduced order $\rho$, $\mc{P}^=_\rho$ has at most $N^{r^2} + 1$ vertices
\end{lemma}

\begin{proof}
    For simplicity, assume (after possibly relabeling indices) that $\rho(l) = l$ for all $l \in [N^r]$, and that $S_1 = \{1, \dots, |S_1|\}, S_2 = \{1 + |S_1|, \dots, |S_1| + |S_2|\}$ etc. Then the polyhedron $\mc{P}^=_\rho$ is the set of all $x$ such that $A_i^\top x \ge 1$ for all $i \in [r]$ and
    \[
        x_1 = \dots = x_{|S_1|} \ge x_{|S_1| + 1} = \dots = x_{|S_1| + |S_2|} \ge \dots \ge x_{d - |S_{N^r}| + 1} = \dots = x_{d} \ge 0.
    \]
    Any vertex corresponds to a set of $d$ (linearly independent) inequalities. The constraints of the polytope have $d - N^r$ equalities and $N^r + r$ inequalities. Therefore, each vertex corresponds to some $N^r$ of the $N^r + r$ inequalities being tight. The number of such choices is $\binom{N^r + r}{N^r}$. Then,
    \begin{align*}
        \binom{N^r + r}{N^r} = \binom{N^r + r}{r} \le \left(1 + \frac{N^r}{r}\right)^r.
    \end{align*}
    For $r = 1$, this is at most $1 + N^r$. For $r \ge 2$, $1 + \frac{N^r}{r} \le N^r \le N^{r^2}$.
\end{proof}

\subsection{Primal-dual counting}\label{sec: covering-polyhedra-dual-counting}

In this section, we study the set $\Pi$ of reduced orders such that for any weight vector $w$, $x(w)$ satisfies some reduced order $\rho \in \Pi$, i.e, $\Pi = \{\text{reduced order } \rho: \exists \: w \text{ where } x(w) \text{ satisfies } \rho \}$. We will prove the following:

\begin{lemma}\label{lem: bound-on-number-of-orders}
The number of possible reduced orders  $|\Pi| \le N^{2r(r - 1)}$.
\end{lemma}

The main idea is to count reduced orders not on $x(w)$, but in a \emph{dual space}. We write the following modified primal and dual, and denote $\lambda(w) = {\arg\min}_{\lambda \in \simplex{r}} \|\lambda^\top A \|_{(w)}^*$:

\begin{equation}\label{lp: polyhedron-norm-minimization-prime}
\min \|x\|_{(w)} \quad \text{s.t.} \quad Ax \ge \mathbf{1}_r, x \in \domain. \tag{Primal'}
\end{equation}

\begin{equation}\label{lp: polyhedron-norm-minimization-dual-prime}
\min \|A^\top \lambda\|_{(w)}^* \quad \text{s.t.} \quad  \lambda \in \simplex{r} \tag{Dual}
\end{equation}

Note that $(A^\top \lambda)_j$ is simply the dot product of the $j$th column of $A$ with $\lambda$. Further, recall for all $j, j' \in S_l$ for any $l \in [N^r]$, the $j$th and $j'$th columns of $A$ are equal. Therefore, we have $(A^\top \lambda)_j = (A^\top \lambda)_{j'}$ for any $\lambda$. By definition, this means that $A^\top \lambda \in \domain$ for all $\lambda \ge 0$.

The next lemma establishes the crucial connection between reduced orders in \ref{lp: polyhedron-norm-minimization-prime} and \ref{lp: polyhedron-norm-minimization-dual-prime}. It uses Lemma \ref{lem: ordered-cauchy-schwarz} (Ordered Cauchy-Schwarz) along with a Lagrangian function; we defer its proof to Appendix \ref{sec: missing-proofs-2}.

\begin{restatable}{lemma}{primaldualcounting}\label{lem: primal-dual-counting}
Given a weight vector $w$, $\|x(w)\|_{(w)} \|A^\top \lambda(w)\|_{(w)}^* = 1$. Further, there is a reduced order $\rho$ such that both $x(w), A^\top\lambda(w)$ satisfy $\rho$.
\end{restatable}

As a consequence of this lemma, we get that it is sufficient to count reduced orders in the dual:
\begin{align*}
    \Pi &= \{\text{reduced order } \rho: \exists \: w \text{ where } x(w) \text{ satisfies } \rho \} \\
    &= \{\text{reduced order } \rho: \exists \: w \text{ where } A^\top\lambda(w) \text{ satisfies } \rho \} \\
    &\subseteq \{\text{reduced order } \rho: \exists \: \lambda \in \simplex{r} \text{ where } A^\top\lambda  \text{ satisfies } \rho \}.
\end{align*}

Denote $\Pi^* = \{\text{reduced order } \rho: \exists \: \lambda \in \simplex{r} \text{ where } A^\top\lambda  \text{ satisfies } \rho \}$. We will show that $|\Pi^*| \le N^{2r(r - 1)}$. From the above, this is sufficient to prove Lemma \ref{lem: bound-on-number-of-orders}. Our final lemma is a geometric counting inequality.\footnote{This result also follows from \cite{winder1966partitions}'s (stronger) bound on the number of regions induced by $T$ hyperplanes in an $r$-dimensional Euclidean space. For completeness, we provide a (shorter) proof here.}
\begin{lemma}
    $T$ hyperplanes partition $\simplex{r}$ into at most $T^{r - 1} + 1$ regions.
\end{lemma}

\begin{proof}
    The result is trivially true for $r = 1$ since $\simplex{1}$ is a point. For $r = 2$, $\simplex{2}$ is a line segment, and $T$ `hyperplanes' partition it into $\le T + 1$ regions. For $r \ge 3$, we use induction on $T$. 1 hyperplane clearly divides any convex body into at most $2 \le 1^{r - 1} + 1$ regions. Suppose $T > 1$. Let the $T$th hyperplane be $\mc{H}$. By the induction hypothesis, the first $T - 1$ hyperplanes divide $\simplex{r}$ into at most $(T - 1)^{r - 1} + 1$ regions. If $\simplex{r} \subseteq \mc{H}$, then $\mc{H}$ does not add any new regions, and we are done.

    Otherwise, the number of new regions $\mc{H}$ adds is the number of regions that the first $T - 1$ hyperplanes partition $\simplex{r} \cap \mc{H}$ into. But $\simplex{r} \cap \mc{H}$ can be linearly transformed into $\simplex{r - 1}$ in this case, and so the number of new regions is at most $(T - 1)^{r - 2} + 1$. Therefore, by the induction hypothesis, the total number of regions with $T$ hyperplanes is at most
    \[
        ((T - 1)^{r - 1} + 1) + ((T - 1)^{r - 2} + 1) \le T^{r - 1} + 1 \quad \forall \: T \ge 1, r \ge 3. \qedhere
    \]
\end{proof}

We are ready to finish the proof of Lemma \ref{lem: bound-on-number-of-orders}. Partition $\simplex{r}$ into regions $\{R_\rho: \rho \in \Pi^*\}$, where $R_\rho := \{\lambda \in \simplex{r}: A^\top \lambda \text{ satisfies } \rho\}$. The size $|\Pi^*|$ is exactly the number of such regions. Pick $j, j' \in [d]$ such that $j, j'$ belong to different sets $S_l, S_{l'}$. Then these regions are separated by hyperplanes of the form $\{\lambda: (A^\top \lambda)_j = (A^\top \lambda)_{j'}\}$, i.e., different reduced orders exist on different sides of these hyperplanes. There are $\binom{N^r}{2}$ such hyperplanes, each corresponding to a pair of sets $S_l, S_{l'}$. By the above lemma, these partition $\simplex{r}$ into at most
\[
    \binom{N^r}{2}^{r - 1} + 1 = \left(\frac{N^r(N^r - 1)}{2}\right)^{r - 1} + 1 \le N^{2r(r - 1)}.
\]
regions. Thus, $|\Pi| \le |\Pi^*| = |\{R_\rho: \rho \in \Pi^*\}| \le N^{2r(r - 1)}$. This finishes the proof of Lemma \ref{lem: bound-on-number-of-orders}, and therefore the proof of Theorem \ref{thm: covering-polyhedron}.

We finally remark that this can be converted into an algorithm that runs in time $\poly(N^{r^2}, d)$: tracing back, find the set $\Pi^*$ using the above hyperplane argument, and then simply output the union of vertices of $\mc{P}^{=}_\rho$ for all $\rho \in \Pi^*$.

\section{\iterativeordering framework}\label{sec: iterative-ordering}

This section presents our \iterativeordering framework to obtain simultaneous approximations for various combinatorial problems, including \completiontimes, \orderedtsp, and \osc. As we will show, all of these problems (1) involve a set of clients and a set of objects that \emph{satisfy} clients, and (2) seek an order on the objects that minimizes the satisfaction time of clients. This is formalized in Definition~\ref{def: ordered-satisfaction-problem}.
Additionally, such problems are often \emph{composable}, in the sense that orders on different subsets of objects can be combined into a single order on the union of the subsets; this is formalized in Definition \ref{def: composable-problems}.

Various norms of the vector of satisfaction times correspond to different fairness objectives and to different combinatorial problems.
We are interested in global guarantees, i.e., simultaneous approximations for all symmetric monotonic norms of this vector. \emph{A priori}, it is unclear whether a given problem even admits good simultaneous approximations.
As \cite{golovin_all-norms_2008} note, many previous works \cite{blum_minimum_1994, golovin_all-norms_2008, farhadi_traveling_2021} contain similar algorithmic ideas to obtain polynomial-time simultaneous approximations for such problems. We go a step further and formalize the underlying algorithm as \iterativeordering. As we show in Theorem \ref{thm: iterative-ordering}, applying it to \completiontimes gives the first constant-factor simultaneous approximations for this problem, to the best of our knowledge.
Applying it to \orderedtsp and \osc proves the \emph{existence} of better-than state-of-the-art simultaneous $O(1)$-approximations. Similar ideas apply to \clustering problems; we present improved simultaneous approximation to \clustering in Appendix \ref{sec: clustering-and-facility-location}.

We begin by formally defining the combinatorial problems considered in this section:
\begin{itemize}
    \item \completiontimes. The input consists of $n$ jobs, $d$ machines, and processing times $p_{i, j} > 0$ for each job $j \in [n]$ on machine $i \in [d]$. The output is an assignment of jobs to machines, and an order on the jobs assigned to each machine. Given a norm $\|\cdot\|_f$ on $\R^n$, the objective is to minimize the norm of the \emph{completion times} of jobs.\footnote{Note that this is different from minimizing norms of machine loads that we considered in \mlij. The two problems have different fairness interpretations: \completiontimes captures fairness for jobs while \mlij captures fairness for machines.} Special cases include average completion time minimization (for the $L_1$ norm) \cite{williamson_design_2010}, and makespan minimization (for the $L_\infty$ norm) \cite{graham_bounds_1969}.

    \item \osc. The input consists of a ground set of $n$ elements and $m$ subsets $S_1, \dots, S_m$ of the ground set. The output is an order on the subsets; each output induces a vector of cover times of elements in the ground set, defined for an element as the position of the first set in the order containing it. Given a norm $\|\cdot\|_f$ on $\R^n$, the objective is to minimize the norm of cover times. Special cases include classical Set Cover (for the $L_\infty$ norm) \cite{johnson_approximation_1973}, and Min-Sum Set Cover or MSSC (for the $L_1$ norm) \cite{feige_approximating_2004}.

    \item \ovc. This is a special case of \osc where the ground set corresponds to edges of an undirected graph and the subsets correspond to vertices of the graph. Special cases include classical Vertex Cover (for the $L_\infty$ norm), and Min-Sum Vertex Cover or MSVC (for the $L_1$ norm) \cite{feige_approximating_2004}.

    \item \orderedtsp. The input consists of a metric space on $n$ points or vertices $V$ and a starting vertex $v_0 \in V$. The output is a Hamiltonian tour of the vertices starting at $v_0$; each tour induces a vector of visit times of the vertices, defined for a vertex as its distance from $v_0$ along the tour. Given a norm $\|\cdot\|_f$ on $\R^n$, the objective is to minimize the norm of visit times. Special cases include the Traveling Salesman Problem or TSP (for the $L_\infty$ norm) \cite{lawler_traveling_1991}, the Traveling Repairman Problem (for the $L_1$ norm) \cite{goemans_improved_1998}, and the Traveling Firefighter Problem (for the $L_2$ norm) \cite{farhadi_traveling_2021}.
\end{itemize}

\subsection{\ordsat problems}

Next, we formally define \ordsat problems that capture the common structure among the above-mentioned problems.

\begin{definition}\label{def: ordered-satisfaction-problem}
An \ordsat problem is specified by
\begin{enumerate}
    \item A set of clients $C$.
    \item A set $\mc{X}$ of \emph{objects}. Each object $x \in \mc{X}$ is associated with a subset $C(x)$ of clients that it \emph{satisfies}.
    \item Each collection $X \subseteq \mc{X}$ of objects is called a \emph{satisfier}, and is said to satisfy the clients in the union $C(X) := \bigcup_{x \in X} C(x)$.
    \item For each satisfier $X \subseteq \mc{X}$ and an order $\pi \in \Perm(X)$ on $X$, there is an associated time vector $t(X, \pi) \in \R_{\ge 0}^{X}$ that must satisfy the following downward closure property: given any time $T \in \R_{\ge 0}$ define another satisfier $X_T := \{x \in T: t(X, \pi)_x \le T\} \subseteq X$ with corresponding order $\pi_T$ on $X_T$ induced from $\pi$. Then we must have for all $x \in X_T$ that
    \begin{equation}\label{eqn: downward-closure}
    t(X_T, \pi_T)_x \le t(X, \pi)_x.
    \end{equation}
\end{enumerate}
For each satisfier $X \subseteq \mc{X}$ and order $\pi$ on $X$, also define the \emph{satisfaction time vector} $s(X, \pi) \in \R_{\ge 0}^{C(X)}$ as follows: for each client $e \in C(X)$, let $x \in X$ be the first object in the order $\pi$ that satisfies $e$, i.e., $x = {\arg\min}_{y \in X: e \in C(y)} \pi(y)$. Then the satisfaction time $s(X, \pi)_e$ of client $e$ is defined as
\begin{equation}\label{eqn: satisfaction-time-definition}
s(X, \pi)_e = t(X, \pi)_x.
\end{equation}
\end{definition}

The goal is to output a satisfier $X \subseteq \mc{X}$ that satisfies all clients (i.e., $C(X) = C$) and an order $\pi$ on $X$. The $L_1$ norm or the min-sum objective is to minimize the total satisfaction time $\sum_{e \in C} s(X, \pi)_e$ of clients and the $L_\infty$ norm or min-max objective is to minimize the maximum satisfaction time $\max_{e \in C} s(X, \pi)_e$ of clients across all $(X, \pi)$. More generally, given a symmetric monotonic norm $\|\cdot\|_f$ on $\R^C$, the corresponding objective is to minimize $\|s(X, \pi)\|_f$. We seek simultaneous approximations with  guarantees for all symmetric monotonic norms.

\begin{lemma}\label{lem: ordered-satisfaction-problems}
\completiontimes, \osc, \ovc, and \orderedtsp are \ordsat problems.
\end{lemma}

We give the proof for \completiontimes here, deferring the proof for the other three problems to Appendix \ref{sec: iterative-ordering-appendix}.

For \completiontimes, choose the set of clients $C = [n]$ as the set of jobs. Choose the set of objects to be $\mc{X} = [n] \times [d] = \{(j, i): j \in [n], i \in [d]\}$. The object $(j, i)$ represents the assignment of job $j$ to machine $i$; and we define $C(j, i) = \{j\}$, i.e., assigning job $j$ to machine $i$ satisfies job $j$. A satisfier $X \subseteq \mc{X}$ corresponds to a \emph{partial assignment}, where some jobs may be unassigned or assigned to multiple machines. Given machine $i \in [d]$, let $J_i(X)$ be the set of jobs assigned to machine $i$ in partial assignment $X$, i.e $J_i(X) = \{j \in [n]: (j, i) \in X\}$. Then any order $\pi$ on $X$ induces an order on $J_i(X)$.

Given $(j, i) \in X$ and an order $\pi$ on $X$, time $t(X, \pi)_{(j, i)}$ is defined naturally as the completion time of job $j$ on machine $i$, or more formally, as
\begin{equation}\label{eqn: scheduling-time-vector}
t(X, \pi)_{(j, i)} := \sum_{\substack{j' \in J_i(X): \\ \pi(j', i) \le \pi(j, i)}} p_{j', i},
\end{equation}
Similarly, the satisfaction time of a job $j$ is the least time across machines when it is completed: $s(X, \pi)_j = \min_{i: \ j \in J_i(X)} t(X, \pi)_{(j, i)}$. It is easy to see that the vectors satisfy downward closure (eqn. (\ref{eqn: downward-closure})) with equality: $X_T$ is simply the partial assignment for all jobs that finish under time $T$.

The goal is to find a schedule (with jobs possibly assigned to multiple machines), i.e., a pair $(X, \pi)$ such that $C(X) = [n]$.

\subsection{$\gamma$-\composable problems}

Given an \ordsat problem, consider the following process of composing subproblems: given satisfiers $X_1, \dots, X_k \subseteq \mc{X}$ with corresponding orders $\pi_1, \dots, \pi_k$ on them, consider the satisfier $\bigcup_{j \in [k]} X_j$ with a \emph{composed order} (denoted $\bigoplus_{j \in [k]} \pi_j$) where every object $x \in X_1$ is ordered first according to $\pi_1$, then every object $x \in X_2 \setminus X_1$ is ordered according to $\pi_2$, and so on. For example, in \completiontimes, this process corresponds to composing partial assignments one after the other, scheduling the jobs in the first partial assignment, then those in the next partial assignment, and so on.

In many \ordsat problems, including \completiontimes, such compositions suitably maintain the satisfaction times of the clients.
To formalize this, define the `cost' of a satisfier $X$ and order $\pi$ as $c(X, \pi) := \max_{x \in X} t(X, \pi)_x$. For example, for \completiontimes, $c(X, \pi)$ is the makespan of the corresponding partial assignment.

\begin{definition}\label{def: composable-problems}
Given $\gamma \ge 1$, an $\ordsat$ problem is called $\gamma$-\composable if for all satisfiers $X_1, \dots, X_k \subseteq \mc{X}$ and corresponding orders $\pi_1, \dots, \pi_k$, the time vector $t\left(X, \pi\right)$ for the composition $X := \bigcup_{j \in [k]} X_j$ and $\pi := \bigoplus_{j \in [k]} \pi_j$ satisfies the following: for each $j \in [k]$ and each object $x \in X_j \setminus \left(\bigcup_{l \in [j - 1]} X_l\right)$, we must have
\begin{equation}\label{eqn: gamma-composable}
t\left(X, \pi\right)_x \le \gamma \left(\sum_{l \in [j - 1]} c(X_l, \pi_l)\right) + t(X_j, \pi_j)_x.
\end{equation}
\end{definition}

For example, we show that $\completiontimes$ is $1$-composable: indeed, if partial assignments corresponding to $(X_1, \pi_1), (X_2, \pi_2), \dots, (X_k, \pi_k)$ are put one after the other to form a composed assignment $(X, \pi)$, then all jobs $j$ scheduled in $(X_1, \pi_1)$ finish by their completion time in partial assignment $(X_1, \pi_1)$, all jobs $j$ scheduled in $(X_2, \pi_2)$ finish by time $(\text{makespan of } (X_1, \pi_1) + \text{ completion time of } j \text{ in } (X_2, \pi_2))$, and so on.

We show in Appendix \ref{sec: iterative-ordering-appendix} that \osc and \ovc are both $1$-composable and \orderedtsp is $2$-composable.

\begin{lemma}\label{lem: composable-problems}
\completiontimes, \osc, and \ovc are $1$-composable and \orderedtsp is $2$-composable.
\end{lemma}

The next lemma follows by various definitions; we include its proof in Appendix \ref{sec: iterative-ordering-appendix}.
\begin{lemma}\label{lem: restriction-increases-cardinality}
Suppose we are given satisfier $X \subseteq \mc{X}$, order $\pi$ on $X$, and $T > 0$ for an \ordsat problem. Define $X_T = \{x \in X: t(X, \pi)_x \le T\}$, and let the restriction of $\pi$ to $T$ be denoted $\pi_T$.
Then
\begin{enumerate}
    \item $c(X_T, \pi_T) \le T$
    \item The number of clients $|C(X_T)|$ satisfied by $X_T$ is at least the number of clients $(X, \pi)$ satisfies within time $T$, i.e.
    \[
        |C(X_T)| \ge \left| \left\{e \in C(X): s(X, \pi)_e \le T \right\} \right|.
    \]
\end{enumerate}
\end{lemma}

\subsection{Algorithm \iterativeordering}\label{sec: iterative-ordering-algorithm}

In this subsection, we give simultaneous approximation algorithm \iterativeordering for $\gamma$-\composable problems. Among other results, we show the existence of a simultaneous $(\sqrt{\gamma} + 1)^2$-approximation, improving upon the state-of-the-art simultaneous approximations for \orderedtsp and \osc.
We also obtain various \emph{polynomial-time} approximations, giving the first simultaneous $O(1)$-approximation for \completiontimes.

Formally, given an approximation ratio $\alpha \ge 1$, a simultaneous $\alpha$-approximation for an \ordsat problem is a satisfier $X \subseteq \mc{X}$ and an order $\pi$ on $X$ such that $C(X) = C$ and for any other $X', \pi'$ with $C(X') = C$, and for any symmetric monotonic norm  $\|\cdot\|_f$ on $\R^C$, the corresponding satisfaction times of clients satisfy
\[
    \|s(X, \pi)\|_f \le \alpha \|s(X', \pi')\|_f.
\]

We need one last piece of the framework to state the algorithm \iterativeordering. Given a $\gamma$-\composable problem and some \emph{budget} $B$, consider the satisfier $X' \subseteq \mc{X}$ and order $\pi'$ on $X'$ that satisfies as many clients $|C(X')|$ as possible under the cost constraint $c(X', \pi') \le B$. Now consider the following relaxation: given $\beta \ge 1$, we call another satisfier $X$ and order $\pi$ on $X$ a $(\beta, B)$-satisfier if $c(X, \pi) \le \beta B$ and $|C(X)| \ge |C(X')|$, i.e., $(X, \pi)$ has cost within factor $\beta$ of the budget $B$ and satisfies at least as many clients as $(X', \pi')$.

Of course, $(X', \pi')$ (corresponding to $\beta = 1$) can always be found using an exhaustive search for any (finite) problem, but this search may take time exponential in the input size. For example, for \completiontimes, this search for $(X', \pi')$ for a given $B$ amounts to searching over all possible partial assignments with makespan $\le B$. As we show later, this is still useful in obtaining our results for the existence of simultaneous approximations.
For many problems, choosing a larger $\beta$ allows finding a $(\beta, B)$-satisfier in \emph{polynomial-time}, e.g., $\beta = 2$ for \orderedtsp and \completiontimes (see Appendix \ref{sec: iterative-ordering-appendix}). This difference accounts for the gap between our approximations for existence and polynomial-time algorithms.

\begin{algorithm}[t]\label{alg: iterative-ordering}
\caption{\texttt{IterativeOrdering}($\beta$)}
\DontPrintSemicolon
\KwData{A $\gamma$-\composable problem and parameter $\beta \ge 1$}
\KwResult{A satisfier $X \subseteq \mc{X}$ and order $\pi$ on $X$ such that $C(X) = C$}
set $\theta = \sqrt{\gamma} + 1$ \;
$j \gets 0$\;
\While{$\bigcup_{l \in [0, j - 1]} C(X_l) \neq C$}
{
    Set budget $B = \theta^j$\;
Find $(\beta, B)$-satisfier $X_j$ and corresponding order $\pi_j$\;
Increase counter $j \gets j + 1$\;
}
define satisfier $X = \bigcup_{i \in [0, j]} X_i$ and composed order $\pi \gets \oplus_{i \in [0, j]} \pi_i$\;
\Return{$X$ and $\pi$}
\end{algorithm}

Algorithm \iterativeordering is inspired by \cite{blum_minimum_1994}'s algorithm for the Traveling Repairman Problem (TRP), which was subsequently also used for \osc, \ovc by \cite{golovin_all-norms_2008}, who also mention its applicability to similar covering problems.
It takes as input a $\gamma$-composable problem $\beta \ge 1$, and constructs a simultaneous $\beta(\sqrt{\gamma} + 1)^2$-approximation to the problem. Choosing $\beta = 1$ gives the existence results while choosing appropriate $\beta > 1$ gives polynomial-time results.
We assume by re-scaling all costs that the minimum non-zero cost $c(X, \pi)$ across satisfiers $X \subseteq \mc{X}$ and orders $\pi$ on $X$ is  $1$.

\begin{lemma}\label{prop: composable-all-norm-approximation}
Given a $\gamma$-\composable problem and $\beta \ge 1$, \textup{\iterativeordering}$(\beta)$ gives a simultaneous $\left(\beta(\sqrt{\gamma} + 1)^2\right)$-approximation.
\end{lemma}

\begin{proof}
    Suppose there were $k$ total iterations in \iterativeordering; then the output satisfier is $X = \bigcup_{j \in [0, k]} X_j$ and corresponding order is $\pi = \bigoplus_{j \in [0, k]} \pi_j$.

    Fix symmetric monotonic norm $\|\cdot\|_f$. Let the optimal solution for this norm be $(X^*, \pi^*)$. We will show that for all $T > 0$, if $(X^*, \pi^*)$ satisfies $i$ clients within time $T$, then $(X, \pi)$ satisfies $\ge i$ clients within time $\beta(\sqrt{\gamma} + 1)^2 T$. Given corresponding satisfaction time vectors $s(X, \pi)$, $s(X^*, \pi^*) \in \R^C$; this is equivalent to saying that for any $i \in \{1, \dots, |C|\}$, the $i$th smallest entry of $s(X, \pi)$ is at most $\beta(\sqrt{\gamma} + 1)^2$ times the $i$th smallest entry of $s(X^*, \pi^*)$. Since $\|\cdot\|_f$ is symmetric and monotone, this implies that $(X, \pi)$ is a $\beta(\sqrt{\gamma} + 1)^2$-approximation.

    Given $T > 0$, define $X^*_T := \{x \in X^*: t(X^*, \pi^*)_x \le T\}$, and let $\pi^*_{T}$ be the restriction of $\pi^*$ to $X^*_T$. Then, by Lemma \ref{lem: restriction-increases-cardinality},
    \begin{equation}\label{eqn: composable-approximation-proof-1}
    c(X^*_T, \pi^*_T) \le T.
    \end{equation}
    Also by the same lemma,
    \begin{equation}\label{eqn: composable-approximation-proof-2}
    |C(X^*_T)| \ge |\{e \in C(X^*): s(X^*, \pi^*)_e \le T\}| := i.
    \end{equation}

    Let $j \in \Z_{\ge 0}$ be the unique integer such that $T \in (\theta^{j - 1}, \theta^j]$. Then, by definition of a $(\beta, B)$-satisfier, in iteration $j$ of the algorithm, we get $(X_j, \pi_j)$ such that (a) $c(X_j, \pi_j) \le \beta \theta^j$, and (b) $|C(X_j)| \ge |C(X^*_T)| \ge i$.

    For all clients $e \in C(X_j)$, $\gamma$-composability implies that the satisfaction time $s(X, \pi)_e \le \gamma \left(\sum_{l \in [0, j - 1]} c(X_l, \pi_l)\right) + c(X_j, \pi_j)$. Since the cost $c(X_l, \pi_l) \le \beta \theta^l$ for all $l \in [0, k]$, we have
    \[
        s(X, \pi)_e \le \gamma \left(\sum_{l \in [0, j - 1]} \beta \theta^l\right) + \beta \theta^j \le \beta \left(\gamma \frac{\theta^j}{\theta - 1} + \theta^j\right) = \theta^{j - 1} \times \beta \theta \left(\frac{\gamma}{\theta - 1} + 1\right).
    \]
    Therefore, $(X, \pi)$ satisfies at least $|C(X_j)| \ge i$ clients within time $\beta \theta \left(\frac{\gamma}{\theta - 1} + 1\right) \times \theta^{j - 1} < \beta \theta \left(\frac{\gamma}{\theta - 1} + 1\right) \times T$. Since $\theta = \sqrt{\gamma} + 1$, we have $\theta \left(\frac{\gamma}{\theta - 1} + 1\right) = (\sqrt{\gamma} + 1)^2$.
\end{proof}

This leads to the following results proving the existence of various simultaneous approximations, and a polynomial-time $8$-approximation for \completiontimes:

\begin{theorem}\label{thm: iterative-ordering}
\begin{enumerate}
    \item For any $\gamma$-\prob{Composable}\ problem, there always exists a simultaneous $(\sqrt{\gamma} + 1)^2$-approximation.
    \item For \osc, \ovc, and \completiontimes,  there always exists a simultaneous $4$-approximation.
    \item For \orderedtsp, there always exists a simultaneous $(3 + 2\sqrt{2})$-approximation.
    \item For \completiontimes, a simultaneous $8$-approximation can be found in polynomial-time.
\end{enumerate}
\end{theorem}

Part 1 of the theorem follows by choosing $\beta = 1$ in Proposition \ref{prop: composable-all-norm-approximation} and parts 2 and 3 follow from our observations in Lemma \ref{lem: composable-problems} that \osc, \ovc, and \completiontimes are $1$-composable while \orderedtsp is $2$-composable.

The proof of the last part involves giving a subroutine for \completiontimes that outputs a $(2, B)$-satisfier for each budget $B > 0$. This is equivalent to asking the following: given a time limit $B$, schedule as many given jobs as possible on the machines.
We show that \cite{shmoys_approximation_1993}'s $2$-approximation for makespan minimization generalizes to this setting (proof of the lemma in Appendix \ref{sec: iterative-ordering-appendix}):

\begin{lemma}\label{lem: partial-scheduling-approximation}
Given processing times $p_{i, j}$ for jobs $j \in [n]$ on machines $i \in [m]$, and a time budget $B$, find a partial schedule of jobs that (1) finishes within time $2B$, (2) schedules at least as many jobs as any partial schedule that finishes within time $B$.
\end{lemma}

\section{Discussion and open problems}\label{sec: discussion}

Motivated by fairness concerns in workload distribution and placement of critical facilities, we considered the portfolio problem that seeks a small number of feasible solutions to a given optimization problem with guarantees for all fairness criteria. We studied portfolios from an approximation and polyhedral perspective, and gave the first characterization of the trade-off between portfolio size and the approximation factors for the problem of scheduling identical jobs on unidentical machines, and then extended this result to covering polyhedra.
We also proposed the \iterativeordering framework and gave new or improved simultaneous approximations for various combinatorial problems.

Questions about the design of portfolios can be asked for any setting in optimization and for any class of objectives: fundamentally, portfolios simply ask if the set of feasible solutions can be represented by a smaller subset and still enjoy some guarantees for optimization for a given class of functions. We state some open questions here:
\begin{enumerate}
    \item {\bf General \coveringpolyhedra:} For covering polyhedra in dimension $d$, we improved portfolio sizes from the general bound of $\poly(d)$ when the number of constraints $r = o(\sqrt{\log d}/(\log\log d))$. We conjecture that this is tight up to polylogarithmic factors, i.e, that there exist covering polyhedra in dimension $d$ with $O(\log d)$ constraints such that any $O(\log d)$-approximate portfolios for symmetric monotonic norms must have polynomial size.

    \item {\bf Scheduling with unidentical jobs:} We show $O(1)$-approximate portfolios of size $O(\log d)$ for \mlij, i.e., machine load minimization on $d$ machines with identical jobs. It is open if there exists a similar-sized portfolio for the more general problem of machine-load minimization with \emph{unidentical} jobs. We believe that this may not be true.

    \item {\bf Determining best-possible simultaneous approximations.} For \orderedtsp, it is unlikely that our simultaneous $5.83$-approximation is the best-possible, since the only known lower bound on this number is $1.78$ \cite{farhadi_traveling_2021}. It would be interesting to close this gap in either direction. Similarly, it is unclear if our simultaneous $4$-approximation for \completiontimes or for \osc is tight.

    \item {\bf Gap between computability and existence:} For simultaneous approximations, there is also a gap between existence bounds and polynomial-time bounds (see Table \ref{tab: portfolios-size-1}). For \osc, this gap (factor $4$ vs $O(\log n)$, respectively) is explained by complexity theoretic conjectures; however, it is unclear why this gap exists for other problems, such as \orderedtsp (factor $5.83$ vs $8$ respectively) and \completiontimes (factor $4$ vs $8$ respectively).

    \item {\bf Class of equity objectives:} Our work focused on understanding portfolios for various families of symmetric monotonic norms. However, many more notions of equity have been proposed in the literature, such as lexicographically optimal solutions \cite{kumar_fairness_2000}, for which such questions are largely open.
\end{enumerate}

\bibliographystyle{abbrvnat}
\bibliography{references}

\begin{thebibliography}{43}
\providecommand{\natexlab}[1]{#1}
\providecommand{\url}[1]{\texttt{#1}}
\expandafter\ifx\csname urlstyle\endcsname\relax
  \providecommand{\doi}[1]{doi: #1}\else
  \providecommand{\doi}{doi: \begingroup \urlstyle{rm}\Url}\fi

\bibitem[Azar and Taub(2004)]{azar_all-norm_2004}
Y.~Azar and S.~Taub.
\newblock All-{Norm} {Approximation} for {Scheduling} on {Identical}
  {Machines}.
\newblock In \emph{Algorithm {Theory} - {SWAT}}, pages 298--310, 2004.
\newblock \doi{10.1007/978-3-540-27810-8_26}.

\bibitem[Battersby and Crush(2014)]{battersby_africas_2014}
J.~Battersby and J.~Crush.
\newblock Africa’s {Urban} {Food} {Deserts}.
\newblock \emph{Urban Forum}, 25\penalty0 (2):\penalty0 143--151, June 2014.
\newblock ISSN 1874-6330.
\newblock \doi{10.1007/s12132-014-9225-5}.
\newblock URL \url{https://doi.org/10.1007/s12132-014-9225-5}.

\bibitem[Bernhardt et~al.(2021)Bernhardt, Suleiman, and
  Kresge]{bernhardt2021data}
A.~Bernhardt, R.~Suleiman, and L.~Kresge.
\newblock Data and algorithms at work: the case for worker technology rights.
\newblock 2021.

\bibitem[Blum et~al.(1994)Blum, Chalasani, Coppersmith, Pulleyblank, Raghavan,
  and Sudan]{blum_minimum_1994}
A.~Blum, P.~Chalasani, D.~Coppersmith, B.~Pulleyblank, P.~Raghavan, and
  M.~Sudan.
\newblock The minimum latency problem.
\newblock In \emph{Symposium on {Theory} of {Computing} ({STOC})}, pages
  163--171, 1994.
\newblock ISBN 978-0-89791-663-9.
\newblock \doi{10.1145/195058.195125}.

\bibitem[Boyd and Vandenberghe(2004)]{boyd_convex_2004}
S.~Boyd and L.~Vandenberghe.
\newblock \emph{Convex {Optimization}}.
\newblock Cambridge University Press, Cambridge, 2004.

\bibitem[Chakrabarty and Swamy(2018)]{chakrabarty_interpolating_2018}
D.~Chakrabarty and C.~Swamy.
\newblock Interpolating between k-{Median} and k-{Center}: {Approximation}
  algorithms for ordered k-{Median}.
\newblock In \emph{International {Colloquium} on {Automata}, {Languages}, and
  {Programming} ({ICALP})}, volume 107, pages 29:1--29:14, 2018.
\newblock URL \url{http://drops.dagstuhl.de/opus/volltexte/2018/9033}.
\newblock ISSN: 1868-8969.

\bibitem[Chakrabarty and Swamy(2019)]{chakrabarty_approximation_2019}
D.~Chakrabarty and C.~Swamy.
\newblock Approximation algorithms for minimum norm and ordered optimization
  problems.
\newblock In \emph{Symposium on {Theory} of {Computing} ({STOC}) 2019}, pages
  126--137, June 2019.
\newblock ISBN 978-1-4503-6705-9.
\newblock \doi{10.1145/3313276.3316322}.

\bibitem[Chandra and Wong(1975)]{chandra_worst-case_1975}
A.~K. Chandra and C.~K. Wong.
\newblock Worst-{Case} {Analysis} of a {Placement} {Algorithm} {Related} to
  {Storage} {Allocation}.
\newblock \emph{SIAM Journal on Computing}, 4\penalty0 (3):\penalty0 249--263,
  Sept. 1975.
\newblock ISSN 0097-5397.
\newblock \doi{10.1137/0204021}.
\newblock URL \url{https://epubs.siam.org/doi/abs/10.1137/0204021}.

\bibitem[Charikar et~al.(2001)Charikar, Khuller, Mount, and
  Narasimhan]{charikar_algorithms_2001}
M.~Charikar, S.~Khuller, D.~M. Mount, and G.~Narasimhan.
\newblock Algorithms for facility location problems with outliers.
\newblock \emph{Symposium on Discrete Algorithms (SODA)}, pages 642--651, 2001.
\newblock ISSN 0898714907.
\newblock URL
  \url{http://www.scopus.com/inward/record.url?scp=26944440987&partnerID=8YFLogxK}.

\bibitem[Chlamtáč et~al.(2022)Chlamtáč, Makarychev, and
  Vakilian]{chlamtac_approximating_2022}
E.~Chlamtáč, Y.~Makarychev, and A.~Vakilian.
\newblock Approximating {Fair} {Clustering} with {Cascaded} {Norm}
  {Objectives}.
\newblock \emph{Symposium on Discrete Algorithms (SODA)}, pages 2664--2683,
  Jan. 2022.
\newblock \doi{10.1137/1.9781611977073.104}.
\newblock URL \url{https://epubs.siam.org/doi/abs/10.1137/1.9781611977073.104}.

\bibitem[Chouldechova(2017)]{chouldechova_fair_2017}
A.~Chouldechova.
\newblock Fair {Prediction} with {Disparate} {Impact}: {A} {Study} of {Bias} in
  {Recidivism} {Prediction} {Instruments}.
\newblock 5\penalty0 (2):\penalty0 153--163, June 2017.
\newblock ISSN 2167-6461.
\newblock \doi{10.1089/big.2016.0047}.
\newblock URL \url{https://www.liebertpub.com/doi/abs/10.1089/big.2016.0047}.

\bibitem[Conitzer et~al.(2019)Conitzer, Freeman, Shah, and
  Vaughan]{conitzer_group_2019}
V.~Conitzer, R.~Freeman, N.~Shah, and J.~W. Vaughan.
\newblock Group fairness for the allocation of indivisible goods.
\newblock In \emph{Proceedings of the {Thirty}-{Third} {AAAI} {Conference} on
  {Artificial} {Intelligence} and {Thirty}-{First} {Innovative} {Applications}
  of {Artificial} {Intelligence} {Conference}}, pages 1853--1860, Jan. 2019.
\newblock ISBN 978-1-57735-809-1.
\newblock \doi{10.1609/aaai.v33i01.33011853}.

\bibitem[Cummins and Macintyre(2002)]{cummins_food_2002}
S.~Cummins and S.~Macintyre.
\newblock “{Food} deserts”—evidence and assumption in health policy
  making.
\newblock \emph{BMJ : British Medical Journal}, 325\penalty0 (7361):\penalty0
  436--438, Aug. 2002.
\newblock ISSN 0959-8138.
\newblock URL \url{https://www.ncbi.nlm.nih.gov/pmc/articles/PMC1123946/}.

\bibitem[Farhadi et~al.(2021)Farhadi, Toriello, and
  Tetali]{farhadi_traveling_2021}
M.~Farhadi, A.~Toriello, and P.~Tetali.
\newblock The {Traveling} {Firefighter} {Problem}.
\newblock In \emph{Applied and {Computational} {Discrete} {Algorithms}
  ({ACDA})}, pages 205--216. 2021.
\newblock \doi{10.1137/1.9781611976830.19}.
\newblock URL \url{https://epubs.siam.org/doi/abs/10.1137/1.9781611976830.19}.

\bibitem[Feige(1998)]{feige_threshold_1998}
U.~Feige.
\newblock A threshold of ln n for approximating set cover.
\newblock \emph{Journal of the ACM}, 45\penalty0 (4):\penalty0 634--652, July
  1998.
\newblock ISSN 0004-5411.
\newblock URL \url{https://dl.acm.org/doi/10.1145/285055.285059}.

\bibitem[Feige et~al.(2004)Feige, Lovász, and
  Tetali]{feige_approximating_2004}
U.~Feige, L.~Lovász, and P.~Tetali.
\newblock Approximating {Min} {Sum} {Set} {Cover}.
\newblock \emph{Algorithmica}, 40\penalty0 (4):\penalty0 219--234, Dec. 2004.
\newblock ISSN 1432-0541.
\newblock URL \url{https://doi.org/10.1007/s00453-004-1110-5}.

\bibitem[Feldman et~al.(2015)Feldman, Friedler, Moeller, Scheidegger, and
  Venkatasubramanian]{feldman_certifying_2015}
M.~Feldman, S.~A. Friedler, J.~Moeller, C.~Scheidegger, and
  S.~Venkatasubramanian.
\newblock Certifying and {Removing} {Disparate} {Impact}.
\newblock In \emph{Proceedings of the 21th {ACM} {SIGKDD} {International}
  {Conference} on {Knowledge} {Discovery} and {Data} {Mining}}, {KDD} '15,
  pages 259--268, Aug. 2015.
\newblock ISBN 978-1-4503-3664-2.
\newblock \doi{10.1145/2783258.2783311}.

\bibitem[Gartin(2012)]{gartin_food_2012}
M.~Gartin.
\newblock Food deserts and nutritional risk in {Paraguay}.
\newblock \emph{American Journal of Human Biology}, 24\penalty0 (3):\penalty0
  296--301, 2012.
\newblock ISSN 1520-6300.
\newblock \doi{10.1002/ajhb.22270}.
\newblock URL \url{https://onlinelibrary.wiley.com/doi/abs/10.1002/ajhb.22270}.

\bibitem[Goel and Meyerson(2006)]{goel_simultaneous_2006}
A.~Goel and A.~Meyerson.
\newblock Simultaneous {Optimization} via {Approximate} {Majorization} for
  {Concave} {Profits} or {Convex} {Costs}.
\newblock \emph{Algorithmica}, 44\penalty0 (4):\penalty0 301--323, Apr. 2006.
\newblock ISSN 1432-0541.
\newblock \doi{10.1007/s00453-005-1177-7}.
\newblock URL \url{https://doi.org/10.1007/s00453-005-1177-7}.

\bibitem[Goemans and Kleinberg(1998)]{goemans_improved_1998}
M.~Goemans and J.~Kleinberg.
\newblock An improved approximation ratio for the minimum latency problem.
\newblock \emph{Mathematical Programming}, 82\penalty0 (1):\penalty0 111--124,
  June 1998.
\newblock ISSN 1436-4646.
\newblock URL \url{https://doi.org/10.1007/BF01585867}.

\bibitem[Golovin et~al.(2008)Golovin, Gupta, Kumar, and
  Tangwongsan]{golovin_all-norms_2008}
D.~Golovin, A.~Gupta, A.~Kumar, and K.~Tangwongsan.
\newblock All-{Norms} and {All}-{Lp}-{Norms} {Approximation} {Algorithms}.
\newblock \emph{IARCS Annual Conference on Foundations of Software Technology
  and Theoretical Computer Science}, 2008.

\bibitem[Graham(1966)]{graham_bounds_1966}
R.~L. Graham.
\newblock Bounds for certain multiprocessing anomalies.
\newblock \emph{The Bell System Technical Journal}, 45\penalty0 (9):\penalty0
  1563--1581, 1966.
\newblock \doi{10.1002/j.1538-7305.1966.tb01709.x}.

\bibitem[Graham(1969)]{graham_bounds_1969}
R.~L. Graham.
\newblock Bounds on {Multiprocessing} {Timing} {Anomalies}.
\newblock \emph{SIAM Journal on Applied Mathematics}, 17\penalty0 (2):\penalty0
  416--429, Mar. 1969.
\newblock ISSN 0036-1399.
\newblock URL \url{https://epubs.siam.org/doi/abs/10.1137/0117039}.

\bibitem[Gupta et~al.(2023)Gupta, Moondra, and Singh]{gupta_which_2023}
S.~Gupta, J.~Moondra, and M.~Singh.
\newblock Which {Lp} norm is the fairest? {Approximations} for fair facility
  location across all "p".
\newblock In \emph{Economics and {Computation} ({EC}) 2023}, page 817, July
  2023.
\newblock ISBN 9798400701047.
\newblock URL \url{https://doi.org/10.1145/3580507.3597664}.

\bibitem[Hardt et~al.(2016)Hardt, Price, Price, and
  Srebro]{hardt_equality_2016}
M.~Hardt, E.~Price, E.~Price, and N.~Srebro.
\newblock Equality of {Opportunity} in {Supervised} {Learning}.
\newblock In \emph{Advances in {Neural} {Information} {Processing} {Systems}},
  volume~29, 2016.
\newblock URL
  \url{https://papers.nips.cc/paper_files/paper/2016/hash/9d2682367c3935defcb1f9e247a97c0d-Abstract.html}.

\bibitem[Hardy et~al.(1952)Hardy, Littlewood, and
  Pólya]{hardy_inequalities_1952}
G.~H. Hardy, J.~E. Littlewood, and G.~Pólya.
\newblock \emph{Inequalities}.
\newblock Cambridge University Press, Cambridge, 1952.
\newblock URL
  \url{https://www.cambridge.org/us/academic/subjects/mathematics/abstract-analysis/inequalities}.

\bibitem[Johnson(1973)]{johnson_approximation_1973}
D.~S. Johnson.
\newblock Approximation algorithms for combinatorial problems.
\newblock In \emph{Symposium on {Theory} of {Computing} ({STOC})}, pages
  38--49, Apr. 1973.
\newblock URL \url{https://dl.acm.org/doi/10.1145/800125.804034}.

\bibitem[Kleinberg(2018)]{kleinberg_inherent_2018}
J.~Kleinberg.
\newblock Inherent {Trade}-{Offs} in {Algorithmic} {Fairness}.
\newblock In \emph{2018 {ACM} {International} {Conference} on {Measurement} and
  {Modeling} of {Computer} {Systems} ({SIGMETRICS})}, page~40, June 2018.
\newblock ISBN 978-1-4503-5846-0.
\newblock \doi{10.1145/3219617.3219634}.

\bibitem[Knapp et~al.(1995)Knapp, Koutsogeorgopoulou, and
  Smith]{knapp1995volunteers}
M.~Knapp, V.~Koutsogeorgopoulou, and J.~D. Smith.
\newblock \emph{Who volunteers and why?: The key factors which determine
  volunteering}.
\newblock Volunteer Centre UK, UK, 1995.

\bibitem[Kumar and Kleinberg(2000)]{kumar_fairness_2000}
A.~Kumar and J.~Kleinberg.
\newblock Fairness measures for resource allocation.
\newblock In \emph{Symposium on {Foundations} of {Computer} {Science}
  ({FOCS})}, pages 75--85, Nov. 2000.
\newblock \doi{10.1109/SFCS.2000.892067}.
\newblock ISSN: 0272-5428.

\bibitem[Lawler et~al.(1991)Lawler, Lenstra, Rinnooy~Kan, and
  Shmoys]{lawler_traveling_1991}
E.~Lawler, J.~K. Lenstra, A.~H.~G. Rinnooy~Kan, and D.~B. Shmoys.
\newblock \emph{The {Traveling} {Salesman} {Problem}: {A} {Guided} {Tour} of
  {Combinatorial} {Optimization}}.
\newblock Wiley {Series} in {Discrete} {Mathematics} \& {Optimization}. John
  Wiley and Sons, Hoboken, New Jersey, USA, 1991.

\bibitem[Locke et~al.(2003)Locke, Ellis, and Smith]{locke2003hold}
M.~Locke, A.~Ellis, and J.~D. Smith.
\newblock Hold on to what you’ve got: The volunteer retention literature.
\newblock \emph{Voluntary Action}, 5\penalty0 (3):\penalty0 81--99, 2003.

\bibitem[Lyu et~al.(2022)Lyu, Zhang, Guo, Hong, Yang, Wang, Yang, Liu, and
  Zhang]{lyu2022towards}
W.~Lyu, K.~Zhang, B.~Guo, Z.~Hong, G.~Yang, G.~Wang, Y.~Yang, Y.~Liu, and
  D.~Zhang.
\newblock Towards fair workload assessment via homogeneous order grouping in
  last-mile delivery.
\newblock In \emph{International Conference on Information \& Knowledge
  Management}, pages 3361--3370, 2022.

\bibitem[Manshadi and Rodilitz(2020)]{manshadi_online_2020}
V.~Manshadi and S.~Rodilitz.
\newblock Online {Policies} for {Efficient} {Volunteer} {Crowdsourcing}.
\newblock In \emph{Economics and {Computation} ({EC}) 2020}, pages 315--316,
  July 2020.
\newblock URL \url{https://doi.org/10.1145/3391403.3399519}.

\bibitem[Nikolić et~al.(2022)Nikolić, Dimitrijević, Nikolić, and
  Stojcev]{nikolic_survey_2022}
G.~S. Nikolić, B.~R. Dimitrijević, T.~R. Nikolić, and M.~K. Stojcev.
\newblock A {Survey} of {Three} {Types} of {Processing} {Units}: {CPU}, {GPU}
  and {TPU}.
\newblock pages 1--6, June 2022.
\newblock \doi{10.1109/ICEST55168.2022.9828625}.

\bibitem[O’Meara et~al.(2021)O’Meara, Culpepper, Misra, and
  Jaeger]{o2021equity}
K.~O’Meara, D.~Culpepper, J.~Misra, and A.~Jaeger.
\newblock Equity-minded faculty workloads: What we can and should do now.
\newblock \emph{American Council on Education}, 2021.

\bibitem[Patton et~al.(2023)Patton, Russo, and Singla]{patton_submodular_2023}
K.~Patton, M.~Russo, and S.~Singla.
\newblock Submodular {Norms} with {Applications} {To} {Online} {Facility}
  {Location} and {Stochastic} {Probing}.
\newblock In \emph{{APPROX}}, volume 275, pages 23:1--23:22, 2023.
\newblock ISBN 978-3-95977-296-9.
\newblock URL \url{https://drops.dagstuhl.de/opus/volltexte/2023/18848}.

\bibitem[Rooddehghan et~al.(2015)Rooddehghan, ParsaYekta, and
  Nasrabadi]{rooddehghan2015nurses}
Z.~Rooddehghan, Z.~ParsaYekta, and A.~N. Nasrabadi.
\newblock Nurses, the oppressed oppressors: A qualitative study.
\newblock \emph{Global journal of health science}, 7\penalty0 (5):\penalty0
  239, 2015.

\bibitem[Shmoys and Tardos(1993)]{shmoys_approximation_1993}
D.~B. Shmoys and E.~Tardos.
\newblock An approximation algorithm for the generalized assignment problem.
\newblock \emph{Mathematical Programming}, 62\penalty0 (1):\penalty0 461--474,
  Feb. 1993.
\newblock ISSN 1436-4646.
\newblock URL \url{https://doi.org/10.1007/BF01585178}.

\bibitem[USDA(2023)]{usda_usda_2023}
USDA.
\newblock {USDA} {ERS} - {Go} to the {Atlas}, 2023.
\newblock URL
  \url{https://www.ers.usda.gov/data-products/food-access-research-atlas/go-to-the-atlas/}.

\bibitem[Wang et~al.(2015)Wang, Liang, and Li]{wang_multi-resource_2015}
W.~Wang, B.~Liang, and B.~Li.
\newblock Multi-{Resource} {Fair} {Allocation} in {Heterogeneous} {Cloud}
  {Computing} {Systems}.
\newblock \emph{IEEE Transactions on Parallel and Distributed Systems},
  26\penalty0 (10):\penalty0 2822--2835, 2015.
\newblock ISSN 1558-2183.
\newblock \doi{10.1109/TPDS.2014.2362139}.

\bibitem[Williamson and Shmoys(2010)]{williamson_design_2010}
D.~P. Williamson and D.~B. Shmoys.
\newblock \emph{The {Design} of {Approximation} {Algorithms}}.
\newblock Cambridge University Press, Cambridge, 2010.
\newblock URL \url{https://www.designofapproxalgs.com/}.

\bibitem[Winder(1966)]{winder1966partitions}
R.~O. Winder.
\newblock Partitions of n-space by hyperplanes.
\newblock \emph{SIAM Journal on Applied Mathematics}, 14\penalty0 (4):\penalty0
  811--818, 1966.

\end{thebibliography}

\appendix

\section{Omitted proofs from Section \ref{sec: preliminaries}}\label{sec: missing-proofs-1}

\begin{proof}[Proof of Lemma \ref{lem: portfolio-composition}]
    \begin{enumerate}
        \item For any $f \in \mathbf{C}$, $\min_{x \in X_2} f(x) \le \alpha_2 \min_{x \in X_1} f(x) \le \alpha_2 \alpha_1 \min_{x \in \domain} f(x)$. The first inequality follows since $X_2$ is an $\alpha_2$-approximate portfolio for $\mathbf{C}$ over $X_1$ and the second inequality follows since $X_1$ is an $\alpha_1$-approximate portfolio for $\mathbf{C}$ over $\domain$.

        \item For each $f \in \mathbf{C}$,
        \[
            \min_{x \in \domain} f(x) = \min_{i \in [n]} \min_{x \in \domain_i} f(x) \le \min_{i \in [n]} \alpha \min_{x \in X_i} f(x) = \alpha \min_{x \in \cup_{i \in [n]} X_i} f(x).
        \]
        Therefore, $\cup_{i \in [n]} X_i$ is an $\alpha$-approximate portfolio for $\mathbf{C}$ over $\domain$.
    \end{enumerate}
\end{proof}

Next, we prove Lemma \ref{obs: polynomial-portfolios-for-smn} that gives $(1 + \epsilon)$-approximate portfolio of size-$\poly(d^{1/\epsilon)}$ for symmetric monotonic norms for any $\domain \in \R_{\ge 0}^d$. Our proof is a slight modification of \cite{chakrabarty_approximation_2019}'s proof that counts the number of ordered norms up to a $(1 + \epsilon)$-approximation.

\begin{proof}[Proof of Lemma \ref{obs: polynomial-portfolios-for-smn}]
    Denote $v^* = \min_{x \in \domain} \|x\|_{\infty}$, with the corresponding vector denoted $x^*$. Let $\ov{\domain} = \{x \in \domain: \|x\|_{\infty} \le d v^*\}$. We first claim that $\ov{\domain}$ is an optimal portfolio for all symmetric monotonic norms over $\domain$, i.e., for each symmetric monotonic norm $\|\cdot\|_f$, the corresponding minimum norm point ${\arg\min}_{x \in \domain} \|x\|_f \in \ov{\domain}$. To see this, let $\ov{x} = {\arg\min}_{x \in \domain} \|x\|_f$. Then,
    \begin{align*}
        \|\ov{x}\|_\infty \|(1, 0, \dots, 0)\|_f &\le \|\ov{x}\|_f & (\|\cdot\|_f \text{ is symmetric}) \\
        &\le \|x^*\|_f & (\text{optimality of } \ov{x}) \\
        &\le \|x^*\|_{\infty} \|(1, \dots, 1)\|_f  \\
        &\le v^*d \|(1, 0, \dots, 0)\|_f.
    \end{align*}

    This implies that $\|\ov{x}\|_\infty \le d v^*$, or that $\ov{x} \in \ov{\domain}$. Next, we will place all vectors in $\ov{\domain}$ in one of $d^{O(1/\epsilon)}$ buckets such that for any two vector $x, y$ in the same bucket, $x \preceq (1 + \epsilon) y$ and $y \preceq (1 + \epsilon) x$, so that by Lemma \ref{lem: majorization-norm-order-2}, $\|x\|_f \simeq_{1 + \epsilon} \|y\|_f$ for all symmetric monotonic norms $\|\cdot\|_f$. Consequently, it is sufficient to pick just one vector in each bucket to get a $(1 + \epsilon)$-approximate portfolio for all symmetric monotonic norms over $\domain$.

    Denote $T = \lceil \log_{1 + \frac{\epsilon}{3}} d \rceil$. Each bucket $B(a_1, \dots, a_T)$ is specified by an increasing sequence $a_1 \le a_2 \le \dots \le a_{T}$ of integers that lie in $[0, 2T]$. The number of such sequences is $\binom{3T}{T} \le 3^T = d^{O(1/\epsilon)}$, bounding the number of buckets. Let $c_i = \left\lfloor (1 + \epsilon/3)^{i} \right\rfloor$ for $i \in [T]$. Then $x$ lies in bucket $B(a_1, \dots, a_T)$ where $a_i = \left\lfloor \log_{1 + \frac{\epsilon}{3}} \left(\frac{1}{v^*} \|x\|_{\mathbf{1}_{c_i}}\right) \right\rfloor$.

    First, we show that this assignment is valid, i.e., each $a_i \in [0, 2T]$. Indeed,
    \[
        \frac{1}{v^*} \|x\|_{\mathbf{1}_{c_i}} \le \frac{1}{v^*} c_i \|x\|_\infty \le \frac{d \|x\|_\infty}{v^*} \le d^2.
    \]
    The final inequality follows since $x \in \ov{\domain}$. Therefore, $a_i \le \log_{1 + \frac{\epsilon}{3} d^2} \le 2T$. Next, we claim that for any $x, y \in B(a_1, \dots, a_d)$, $x \preceq (1 + \epsilon) y$. Fix any $k \in [d]$, and let $i \in [0, T]$ such that $c_i \le k < c_{i + 1}$. Note that by definition of $a_i$, we have $a_i \le \log_{1 + \frac{\epsilon}{3}} \left(\frac{1}{v^*} \|x\|_{\mathbf{1}_{c_i}}\right) \le a_i + 1$, and the same inequality also holds for $y$. Then,
    \begin{align*}
        \|x\|_{\mathbf{1}_k} &\le \frac{k}{c_i} \|x\|_{\mathbf{1}_{c_i}} \le \frac{k}{c_i} \left(v_i \left(1 + \epsilon/3\right)^{a_i + 1}\right) \\
        &= \frac{k (1 + \epsilon/3)}{c_i} \left(v_i \left(1 + \epsilon/3\right)^{a_i + 1}\right) \\
        &\le \frac{k (1 + \epsilon/3)}{c_i} \|y\|_{\mathbf{1}_{c_i}} \le \frac{k (1 + \epsilon/3)}{c_i} \|y\|_{\mathbf{1}_{k}}.
    \end{align*}

    Finally, $\frac{k}{c_i} \le \frac{c_{i + 1} - 1}{c_i} \le (1 + \epsilon/3)$, so that $\frac{\|x\|_{\mathbf{1}_k}}{\|y\|_{\mathbf{1}_k}} \le (1 + \epsilon/3)^2 = 1 + \frac{2}{3}\epsilon + \frac{1}{9} \epsilon^2 \le 1 + \epsilon$ for all $\epsilon \in (0, 1]$.
\end{proof}

Given a vector $v \in \R^d$, we denote $\sigma v = (v_1, v_1 + v_2, \dots, v_1 + \dots + v_d)$ and $\Delta v = (v_1 - v_2, v_2 - v_3, \dots, v_d)$ for brevity. Note that with this notation, we have the top-$k$ norm $\|v\|_{\mathbf{1}_k} = (\sigma |v|^\da)_k$. Further, we have $v^\top u = (\sigma v)^\top (\Delta u)$ for all vectors $u, v \in \R^d$.

\begin{proof}[Proof of Lemma \ref{lem: dual-ordered-norm}]
    Let $K = \{x \in \R^d: \|x\|_{(w)} \le 1\}$ be the unit norm ball for $\|\cdot\|_{(w)}$, and let $K^* = \{y \in \R^d: y^\top x \le 1 \: \forall \: x \in K\}$ be the unit norm ball of its dual norm. Also denote $\ov{K} = \left\{y \in \R^d: \max_{k \in [d]} \frac{\|y\|_{\mathbf{1}_k}}{\|w\|_{\mathbf{1}_k}} \le 1\right\}$. We will show that $\ov{K} = K^*$.

    Suppose $y \in \ov{K}$. For any $x \in K$, we have
    \begin{align*}
        y^\top x &\le (|y|^\da)^\top |x|^\da & (\text{rearrangment inequality}) \\
        &= (\sigma |y|^\da)^\top (\Delta |x|^\da) & (\text{alternating sum}) \\
        &\le (\sigma w)^\top (\Delta |x|^\da) & (y \in \ov{K}) \\
        &= \|x\|_{(w)} & (\text{alternating sum}) \\
        &\le 1. & (x \in K)
    \end{align*}
    That is, $y \in K^*$. Conversely, assume $y \in K^*$ so that $y^\top x \le 1$ for each $x \in K$. Since $K^*$ is symmetric, assume without loss of generality that $y_1 \ge \dots \ge y_d \ge 0$, other cases are handled similarly. It is easy to check that for each $k \in [d]$, $x(k) := \frac{1}{(\sigma w)_k} (\underbrace{1, \dots, 1}_k, 0, \dots, 0)$ is in $K$. Therefore $1 \ge y^\top x(k) = \frac{(\sigma y)_k}{(\sigma w)_k} = \frac{(\sigma |y|^\da)_k}{(\sigma w)_k}$, implying that $y \in \ov{K}$.
\end{proof}

\begin{proof}[Proof of Lemma \ref{lem: ordered-cauchy-schwarz}]
    This proof is similar to the previous proof. For any $x, y \in \R^d$, we have
    \begin{align*}
        y^\top x &\le (|y|^\da)^\top |x|^\da & (\text{rearrangment inequality)} \\
        &= \sum_{k \in [d]} (\sigma |y|^\da)_k (\Delta |x|^\da)_k & (\text{alternating sum}) \\
        &\le \|y\|_{(w)}^* \sum_{k \in [d]} (\sigma w)_k (\Delta |x|^\da)_k & (\text{definition of } \|y\|_{(w)}^*) \\
        &= \|y\|_{(w)}^* \|x\|_{(w)} & (\text{alternating sum}).
    \end{align*}
    Further, the first inequality holds if and only if $x, y$ are order-consistent, i.e., if and only if there exists an order $\pi$ such that $x^\da = x_{\pi}$ and $y^\da = y_{\pi}$. The second inequality holds if and only if for each $k$, $(\sigma |y|^\da)_k (\Delta |x|^\da)_k = \|y\|_{(w)}^* (\sigma w)_k (\Delta |x|^\da)_k$, which happens if and only if $(\Delta |x|^\da)_k = 0$ or $\frac{(\sigma |y|^\da)_k}{(\sigma w)_k} = \|y\|_{(w)}^*$.
\end{proof}

\paragraph{Lower bound on portfolio sizes} We prove the following theorem that lower bounds the portfolio sizes for ordered and symmetric monotonic norms:

\begin{theorem}\label{thm: portfolio-size-lower-bound-smn}
There exist polytopes $\domain$ such that any $O(\log d)$-approximate portfolios for $\ordered$ must have size $d^{\Omega(1/\log\log d)}$. The same bound holds for $\smn$.
\end{theorem}

Since $\ordered \subseteq \smn$, it is sufficient to prove the result for ordered norms. First, we need a counting lemma:

\begin{lemma}
    Given $L \ge 1$, Let $T$ be the set of integral sequences $a = (a_0, \dots, a_L)$ such that $a_{i - 1} \le a_i \le a_{i - 1} + 1$ for all $i \in [L]$ and $a_0 = 0$. Then there exists a subset $\ov{T} \subseteq T$ such that (1) $|\ov{T}| \ge 2^L/(2L^2)$, and (2) for any two sequences $a, a' \in \ov{T}$, there exists an $i$ such that $a_i' < a_i$, and vice-versa.
\end{lemma}

\begin{proof}
    We first show that $|T| = 2^L$. For any such sequence $a$, consider $\phi(a) = (a_1 - a_0, \dots, a_L - a_{L - 1})$. Then $\phi(a)$ maps sequences in $T$ to binary sequences $(b_1, \dots, b_L)$; further, $\phi$ is bijective. Therefore, $|T|$ is the number of binary sequences $(b_1, \dots, b_L)$, which is $2^L$.

    Also note that $\ge$ is a partial order on $T$: $a \le a'$ if and only if $a'_i \ge a_i$ for all $i \in [0, L]$. For any distinct $a, a'$ such that $a' \ge a$, we must have that $\sum_{i \in [L]} a'_i \ge 1 + \sum_{i \in [L]} a_i$. Further, $\sum_{i \in [L]} a_i \le L^2$ for all $a \in T$. Therefore, the length of any chain in order $\ge$ on $T$ is at most $L^2 + 1$. This means that any chain decomposition of $\ge$ on $T$ must have at least $|T|/(L^2 + 1) \ge 2^L/(2 L^2)$ chains. By Dilworth's theorem, this is also the size of the largest antichain. But an anti-chain is exactly the set $\ov{T}$ we are looking for.
\end{proof}

We are ready to prove the theorem:

\begin{proof}[Proof of theorem]
    Let $S = \log^3 d$, and let $L$ be such that $S^0 + S^1 + \dots + S^{L} = d$. Then $L = \Theta\left(\frac{\log d}{\log S}\right) = \Theta\left(\frac{\log d}{\log\log d}\right)$, or that $S/L = \Omega(\log^2 d)$.

    Let $\ov{T}$ be the set of integral sequences from the previous lemma, i.e., each sequence $a = (a_0, \dots, a_L)$ is such that $a_{i - 1} \le a_i \le a_{i - 1} + 1$ for all $i \in [L]$ and $a_0 = 0$, and for any two sequences $a, a' \in \ov{T}$, there exists $i$ such that $a_i' < a_i$. Define
    \begin{align*}
        x(a) &= \Big(\underbrace{S^{-a_0}}_{S^0}, \underbrace{S^{-a_1}, \dots, S^{-a_1}}_{S^{1}}, \dots, \underbrace{S^{-a_L}, \dots, S^{-a_L}}_{S^L}\Big).
    \end{align*}
    Note that since $a_i \ge a_{i - 1}$, $x^\da = x$. Further, since $ a_i \le a_{i - 1} + 1$, we have $ a_i - i \le a_{i - 1} - (i - 1)$. Define
    \[
        w(a) = \Big(\underbrace{S^{a_0 - 0}}_{S^0}, \underbrace{S^{a_1 -1}, \dots, S^{a_1 - 1}}_{S^{1}}, \dots, \underbrace{S^{a_L - L}, \dots, S^{a_L - L}}_{S^L}\Big).
    \]
    Then
    \[
        \|x(a)\|_{(w(a))} = x(a)^\top w(a) = \sum_{i \in [0, L]} S^{-a_i} S^{a_i - i} S^i = L.
    \]
    Further, for any other $a' = (a_0', \dots, a'_L) \in \ov{T}$, there exists $i$ such that $a'_i < a_i$, we get
    \[
        \|x(a')\|_{(w(a))} \ge S^{-a_i'} S^{a_i - i} S^i > S.
    \]
    Since $S/L = \Omega(\log ^2d)$, this means that $x(a')$ is an $\omega(\log d)$-approximation for $\|\cdot\|_{(w(a))}$. That is, any $O(\log d)$-approximate portfolio for $\overline{T}$ for ordered norms must have size $|\ov{T}| \ge 2^L/(2L^2)$. However,
    \[
        \frac{2^L}{2L^2} = 2^{\Theta((\log d)/\log\log d)} \Theta\left(\frac{(\log\log d)^2}{(\log d)^2}\right) = d^{\Theta\left(\frac{1}{\log\log d}\right) - O\left(\frac{\log\log d}{\log d}\right)} = d^{\Omega\left(\frac{1}{\log\log d}\right)}.
    \]

    To prove the second part of the theorem, we claim that in fact even for $\conv(\ov{T})$, we have any $O(\log d)$ portfolio for ordered norms must have size $\ge |\ov{T}| = d^{\Omega(1/\log\log d)}$. Let $x = \sum_{b \in \ov{T}} \lambda_b x(b) \in \conv(\ov{T})$. Fix $a \in \ov{T}$. We will show that for all $x$ such that $1 - \lambda_a > 1/4$, $\|x\|_{w(a)} = \Theta(S/L) \|x(a)\|_{w(a)}$. That is, the any $O(\log d)$-approximate minimizer $x$ of $\|\cdot\|_{w(a)}$ in $\conv(\ov{T})$ must have $\lambda_a \ge \frac{3}{4}$, implying the claim.

    First, note that for each $b$, $x(b)^\da = x(b)$. Therefore,
    \begin{align*}
        \|x\|_{w(a)} &= \bigg(\sum_{b \in \ov{T}} \lambda_b x(b) \bigg)^\top w(a) = \sum_{b \in \ov{T}} \lambda_b \|x(b)\|_{w(a)} \\
        &= \lambda_a \|x(a)\|_{w(a)} + \sum_{b \neq a} \lambda_b \|x(b)\|_{w(a)} \\
        &\ge \lambda_a L + S \sum_{b \neq a} \lambda_b \ge S(1 - \lambda_a) \ge S/4.
    \end{align*}
    Where the last inequality follows by the assumption that $1 - \lambda_a \ge 1/4$. Therefore, $\|x\|_{w(a)} = \Theta(S/L) = \omega(\log d)$. This finishes the proof.
\end{proof}

\section{Omitted proofs from Section \ref{sec: covering-polyhedra}}\label{sec: missing-proofs-2}

\begin{proof}[Proof of Lemma \ref{lem: sparsification-approximation}]
    For each $i \in [r], j \in [d]$, by construction we have $\widetilde{A}_{i, j} \le A_{i, j}$, so that if $x \in \mc{P}$, then $Ax \ge \widetilde{A}x \ge \mathbf{1}_r$, i.e., $\widetilde{\mc{P}} \subseteq \mc{P}$.

    We claim that for all $x \in \mc{P}$ there is some $\widetilde{x} \in \widetilde{\mc{P}}$ such that $\widetilde{x} \preceq (1 + \epsilon) x$. From Lemma \ref{lem: majorization-norm-order-2}, this claim implies that $\|\widetilde{x}\|_f \le (1 + \epsilon) \|x\|_f$ for all norms $\|\cdot\|_f \in \smn$, and therefore that $\min_{\widetilde{x} \in \widetilde{\mc{P}}} \|\widetilde{x}\|_f \le (1 + \epsilon) \min_{x \in \mc{P}} \|x\|_f$. This implies the lemma.

    Define $\widetilde{x} = \left(1 + \frac{\epsilon}{2}\right) \left(x + \frac{\epsilon \|x\|_1}{3d} (1, \dots, 1)\right)$. We have for all $k \in [d]$
    \[
        \|\widetilde{x}\|_{\mathbf{1}_k} = \left(1 + \frac{\epsilon}{2}\right) \left(\|x\|_{\mathbf{1}_k} + \frac{ k\epsilon \|x\|_1}{3d}\right).
    \]
    However, $\frac{\|x\|_1}{d} \le \frac{\|x\|_{\mathbf{1}_k}}{k}$, so that the above gives us
    \begin{align*}
        \|\widetilde{x}\|_{\mathbf{1}_k} \le \left(1 + \frac{\epsilon}{2}\right) \left(\|x\|_{\mathbf{1}_k} + \frac{\epsilon \|x\|_{\mathbf{1}_k}}{3}\right) = \left(1 + \frac{\epsilon}{2}\right)\left(1 + \frac{\epsilon}{3}\right) \|x\|_{\mathbf{1}_k}.
    \end{align*}
    For all $\epsilon \in (0, 1]$, $\left(1 + \frac{\epsilon}{2}\right)\left(1 + \frac{\epsilon}{3}\right) \le 1 + \epsilon$, so that $\widetilde{x} \preceq (1 + \epsilon) x$. Next, we show that $\widetilde{x} \in \widetilde{\mc{P}}$. Clearly, $\widetilde{x} \ge x \ge 0$; it remains to show that $\widetilde{A}\widetilde{x} \ge \mathbf{1}_r$.

    Fix $i \in [r]$; denote the $i$th rows of $A, \widetilde{A}$ respectively by $A_i, \widetilde{A}_i$. From the algorithm, for $j \not\in B(i)$, we have $\widetilde{A}_{i, j} \ge \frac{1}{1 + \frac{\epsilon}{2}} A_{i, j}$. Therefore,
    \begin{align*}
        \widetilde{A}_i^\top \widetilde{x} &= \sum_{j \in [d]} \widetilde{A}_{i, j} \widetilde{x}_j = \sum_{j \not\in B(i)} \widetilde{A}_{i, j} \widetilde{x}_j & (\widetilde{A}_{i, j} = 0 \: \forall \: j \in B(i)), \\
        &\ge \frac{1}{1 + \frac{\epsilon}{2}} \sum_{j \not\in B(i)} A_{i, j} \widetilde{x}_j \\
        &= \frac{1}{1 + \frac{\epsilon}{2}} \left(\sum_{j \not\in B(i)} A_{i, j} \left(1 + \frac{\epsilon}{2}\right)\left( x_j + \frac{\epsilon \|x\|_1}{3d}\right)\right) \\
        &= \sum_{j \not\in B(i)} A_{i, j} x_j + \frac{\epsilon \|x\|_1}{3d} \sum_{j \not\in B(i)} A_{i, j}.
    \end{align*}
    Now, $\sum_{j \not\in B(i)} A_{i, j} \ge a_i^* \ge \frac{\mu}{d} \sum_{j \in B(i)} A_{i, j} = \frac{3d}{\epsilon} \sum_{j \in B(i)} A_{i, j}$. Therefore,
    \[
        \frac{\epsilon \|x\|_1}{3d} \sum_{j \not\in B(i)} A_{i, j} \ge  \frac{\epsilon \|x\|_1}{3d} \cdot \frac{3d}{\epsilon} \cdot \sum_{j \in B(i)} A_{i, j} \ge \sum_{j \in B(i)} A_{i, j} x_j.
    \]
    Together, this means that $\widetilde{A}_i^\top \widetilde{x} \ge A_i^\top x \ge 1$. Since this holds for all $i \in [r]$, $\widetilde{x} \in \mc{P}$.
\end{proof}

\subsection{Proof of Lemma \ref{lem: primal-dual-counting}}

We restate the relevant convex programs and the lemma here for convenience:

\vspace{-20pt}
\begin{multicols}{2}
    \begin{equation}
        \min \|x\|_{(w)} \quad \text{s.t.} \quad Ax \ge \mathbf{1}_r, x \in \mc{D}. \tag{primal'}
    \end{equation}

    \begin{equation}
        \min \|A^\top \lambda\|_{(w)}^* \quad \text{s.t.} \quad  \lambda \in \simplex{r} \tag{dual}
    \end{equation}
\end{multicols}

\primaldualcounting*

For $j \in [d]$, denote the $j$th column of $A$ as $A^{(j)}$. $A^{(j)}$ is an $r$-dimensional vector. Recall that $S_1, \dots, S_{N^r}$ form a partition of $[d]$ such that for $l \in [N^r]$, and for all $j, j' \in S_l$, $A^{(j)} = A^{(j')}$. Also recall that $\mc{D}$ is the set of all vectors $x \ge 0$ such that $x_j = x_{j'}$ for all $j, j' \in S_l$, for all $l \in [N^r]$. From Lemma \ref{lem: P=-is-optimal}, $x(w) \in \mc{D}$.

First, for all $x \in \mc{P}$ and $\lambda \in \simplex{r}$, we get by ordered Cauchy-Schwarz \ref{lem: ordered-cauchy-schwarz} that $\|x\|_{(w)} \|A^\top \lambda \|_{(w)}^* \ge \lambda^\top A w$. Since $x \in \mc{P}$, $Ax \ge \mathbf{1}_r$, and since $\lambda \in \simplex{r}$, $\lambda^\top A x \ge 1$. Now, suppose that there is some $\lambda \in \simplex{r}$ such that $\|x(w)\|_{(w)} \|A^\top \lambda\|_{(w)}^* = 1$, i.e. equality holds. Then, since $\lambda(w) = {\arg\min}_{\lambda \in \simplex{r}}\|\lambda\|_{(w)}^*$, we get that
\[
    1 = \|x(w)\|_{(w)} \|A^\top \lambda\|_{(w)}^* \ge \|x(w)\|_{(w)} \|A^\top \lambda(w)\|_{(w)}^* \ge 1.
\]
Then equality must hold everywhere, and in particular $\|x(w)\|_{(w)} \|A^\top \lambda(w)\|_{(w)}^* = 1$. Further, from ordered Cauchy-Schwarz, it is necessary that $x(w), A^\top \lambda(w)$ satisfy some order $\pi \in \text{Perm}(d)$.

From Lemma \ref{lem: P=-is-optimal}, $x(w) \in \mc{D}$, i.e., for all $j, j' \in S_l$, for all $l \in [N^r]$, $x(w)_j = x(w)_{j'}$. Similarly, $(A^\top \lambda(w))_j$ is the dot product of the $j$th column of $A$ with $\lambda(w)$, and therefore $A^\top \lambda(w) \in \mc{D}$ as well. Since $x, A^\top \lambda(w)$ both satisfy order $\pi$, $\pi$ must induce a reduced order $\rho$ on $S_1, \dots, S_{N^r}$. This implies the lemma.

It remains to prove that there exists $\lambda$ such that $\|x(w)\|_{(w)} \|A^\top \lambda\|_{(w)}^* = 1$. Our proof is along the lines of the proof of strong duality using Slater's conditions \cite{boyd_convex_2004}, although we use the properties of ordered norms at several places. We will need the following two lemmas:

\begin{lemma}\label{lem: dual-norm-thresholding}
For vector $y \in \R^d$ such that $y_1 \ge \dots \ge y_{d} \ge 0$, let $t_1 \le t_2 \le \dots \le t_{T} = d$ be indices such that
\[
    y_1 = \dots = y_{t_1} \ge y_{t_1 + 1} = \dots = y_{t_2} \ge \dots \ge y_{t_{T - 1} + 1} = \dots = y_{t_T}.
\]
Then for any weight vector $w$, $\|y\|_{(w)}^* = \max_{k \in [d]} \frac{\|y\|_{\mathbf{1}_{k}}}{\|w\|_{\mathbf{1}_{k}}}$ is achieved at some $k \in \{t_1, \dots, t_T\}$.
\end{lemma}

\begin{proof}
    It is sufficient to show that for all $i \in [T]$ and $t_{i -1} \le k \le t_{i}$, we have
    \[
        \max\left\{\frac{\|y\|_{\mathbf{1}_{t_{i - 1}}}}{\|w\|_{\mathbf{1}_{t_{i - 1}}}}, \frac{\|y\|_{\mathbf{1}_{t_i}}}{\|w\|_{\mathbf{1}_{t_i}}}\right\} \ge \frac{\|y\|_{\mathbf{1}_k}}{\|w\|_{\mathbf{1}_k}}.
    \]
    Denote $z = y_{t_{i - 1} + 1} = \dots = y_{t_i}$. Consider $(1 - \lambda) \|y\|_{\mathbf{1}_{t_{i - 1}}} + \lambda \|y\|_{\mathbf{1}_{t_{i}}}$ for $\lambda = \frac{k - t_{i - 1}}{t_i - t_{i - 1}}$. Then $\lambda \in [0, 1]$, and
    \[
        (1 - \lambda)\|y\|_{\mathbf{1}_{t_{i - 1}}} + \lambda \|y\|_{\mathbf{1}_{t_{i}}} = \|y\|_{\mathbf{1}_{t_{i - 1}}} + \lambda z (t_i - t_{i - 1}) =  \|y\|_{\mathbf{1}_{t_{i - 1}}} + (k - t_{i - 1}) z = \|y\|_{\mathbf{1}_{k}}.
    \]
    Further,
    \begin{align*}
    (1 - \lambda) \|w\|_{\mathbf{1}_{t_{i - 1}}} + \lambda \|w\|_{\mathbf{1}_{t_{i}}} &= \|w\|_{\mathbf{1}_{t_{i - 1}}} + \lambda (w_{t_{i - 1} + 1} + \dots + w_{t_i}) \\
    &=  \|w\|_{\mathbf{1}_{t_{i - 1}}} + (k - t_{i - 1}) \frac{w_{t_{i - 1} + 1} + \dots + w_{t_i}}{t_i - t_{i - 1}}.
    \end{align*}
    Since $w_{t_{i - 1} + 1} \ge \dots \ge w_{t_i}$, we get that
    \[
        \frac{w_{t_{i - 1} + 1} + \dots + w_{t_i}}{t_i - t_{i - 1}} \le \frac{w_{t_{i - 1} + 1} + \dots + w_k}{k - t_{i - 1}}.
    \]
    Plugging this back in, we get $(1 - \lambda) \|w\|_{\mathbf{1}_{t_{i - 1}}} + \lambda \|y\|_{\mathbf{1}_{t_{i}}} \le \|w\|_{\mathbf{1}_{k}}$. Therefore,
    \begin{align*}
        \frac{\|y\|_{\mathbf{1}_k}}{\|w\|_{\mathbf{1}_k}} \le \frac{(1 - \lambda)  \|y\|_{\mathbf{1}_{t_{i - 1}}} + \lambda  \|y\|_{\mathbf{1}_{t_i}}}{(1 - \lambda)  \|w\|_{\mathbf{1}_{t_{i - 1}}} + \lambda  \|w\|_{\mathbf{1}_{t_i}}} \le \max\left\{\frac{ \|y\|_{\mathbf{1}_{t_{i - 1}}}}{ \|w\|_{\mathbf{1}_{t_{i - 1}}}}, \frac{ \|y\|_{\mathbf{1}_{t_i}}}{ \|w\|_{\mathbf{1}_{t_i}}}\right\}.
    \end{align*}
\end{proof}

\begin{lemma}\label{lem: dual-sup}
For $\mu \in \RO^r$,
\[
    \sup_{x \in \mc{D}} \mu^\top A x - \|x\|_{(w)} = \begin{cases}
                                                         0 & \text{if} \: \|\mu^\top A\|_{(w)}^* \le 1, \\
                                                         \infty & \text{otherwise}.
    \end{cases}
\]
\end{lemma}

\begin{proof}
    Denote $y = A^\top \mu$. Then $y \in \R^d$, and $y_j = (A^{(j)})^\top \mu$. If $\|y\|_{(w)}^* \le 1$, we get from Lemma \ref{lem: ordered-cauchy-schwarz} (ordered Cauchy-Schwarz) that
    \begin{align*}
        y^\top x - \|x\|_{(w)} \le \|y\|_{(w)}^* \|x\|_{(w)} - \|x\|_{(w)} \le (\|y\|_{(w)}^* - 1) \|x\|_{(w)} \le 0.
    \end{align*}
    However, $0 \in \mc{D}$, and therefore for $x = 0$, $y^\top x - \|x\|_{(w)} = 0$, so that $\sup_{x \in \mc{D}} y^\top x - \|x\|_{(w)} = 0$

    Now suppose that $\|y\|_{(w)}^* \ge 1$. Note that since $y_j = (A^{(j)})^\top \mu$, for all $j, j' \in S_l$ for some $l$, we get $y_j = y_{j'}$.

    Relabel the indices $[N^r]$ so that for all $j \in S_l$ and $j' \in S_{l + 1}$, $y_j \ge y_{j'}$. Further, relabel indices $[d]$ so that $S_1 = \{1, \dots, |S_1|\}$, $S_2 = \{|S_1| + 1, \dots, |S_1| + |S_2|\}$ etc, i.e.,
    \[
        y_1 = \dots = y_{|S_1|} \ge y_{|S_1| + 1} = \dots = y_{|S_1| + |S_2|} \ge \dots \ge y_{d - |S_{N^r}| + 1} = \dots = y_{d} \ge 0.
    \]
    By the previous lemma $\|y\|_{(w)}^* = \max_{k \in [d]} \frac{\|y\|_{\mathbf{1}_k}}{\|w\|_{\mathbf{1}_k}}$ achieved at some $k = |S_1| + \dots + |S_l|$. For brevity, denote this number as $k^*$.

    Define $x$ such that $x_1 = x_2 = \dots = x_{k^*} = \frac{\alpha}{\|w\|_{\mathbf{1}_{k^*}}}$ and $x_{k^* + 1} = \dots = x_d = 0$ where $\alpha$ is an arbitrarily large number. Then $x \in \mc{D}$ and $\|x\|_{(w)} = \alpha$. Further,
    \[
        y^\top x = \|y\|_{\mathbf{1}_{k^*}} \frac{\alpha}{\|w\|_{\mathbf{1}_{k^*}}}.
    \]
    Since $\frac{\|y\|_{\mathbf{1}_{k^*}}}{\|w\|_{\mathbf{1}_{k^*}}} = \|y\|_{(w)}^* > 1$, we get that $y^\top x - \|x\|_{(w)} = \alpha \left(\frac{\|y\|_{\mathbf{1}_{k^*}}}{\|w\|_{\mathbf{1}_{k^*}}} - 1\right)$, which can be arbitrarily large as $\alpha$ grows. This proves the second case as well.
\end{proof}

We proceed to prove that there exists $\lambda$ such that $\|x(w)\|_{(w)} \|A^\top \lambda\|_{(w)}^* = 1$. Let $\mc{A}$ be the set of points $(v_1, \dots, v_r, t)$ such that there exists an $x \in \mc{D}$ with $v_i \ge 1 - A_i^\top x$ for all $i \in [r]$ and $t \ge \|x\|_{(w)}$. It is easy to check that $\mc{A}$ is convex. Next, define $\mc{B} = \{(\underbrace{0, \dots, 0}_{r}, s): s < \|x(w)\|_{(w)}\}$. Clearly, $\mc{B}$ is convex. It is easy to see that $\mc{A} \cap \mc{B} = \emptyset$. Therefore, there is a separating hyperplane between $\mc{A}, \mc{B}$, i.e. there exist $\mu \in \R^d, \delta, \alpha \in \R$ such that
\begin{align}
    \mu^\top v + \delta t &\ge \alpha
    \: \forall \: (v, t) \in \mc{A}, \label{eqn: sep-hyp-1} \\
    \delta s &< \alpha \: \forall s < \|x(w)\|_{(w)}. \label{eqn: sep-hyp-2}
\end{align}

The second equation implies that $\delta \ge 0$ since otherwise we can choose $s$ to be arbitrarily small and $\delta s$ becomes arbitrarily large. Then, we get $\delta \|x(w)\|_{(w)} \le \alpha$.

Further, by a similar argument, $\mu \ge 0$. Applying eqn. (\ref{eqn: sep-hyp-1}), to point $(1 - A_1^\top x, \dots, 1 - A_r^\top x, \|x\|_{(w)}) \in \mc{A}$ that
for all $x \in \mc{D}$, $\sum_{i \in [r]} \mu_i - \mu^\top A x + \delta \|x\|_{(w)} \ge \alpha \ge \delta \|x(w)\|_{(w)}$.

\textbf{Case I}: $\mu = 0$. Then $\delta \|x\|_{(w)} \ge \alpha \ge \delta \|x(w)\|_{(w)}$. Since not both $\mu, \delta$ can be zero, $\delta > 0$. Further, $\|x(w)\|_{(w)} > 0$, so if we pick $x = 0 \in \mc{D}$, we get a contradiction.

\textbf{Case II}: $\mu \neq 0$, so we get that all for all $x \in \mc{D}$, $\sum_{i \in [r]} \mu_i - \mu^\top A x + \delta \|x\|_{(w)} \ge \alpha \ge \delta \|x(w)\|_{(w)}$. If $\delta = 0$, then $\sum_i \mu_i - \mu^\top A x \ge 0$ for all $x \in \mc{D}$. Pick arbitrarily large $x$ again, giving a contradiction. Therefore, $\delta > 0$; assume without loss of generality that it is $1$.

That is, for all $x \in \mc{D}$, $\sum_i \mu_i - \mu^\top A x + \|x\|_{(w)} \ge \|x(w)\|_{(w)}$. Taking infimum on the left-hand side and applying Lemma \ref{lem: dual-sup}, we get that $\sum_i \mu_i \ge \|x(w)\|_{(w)}$ with $\|\mu^\top A\|_{(w)}^* \le 1$. Then $\lambda := \frac{\mu}{\sum_i \mu_i} \in \simplex{r}$. Therefore,
\[
    1 \ge \|\mu^\top A\|_{(w)}^* = \sum_i \mu_i \|\lambda^\top A\|_{(w)}^* \ge \|x(w)\|_{(w)} \|\lambda^\top A\|_{(w)}^*. \qed
\]

\section{Omitted proofs from Section \ref{sec: iterative-ordering}}\label{sec: iterative-ordering-appendix}

We complete the proof of various lemmas in Section \ref{sec: iterative-ordering} on \iterativeordering.

\begin{proof}[Proof of Lemma \ref{lem: ordered-satisfaction-problems}]
    The proof for \completiontimes was supplied in the main body. Further, \ovc is a special case of \osc. Therefore, it suffices to complete the proof for \osc and \orderedtsp.

    - \osc. Given the ground set $E = \{e_1, \dots, e_n\}$ and subsets $S_1, \dots, S_m \subseteq E$, choose the set of clients $C$ as the ground set $E$, the set of objects $\mc{X}$ as the set $\{S_1, \dots, S_m\}$ of subsets, with $C(S_i) = S_i$ for all $i$. Given any satisfier $X \subseteq \mc{X}$ and order $\pi$ on $X$, define the time $t(X, \pi)_{S_i}$ for $S_i \in X$ to be the position $\pi(S_i)$ of $S_i$ in $\pi$. Then the satisfaction time of an element $e \in C(X) = \bigcup_{S \in X} S$ is precisely the cover time of the element.

    Further, given a $T > 0$, $X_T$ is simply the first $T$ subsets in $X$ according to order $\pi$, and clearly, $t(X_T, \pi_T)_{S_i} = t(X, \pi)_{S_i} = \pi(S_i)$ for all such subsets $S_i \in X_T$. This proves downward closure.

    - \orderedtsp. Given a metric on vertices $V$ and starting vertex $v_0$, choose the set of clients $C$ as $V$, and the set of objects $\mc{X}$ as $V$ also, with $C(v) = \{v\}$ for all $v \in \mc{X}$. Given any satisfier $X \subseteq \mc{X}$, any order $\pi$ on $X$ corresponds to a path consisting of the vertices of $X$. We define the time $t(X, \pi)_v$ as follows: if $\pi$ does not start at $v_0$ or if $v_0 \not\in X$, then $t(X, \pi)_v = \infty$ for all $v \in X$. This is to disallow paths that do not start at the starting vertex $v_0$. If $\pi$ starts at $v_0$, then define $t(X, \pi)_v$ to be the length of the path from $v_0$ to $v$. Since $C(v) = \{v\}$ for all $v \in V$, the satisfaction time of a vertex $v \in X$ is the same as $t(X, \pi)_v$.

    We prove downward closure next: given a $T > 0$ and $(X, \pi)$, if $\pi$ does not start at $v_0$, then $X_T = \emptyset$, and downward closure holds trivially. Otherwise, $X_T$ is precisely the set of vertices within distance $T$ of the starting vertex $v_0$ along path $\pi$, and $t(X_T, \pi_T)_v = t(X, \pi)_v$ for all $v \in X_T$.
\end{proof}

\begin{proof}[Proof of Lemma \ref{lem: composable-problems}]
    As before, it suffices to complete the proof for \osc and \orderedtsp.

    - \osc. Given satisfier $X$ of subsets of the ground set and order $\pi$ on $X$, the cost $c(X, \pi) = \max_{S_i \in X} \pi(S_i) = |X|$ is simply the size of $X$. Consider satisfiers $X_1, \dots, X_k$, corresponding orders $\pi_1, \dots, \pi_k$, and some $S \in X_j \setminus (X_1 \cup \dots \cup X_{j - 1})$. For the composed satisfier $X = \bigcup_{l \in [k]} X_l$, and the composed order $\pi = \bigoplus_{l \in [k]} \pi_l$, we have that $t(X, \pi)_S = \pi(S)$ is the position of $S$ in the order when all subsets in $X_1$ are ordered first, all subsets in $X_2 \setminus X_1$ are ordered next, and so on. Therefore,
    \[
        t(X, \pi)_S \le |X_1| + \dots + |X_{j - 1}| + \pi_j(S) = c(X_1, \pi_1) + \dots + c(X_{j - 1}, \pi_{j - 1}) + t(X_j, \pi_j)_S.
    \]

    \orderedtsp. Given satisfier $X \subseteq V$ and path $\pi$ on $X$, the cost $c(X, \pi) = \infty$ if the path does not start at $v_0$ and $c(X, \pi)$ is the length of the path otherwise. Composing paths $\pi_1, \dots, \pi_k$ on vertex sets $X_1, \dots, X_k$ respectively that each starts at $v_0$ amounts to the following: start at $v_0$, complete path $\pi_1$, and return to $v_0$, the complete path $\pi_2$ and return to $v_0$ again, and so on, shortcutting any vertices visited a second time.

    Then, given a vertex $v$ visited in path $\pi_j$, the length of the path from $v_0$ to $v$ in this composed path is at most $2(\text{length}(\pi_1) + \dots + \text{length}(\pi_{j - 1})) + \text{length from } v_0 \text{ to } v \text{ in } \pi_j$, which is precisely $2\left(\sum_{l \in [j - 1]} c(X_l, \pi_l)\right) + t(X_j, \pi_j)_v$.
\end{proof}

\begin{proof}[Proof of Lemma \ref{lem: restriction-increases-cardinality}]
    For part 1, by definition, $c(X_T, \pi_T) = \max_{x \in X_T} t(X_T, \pi_T)_x$. By downward closure, $t(X_T, \pi_T)_x \le t(X, \pi)_x$ for all $x \in X_T$. However, $X_T$ was defined as $\{x \in X: t(X, \pi)_x \le T\}$, and thus $c(X_T, \pi_T) \le T$.

    Part 2: for each client $e$ satisfied within time $T$ by $(X, \pi)$, by definition of satisfaction time there is some object $x \in X$ with $t(X, \pi)_x \le T$. Therefore, $x \in X_T$ and so $e \in C(X_T)$, i.e., $|C(X_T)|$ is at least the number of clients satisfied by $(X, \pi)$ within time $t$.
\end{proof}

\begin{proof}[Proof of Lemma \ref{lem: partial-scheduling-approximation}]
    Given processing times $p$ and budget $B \ge 0$, consider the following linear programming relaxation of the problem:
    \begin{align}
        \max &\sum_{i, j} x_{i, j} & \text{s.t.} \notag \label{lp: partial scheduling}\tag{LP-PS} \\
        \sum_j p_{i, j} x_{i, j} &\le B & \forall \: i \in [d], \label{const: scheduling-makespan-within-budget}\\
        \sum_i x_{i, j} &\le 1 & \forall \: j \in [n], \label{const: scheduling-schedule-at-most-once}\\
        x_{i, j} &= 0 & \text{if} \: p_{i, j} > B \: \forall \: i, j, \label{const: scheduling-dont-schedule-if-budget-exceeded}\\
        x &\ge 0. \notag
    \end{align}
    Variable $x_{i, j}$ indicates whether or not job $j$ has been assigned to machine $i$. The objective is to maximize the number of jobs scheduled under the constraint that the makespan is at most $B$. However, to ensure that the optimal solution does not schedule a cheap job multiple times, we include constraints (\ref{const: scheduling-schedule-at-most-once}). Further, job $j$ should not be scheduled on machine $i$ if $p_{i, j}$ exceeds the makespan $B$ (constraints (\ref{const: scheduling-dont-schedule-if-budget-exceeded})). The optimal solution $\OPT$ to the partial scheduling problem clearly satisfies these constraints, and therefore $\OPT \le \sum_{i, j} x^*_{i, j}$ for the (fractional) optimal solution $x^*$ to the LP.

    We will round $x^*$ to an integral solution $x$ with makespan $\le 2 B$ and $\sum_{i, j} x_{i, j} \ge \sum_{i, j} x^*_{i, j}$ implying that $x$ schedules at least as many jobs as $\OPT$, thus completing the proof.

    Let $k_i = \lceil \sum_{j} x^*_{i, j} \rceil$ for all $i$. We will construct an undirected bipartite graph $G$ with $n + k_1 + \dots + k_d$ vertices: $n$ vertices correspond to jobs and $k_i$ vertices correspond to machine $i$ for all $i$.

    For machine $i$, let $J_i = \{j: x^*_{i, j} > 0\}$ be the set of jobs (fractionally) assigned to $i$ under $x^*$, and relabel them so that $J_i = \{1, 2, \dots, l\}$; assume without loss of generality that $p_{i, 1} \ge \dots \ge p_{i, l}$. Let $v_1, \dots, v_{k_i}$ be the vertices corresponding to machine $i$. Start assigning weights $x^*_{i, 1}, x^*_{i, 2}, \dots $ to edges $v_1 1, v_1 2, \dots$, until we reach a job $a$ such that $x^*_{i, 1} + x^*_{i, 2} + \dots + x^*_{i, a} > 1$. Assign weight $1 - \sum_{b \le a - 1} x^*_{i, b}$, i.e., just enough weight that makes the total weight of edges incident to $v_1$ exactly $1$. The remaining weight for job $a$, $\sum_{b \le a} x^*_{i, b} - 1$ goes to edge $v_2 a$. Continue this process with job $a + 1$ on vertex $v_2$, and so on. Since $\sum_{j \in J_i} x^*_{i, j} \le k_i$, weight $x^*_{i, l}$ is assigned to edge $v_{k_1} l$. Notice that for each of $v_1, \dots, v_{k_i}$, the sum of weights of edges incident on it is at most $1$. Do this for all vertices to get $G$, and denote the weights in $G$ by $w$.

    By construction, the sum of weights of edges incident on a vertex is at most $1$ if it corresponds to a machine. From constraints (\ref{const: scheduling-schedule-at-most-once}) and the construction, the sum of weights of edges incident on vertices corresponding to jobs is also at most $1$. Therefore, $w$ forms a fractional matching on $G$. Further, the sum of all edge weights, $\|w\|_1$, is $\sum_{i, j} x^*_{i, j}$. Since $G$ is bipartite, this fractional matching can be rounded to an integral matching $y$ at least as large as $w$, i.e., $\|y\|_1 \ge \|w\|_1$. Obtain integral solution $x$ by assigning jobs to machines according to matching $y$, i.e., if job $j$ is adjacent to a vertex corresponding to machine $i$, assign $x_{i, j} = 1$; assign $x_{i, j} = 0$ in all other cases. Then we have that $\sum_{i, j} x_{i, j} = \|y\|_1 \ge \|w\|_1 = \sum_{i, j} x^*_{i, j}$.

    It remains to argue that the makespan to each machine is at most $2B$. Fix machine $i$. Suppose jobs $j_1, \dots, j_{k_i}$ are adjacent to vertices $v_1, \dots, v_{k_i}$ respectively in matching $y$. Then, since jobs were sorted in decreasing order, the processing time $p_{i, j_2}$ is upper bounded by the processing time of jobs adjacent to $v_1$ in $G$:
    \[
        p_{i, j_2} = p_{i, j_2} \sum_{j: w(v_1 j) > 0} w(v_1, j) \le \sum_{j: w(v_1 j) > 0} p_{i, j} w(v_1, j).
    \]
    Similarly, for each $b \in [2, j_{k_i}]$, we get $p_{i, j_b} \le \sum_{j: w(v_{b - 1, j}) > 0} p_{i, j} w(v_{b - 1}, j)$.

    Adding these,
    \[
        \sum_{a \in [2, k_i]} p_{i, j_a} \le \sum_{a \in [2, k_i]}\sum_{j: w(v_{b - 1, j}) > 0} p_{i, j} w(v_{b - 1}, j) \le \sum_{j \in J_i} p_{i, j} x^*_{i, j} \le B.
    \]
    Since $p_{i, j_1} \le B$ by constraint (\ref{const: scheduling-makespan-within-budget}), we get the total makespan on machine $i$ under $x$ is
    \[
        p_{i, j_1} + \sum_{a \in [2, k_i]} p_{i, j_a} \le 2B. \qedhere
    \]
\end{proof}

\section{Lower bounds for simultaneous approximations}\label{sec: lower-bounds}

We give two lower bounds on best-possible simultaneous approximations here, for \ovc and \completiontimes, respectively.

\begin{observation}\label{thm: vertex-cover-lower-bound}
There exists an instance of \ovc where no solution is better than $9/8$-simultaneous approximate for the $L_1$ and $L_\infty$ norms (i.e., Min-Sum Vertex Cover and classical Vertex Cover).
\end{observation}

\begin{figure}[h]
    \centering
    \includegraphics[width=0.6\textwidth]{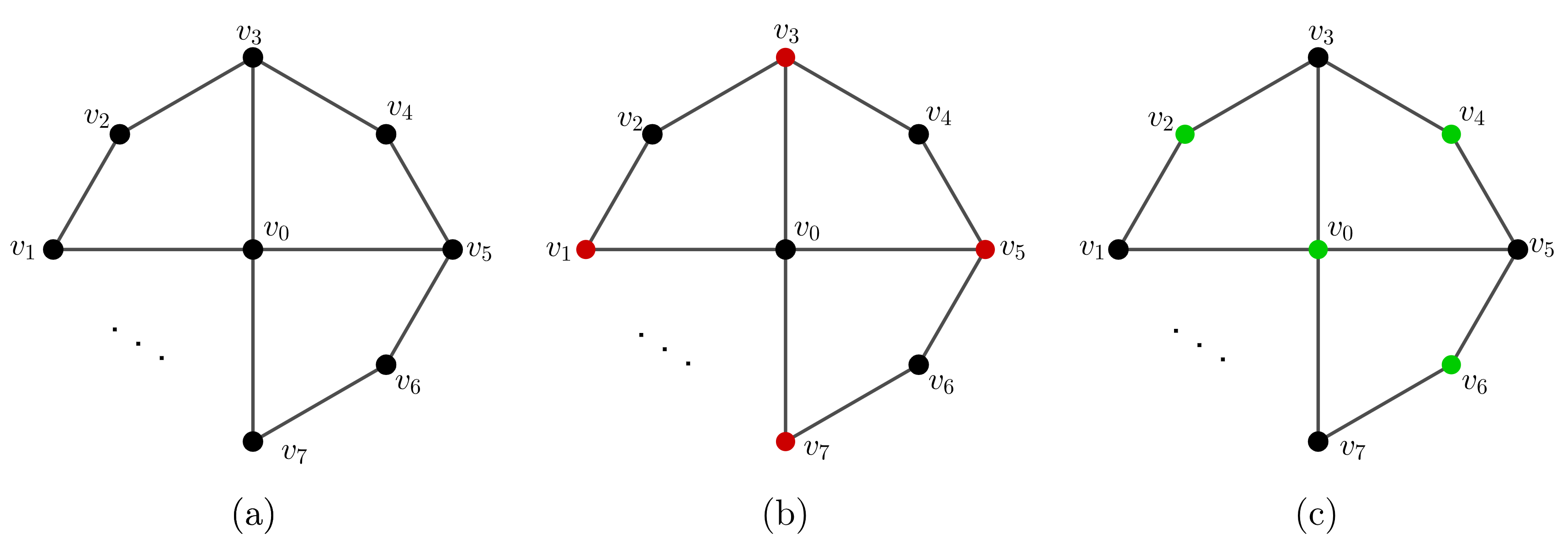}
    \caption{The vertex cover instance used in proof of Observation \ref{thm: vertex-cover-lower-bound}.}
    \label{fig: vertex-cover-lower-bound}
\end{figure}

\begin{proof}
    Consider the following instance: the graph as $2n + 1$ vertices $v_0, \ldots, v_{2n}$ with vertices $v_1, \ldots, v_{2n}$ forming cycle and vertex $v_0$ connected to each of $v_1, v_3, \ldots, v_{2n - 1}$ (Figure \ref{fig: vertex-cover-lower-bound}(a)).

    The smallest vertex cover is $\{v_1, v_3,\ldots, v_{2n - 1}\}$ (Figure \ref{fig: vertex-cover-lower-bound}(b)), and it is the only vertex cover of size $n$. Therefore, any other vertex cover is at best a $\frac{n + 1}{n}$-approximation. When $n = 8$, this is $9/8$.

    We show that this vertex cover is a $9/8$-approximation for MSVC when $n = 8$. Irrespective of the order of the vertices in this vertex cover, exactly $3$ edges are covered by each time step. Therefore, the total cover time of the edges is $3 \times (1 + \ldots + n) = \frac{3}{2}n(n + 1)$. When $n = 8$, this is $108$.

    However, if we instead use the cover $(v_0, v_1, v_3, \ldots, v_{2n - 1})$ (Figure \ref{fig: vertex-cover-lower-bound}(c)) in this order, $n$ edges are covered at the first step, and $2$ edges are covered in each subsequent step, resulting in total cover time of $n + 2 (2 + \ldots + (n + 1)) = n(n + 4)$. When $n = 8$, this is $96 = \frac{8}{9} \times 108$.
\end{proof}

Next, we show a similar bound for \completiontimes:

\begin{observation}\label{thm: job-times-lower-bound}
There exists an instance of \completiontimes where no solution is better than $1.13$-simultaneous approximate for the $L_1$ and $L_\infty$ norms (i.e., average completion time minimization and makespan minimization).
\end{observation}

\begin{proof}
    Consider an instance with two machines (labeled $A, B$) and three jobs. Let $\mu, \delta \in [0, 1)$ be parameters we fix later. Jobs $1, 2$ both have processing time $1$ on machine $A$ and processing time $1 + \delta$ on machine $B$. Job $3$ has processing time $1 + \mu$ on machine $A$ and $2$ on machine $B$.

    Consider solutions where jobs $1, 2$ are on different machines. Then, the optimal solution (for both makespan minimization and average completion time minimization) is to place job $3$ is on machine $A$. The makespan for this solution is $2 + \mu$, and the total completion time is $1 + (2 + \mu) + (1 + \delta) = 4 + \mu + \delta$:
    \begin{align*}
        MS_1 = 2 + \mu, \quad CT_1 = 4 + \mu + \delta.
    \end{align*}

    Suppose jobs $1, 2$ are both on machine $A$ now. The optimal solution (for both makespan and total completion time) is to place job $3$ is on machine $B$. The makespan and total completion time are respectively
    \begin{align*}
        MS_2 = 2, \quad CT_2 = 5.
    \end{align*}
    Suppose jobs $1, 2$ are both on machine $B$. The optimal solution is to place job $3$ is on machine $A$. The makespan and total completion time are:
    \begin{align*}
        MS_3 = 2(1 + \delta), \quad CT_3 = 3(1 + \delta) + (1 + \mu) = 4 + \mu + 3 \delta.
    \end{align*}

    Therefore, when $\mu + \delta \le 1$, the second solution has optimal makespan $2$ and the first solution has the optimal average completion time. The simultaneous approximation ratio of the first solution is $\frac{2 + \mu}{2}$. The simultaneous approximation ratio of the second solution is $\frac{5}{4 + \mu + \delta}$. The simultaneous approximation ratio of the third solution is $\max\left(1 + \delta, \frac{4 + \mu + 3 \delta}{4 + \mu + \delta}\right)$. The best possible simultaneous approximation ratio then is
    \[
        \min\left(\frac{2 + \mu}{2}, \frac{5}{4 + \mu + \delta}, \max\left(1 + \delta, \frac{4 + \mu + 3 \delta}{4 + \mu + \delta}\right)\right).
    \]
    Maximizing this over all $\mu, \delta$ such that $0 \le \mu, \delta$ and $\mu + \delta \le 1$, we get the value $\frac{\sqrt{61} - 1}{6} > 1.13$ at $(\mu, \delta) = \left(\frac{\sqrt{61} - 7}{3}, \frac{\sqrt{61} - 7}{6}\right)$.
\end{proof}

\section{\clustering and \ufl}\label{sec: clustering-and-facility-location}

In this section, we consider \clustering and \ufl. We are given a metric space $(X, \mathrm{dist})$ on $|X| = n$ points (also called \emph{clients}) and are required to choose a subset $F \subseteq X$ of \emph{open facilities}\footnote{One can also forbid opening facilities at some points in $X$; our algorithm still works in this more general setting.}. The induced distance vector $x^F \in \R^X$ is defined as the vector of distances between point $j$ and its nearest open facility, i.e., $x^F_j = \min_{f \in F} \mathrm{dist}(j, f)$ for all $j \in X$. Given a norm $\|\cdot\|_f$ on $\R^n$, \clustering seeks to open a set $F$ of at most $k$ facilities to minimize $\|x^F\|_f$, while \ufl allows any number of facilities to open but penalizes the number of open facilities through the combined objective function $|F| + \|x^F\|_f$.

For \clustering, we consider more general \emph{bicriteria} $(\alpha, \beta)$-approximations with objective value within factor $\alpha$ of the optimum but that violate the bound on the number of open facilities by a factor $\beta$. \cite{golovin_all-norms_2008} show that any solution to \clustering that is simultaneously $O(1)$-approximate for $\smn$ must open at least $\Omega(k\log n)$ facilities, i.e., violate the size bound by factor $\beta = \Omega(\log n)$.

Fix any $\epsilon \in (0, 1]$. Using ideas similar to \iterativeordering, we give the algorithm \iterativeclustering that finds a solution with at most $O\left(\frac{k\log n}{\epsilon}\right)$ open facilities that is simultaneously $(1 + \epsilon)$-approximate for all symmetric monotonic norms, matching the result of \cite{goel_simultaneous_2006}.
In \emph{polynomial-time}, \iterativeclustering finds a solution that is $(3 + \epsilon)$-approximate, improving the previous $(6 + \epsilon)$-approximation of \cite{goel_simultaneous_2006}.

We remark that -- as pointed out to us by a reviewer from STOC 2024 -- carefully combining the rounding techniques from \cite{shmoys_approximation_1993} and the linear program for top-$k$ norm minimization from \cite{chakrabarty_approximation_2019} matches our polynomial-time bound, and gives an even better $(2 + \epsilon)$-approximation if facilities are allowed to open anywhere in $X$. Our algorithm is a natural extension of the \iterativeordering framework that emphasizes common structure across different combinatorial problems.

We also show that the above result for \clustering leads to an $O(\log n)$-approximate portfolio of size $O(\log n)$ for \ufl, the first such result for symmetric monotonic norms to our knowledge.

\subsection{\clustering}\label{sec: clustering-clustering}

We prove the following result:

\begin{theorem}\label{thm: clustering}
For \clustering, Algorithm \textup{\iterativeclustering} gives
\begin{enumerate}
    \item \label{cor: clustering-1} a simultaneous bicriteria $\left(1 + \epsilon, O\left(\frac{\log n}{\epsilon}\right)\right)$-approximation in finite time, and
    \item \label{cor: clustering-2} a simultaneous bicriteria $\left( 3 + \epsilon, O\left(\frac{\log n}{\epsilon}\right)\right)$-approximation can in polynomial-time.
\end{enumerate}
\end{theorem}

Broadly, \iterativeclustering iteratively combines solutions that each contain $k$ facilities. Each of these solutions corresponds to a radius $R$, and subroutine \texttt{PartialClustering} attempts to get the set of $k$ facilities that covers the largest number of points within radius $R$. As with \iterativeordering, radius $R$ increases exponentially across iterations.

\begin{algorithm}[t]
    \caption{\texttt{PartialClustering}($(X, \mathrm{dist}), k, R, \alpha$)}\label{alg: partial-clustering}
    \KwData{A metric space $(X, \mathrm{dist})$, integer $k \ge 1$, radius $R \ge 0$, parameter $\alpha \ge 1$}
    \KwResult{A set $C \subseteq X$ of $k$ facilities that contains at least as many points within distance $\alpha R$ as contained by any other set $C' \subseteq X$ of $k$ facilities within distance $R$, i.e.,
        \[
            \left|B(C, \alpha R)\right| \ge \max_{C' \in \binom{X}{k}} \left|B(C', R)\right|.
        \]
    }
\end{algorithm}

\begin{algorithm}[t]
    \caption{\texttt{IterativeClustering}$((X, \mathrm{dist}), k, \epsilon, \alpha)$}\label{alg: geometric-clustering}
    \KwData{A metric space $(X, \mathrm{dist})$ on $n$ points, integer $k \ge 1$, parameter $\epsilon > 0$, parameter $\alpha \ge 1$}
    \KwResult{A set $C \subseteq X$ of $O\left(\frac{k \log n}{\epsilon}\right)$ facilities}
    \DontPrintSemicolon
    $C \gets \emptyset$\;
    $R_0 = \frac{D \epsilon}{n}$, where $D$ is the $k$-center optimum for $(X, \mathrm{dist})$\;
    \For{$l = 0, 1, \dots, \log_{1 + \epsilon}(n/\epsilon)$ \label{step: geometric-clustering-loop}}
    {
        $R \gets R_0 (1 + \epsilon)^l$\;
        $C_l \gets \texttt{PartialClustering}((X, \mathrm{dist}), k, R, \alpha)$\;
        $C \gets C \cup C_l$\;
    }
    \Return{$C$}
\end{algorithm}

For polynomial-time computations, \texttt{PartialClustering} cannot be solved exactly since it generalizes the $k$-center problem. To get efficient algorithms, we allow it to output $k$ facilities that cover as many points within radius $\alpha R$ as those covered by any $k$ facilities within radius $R$. As \cite{kumar_fairness_2000} note, \cite{charikar_algorithms_2001} give an approximation algorithm for \texttt{PartialClustering} for $\alpha = 3$, which we state in a modified form:

\begin{lemma}[Theorem 3.1, \cite{charikar_algorithms_2001}]\label{thm: clustering-partial-problem}
There exists a polynomial-time algorithm that given metric $(X, \mathrm{dist})$, integer $k \ge 1$, and radius $R$, outputs $k$ facilities that cover at least as many points within radius $3R$ as those covered by any set of $k$ facilities within radius $R$. That is, subroutine \textup{\texttt{PartialClustering}} runs in polynomial-time for $\alpha = 3$.
\end{lemma}

We give some notation: given nonempty $F \subseteq X$ and some radius $R \ge 0$, we denote by $B(F; R)$ the set of all points within distance $R$ of $F$, i.e., $B(F; R) = \{x \in X: \exists \: y \in F \text{ with } \mathrm{dist}(x, y) \le R\}$. We say that a set of facilities $F$ \emph{covers} $p$ points within radius $R$ if $|B(F; R)| \ge p$.

Let $D$ denote the $k$-center optimum for $(X, \mathrm{dist})$. By definition, there are $k$ facilities that can cover all of $X$ within radius $D$. Therefore, the largest radius we need to consider is $D$. What is the smallest radius we need to consider? Since all of our objective norms are monotonic and symmetric, points covered within very small radii do not contribute a significant amount to the norm value. Therefore, we can start at a large enough radius, which has been set to $\frac{D\epsilon}{n}$ with some foresight.

We will first prove the following claim:
\begin{claim}\label{thm: clustering-guarantee}
For parameter $\alpha \ge 1$, \textup{\texttt{IterativeClustering}}  gives a simultaneous bicriteria $\left(\alpha (1 + 2\epsilon), O\left(\frac{ \log n}{\epsilon}\right)\right)$-approximation for symmetric monotonic norms.
\end{claim}

\begin{proof}
    We first show that the number of facilities output by the algorithm is $O\left(\frac{k \log n}{\epsilon}\right)$. The number of iterations in the for loop is $\log_{(1 + \epsilon)}\left(\frac{n}{\epsilon}\right) = O\left(\frac{\log n}{\epsilon} + \frac{\log(1/\epsilon)}{\epsilon}\right)$. When $\epsilon > \frac{1}{n}$, this expression is $O\left(\frac{\log n}{\epsilon}\right)$. Since each iteration adds at most $k$ facilities to $C$, we are done in this case. When $\epsilon \le \frac{1}{n}$, then $\frac{k \log n}{\epsilon} \ge n$, that is, all facilities can be opened anyway.

    Fix any symmetric monotonic norm $\|\cdot\|_f$ on $\R^n$, and let $\OPT$ denote the optimal solution for this norm and $\xOPT \in \R^n$ denote the corresponding distance vector. Let the distance vector for facilities $C$ output by the algorithm be $x$. We need to show that $\|x\|_f \le \alpha(1 + 2\epsilon) \|\xOPT\|_f$.

    By definition, $(\xOPT)^\uparrow_1 \le (\xOPT)^\uparrow_2 \le \dots \le (\xOPT)^\uparrow_n$. Let $j^*$ be the smallest index such that $(\xOPT)^\uparrow_{j^*} > R_0 = \frac{D\epsilon}{n}$. Since $\|\cdot\|_f$ is symmetric, we have $\|x^\uparrow\|_f = \|x\|_f$ and $\|(\xOPT)^\uparrow\|_f = \|\xOPT\|_f$. Our twofold strategy is to show that: \begin{enumerate}
                                                                                                                                                                                                                                                                                                                                                              \item for all $j \ge j^*$, \label{todo: clustering-first-condition} \begin{equation}\label{eqn: clustering-pointwise-bound}
                                                                                                                                                                                                                                                                                                                                                              (x)^\uparrow_j \le \alpha (1 + \epsilon) (\xOPT)^\uparrow_j,
                                                                                                                                                                                                                                                                                                                                                              \end{equation}
                                                                                                                                                                                                                                                                                                                                                              \item the contribution of $x^\uparrow_1, \dots, x^\uparrow_{j^* - 1}$ to $\|x\|_f$ is small; specifically, \label{todo: clustering-second-condition}
                                                                                                                                                                                                                                                                                                                                                              \begin{equation}\label{eqn: clustering-bound-on-small-indices}
                                                                                                                                                                                                                                                                                                                                                              \left\|\left(x^\uparrow_1, \dots, x^\uparrow_{j^* - 1}, 0, \dots, 0 \right)\right\|_f \le \alpha \epsilon \|\xOPT\|_f.
                                                                                                                                                                                                                                                                                                                                                              \end{equation}
    \end{enumerate}

    Consider the first part. We have $R_0 (1 + \epsilon)^{\log_{1 + \epsilon} (n/\epsilon)} = R_0 \frac{n}{\epsilon} = D$. That is, in the final iteration of the for loop, $R = D$. Therefore, by definition of $D$ and \texttt{PartialClustering}, $C_l$ in this iteration covers all of $X$ within radius $\alpha D$. That is, $\|x\|_\infty \le \alpha D$ since $C_l \subseteq C$.

    fix some $j \ge j^*$, and let $l \ge 0$ be the smallest integer such that $(\xOPT)^\uparrow_j \le R_0 (1 + \epsilon)^l$. If $l \ge 1 + \log_{1 + \epsilon}(n/\epsilon)$, then $(\xOPT)^\uparrow_j > R_0 (1 + \epsilon)^{l - 1} = D$. Since $\|x\|_\infty \le \alpha D$, inequality (\ref{eqn: clustering-pointwise-bound}) holds in this case.

    Otherwise, $l \le \log_{1 + \epsilon}(n/\epsilon)$. The $k$ facilities in $\OPT$ cover at least $j$ points within radius $R = R_0 (1 + \epsilon)^l$. By definition of \texttt{PartialClustering}, in iteration $l$ of the for loop, $C_l$ covers at least $j$ points within radius $\alpha R$. Since $C_l \subseteq C$, $C$ also covers at least $j$ points within radius $\alpha R$, so that $x^\uparrow_j \le \alpha R = R_0(1 + \epsilon)^l$. By definition of $l$, $(\xOPT)^\uparrow > R_0 (1 + \epsilon)^{l - 1}$, and so
    \begin{equation*}
        x^\uparrow_j \le \alpha R_0(1 + \epsilon)^l \le \alpha (1 + \epsilon) (\xOPT)^\uparrow_j.
    \end{equation*}

    We move to (\ref{eqn: clustering-bound-on-small-indices}). By definition of $j^*$, $\OPT$ covers at least $j^* - 1$ points within radius $R_0$. In iteration $0$, by definition of \texttt{PartialClustering}, $C_0$ (and therefore $C$) covers at least $(j^* - 1)$ points within radius $\alpha R_0$. That is, $x^\uparrow_{j^* - 1} \le \alpha R_0$.

    Denote $(1, 0, \dots, 0) = \mathbf{e}$. Since $\|\cdot\|_f$ is monotonic and $D$ is the $k$ center optimum, $\|\xOPT\|_f \ge \left(\|\xOPT\|_\infty, 0, \dots, 0\right) \|\mathbf{e}\|_f \ge D \|\mathbf{e}\|_f$. Therefore,
    \begin{align*}
        \left\| \left(x^\uparrow_1, \dots, x^\uparrow_{j^* - 1}, 0, \dots, 0 \right) \right\|_f &\le \sum_{j \in [j^* - 1]} x^\uparrow_j \|\mathbf{e}\|_f & (\text{triangle inequality})\\
        &\le \sum_{j \in [j^* - 1]} \alpha R_0 \|\mathbf{e}\|_f & (x^\uparrow_{j^* - 1} \le \alpha R) \\
        &< n \alpha \frac{D\epsilon}{n} \|\mathbf{e}\|_f & (j^* \le n)\\
        &\le \alpha \epsilon \|\xOPT\|_f. & (\|\xOPT\|_f \ge D \|\mathbf{e}\|_f)
    \end{align*}

    Together, inequalities (\ref{eqn: clustering-pointwise-bound}), (\ref{eqn: clustering-bound-on-small-indices}) imply that
    \begin{align*}
        \|x\|_f &\le \left\| \left(x_1^\uparrow, \dots, x^\uparrow_{j^* - 1}, 0, \dots, 0\right) \right\|_f + \left\| \left(0, \dots, 0, x^\uparrow_{j^*}, \dots, x^\uparrow_n\right) \right\|_f & (\text{triangle inequality})\\
        &\le \alpha \epsilon \|\xOPT\|_f + \alpha (1 + \epsilon) \left\| \left(0, \dots, 0, (\xOPT)^\uparrow_{j^*}, \dots, (\xOPT)^\uparrow_n\right) \right\|_f &(\text{inequalities } (\ref{eqn: clustering-pointwise-bound}), (\ref{eqn: clustering-bound-on-small-indices}))\\
        &\le \alpha \epsilon \|\xOPT\|_f + \alpha (1 + \epsilon) \|\xOPT\|_f = \alpha (1 + 2\epsilon) \|\xOPT\|_f. & (\|\cdot\|_f \text{ is symmetric monotonic})
    \end{align*}
\end{proof}

With this result in hand, our main theorem is simple to derive: we choose $\alpha = 1$ in the claim with $\epsilon/2$ as the parameter for the existence result. We choose $\alpha = 3$ in the claim with $\epsilon/6$ as the parameter for the polynomial-time result; Lemma \ref{thm: clustering-partial-problem} guarantees that the algorithm is polynomial-time.

\subsection{\textsc{Uncapacitated-Facility-Location}}\label{sec: facility-location}

First, we note that a single solution cannot be better than $\Omega(\sqrt{n})$-approximate for even the $L_1$ and $L_\infty$ norms: suppose the metric is a star metric with $n$ leaves. The distance from the center to each leaf is $\sqrt{n}$. Then the optimal $L_1$ solution is to open each facility, and the cost of this solution is $n + 1$. The optimal $L_\infty$ solution is to open just one facility at the center, the cost of this solution is $1 + \sqrt{n}$. Now, any solution that opens fewer than $n/2$ facilities has cost $\ge n/2 + (n/2) \sqrt{n} = \Omega(n \sqrt{n})$ for the $L_1$ norm and therefore is an $\Omega(\sqrt{n})$-approximation. Any solution that opens $\ge n/2$ facilities is an $\Omega(\sqrt{n})$-approximation for the $L_\infty$ norm. A similar example was noted for the $k$-clustering variant in \cite{goel_simultaneous_2006}.

This motivates us to seek larger portfolios and get a smaller approximation. The main theorem of this section gives an $O(\log n)$-approximate portfolio of size $O(\log n)$ for \ufl:

\begin{theorem}\label{cor: facility-location}
There exists a polynomial-time algorithm that given any instance of \ufl on $n$ points, outputs an $O(\log n)$-approximate portfolio of size $O(\log n)$ for symmetric monotonic norms.
\end{theorem}

\begin{proof}
    Assume without loss of generality that the number of points $n$ is a power of $2$. Choose solutions corresponding to $k = 2^0, 2^1, 2^2, \dots, 2^{\log_2 n}$ with $\epsilon = 1$ in Theorem \ref{thm: clustering} part 2. There are $O(\log n)$ of these, and the theorem asserts that they can be found in polynomial time. We claim that these form an $O(\log n)$-approximate portfolio for $\smn$.

    Fix a norm $\|\cdot\|_f \in \smn$, and suppose the optimal solution $\OPT$ for this norm opens $k^* \in [n]$ facilities. Let $l$ be the unique integer such that $2^{l - 1} < k^* \le 2^l$, i.e., $l = \lceil \log_2 k^* \rceil$. We show that the solution corresponding to $k = 2^l$ in our portfolio is an $O(\log n)$-approximation for $\|\cdot\|_f$. Add arbitrary $2^l - k^*$ facilities to $\OPT$; this only decreases the induced distance vector $\xOPT$. For this new set of facilities, we have the guarantee from Theorem \ref{thm: clustering} that $\|x\|_f \le 4 \|\xOPT\|_f$. Therefore, the objective value of the portfolio solution is
    \[
        O(\log n) \cdot 2^l + \|x\|_f = O(\log n) \left(k^* +  \|\xOPT\|_f\right) = O(\log n) \cdot \OPT. \qedhere
    \]
\end{proof}

\end{document}